\documentclass[11pt]{article}
\usepackage{fullpage,graphicx,amsmath,amssymb,verbatim,bm,xcolor}
\allowdisplaybreaks
\newtheorem{definition}{Definition}
\newtheorem{lemma}{Lemma}
\newtheorem{proposition}{Proposition}
\newtheorem{fact}{Fact}
\newtheorem{theorem}{Theorem}
\newtheorem{corollary}{Corollary}

\newenvironment{proof}{\textbf{Proof:}}{\hfill$\square$}

\newcommand{\cA}{\mathcal{A}}
\newcommand{\cB}{\mathcal{B}}
\newcommand{\cC}{\mathcal{C}}
\newcommand{\cD}{\mathcal{D}}
\newcommand{\cE}{\mathcal{E}}
\newcommand{\cF}{\mathcal{F}}
\newcommand{\cN}{\mathcal{N}}

\newcommand{\cQ}{\mathcal{Q}}
\newcommand{\cT}{\mathcal{T}}
\newcommand{\cX}{\mathcal{X}}
\newcommand{\cY}{\mathcal{Y}}

\newcommand{\CMG}{\mathrm{CMG}}

\newcommand{\C}{\mathbb{C}}

\newcommand{\hA}{\hat{A}}

\newcommand{\hM}{\hat{M}}

\newcommand{\hX}{\hat{X}}
\newcommand{\hY}{\hat{Y}}

\newcommand{\hcT}{\hat{\mathcal{T}}}
\newcommand{\hhcT}{\hat{\hat{\mathcal{T}}}}

\newcommand{\hPi}{\hat{\Pi}}
\newcommand{\hrho}{\hat{\rho}}
\newcommand{\hhrho}{\hat{\hat{\rho}}}
\newcommand{\hsigma}{\hat{\sigma}}
\newcommand{\hhsigma}{\hat{\hat{\sigma}}}
\newcommand{\htau}{\hat{\tau}}
\newcommand{\hhtau}{\hat{\hat{\tau}}}

\newcommand{\tsigma}{\tilde{\sigma}}

\newcommand{\Hmin}{H_{\mathrm{min}}}

\newcommand{\supp}{\mathrm{supp}}

\newcommand{\Good}{\mathsf{Good}}

\DeclareMathOperator*{\E}{{\rm {\mathbb E}}\,}
\DeclareMathOperator*{\Tr}{{\rm Tr}\;}
\DeclareMathOperator*{\prob}{{\rm Pr}\;}
\DeclareMathOperator*{\spanning}{{\rm span}\;}

\newcommand{\zero}{\leavevmode\hbox{\small l\kern-3.5pt\normalsize0}}
\newcommand{\one}{\leavevmode\hbox{\small1\kern-3.8pt\normalsize1}}
\newcommand{\I}{\mathbb{I}}

\newcommand{\ket}[1]{| #1 \rangle}
\newcommand{\bra}[1]{\langle #1 |}
\newcommand{\ketbra}[1]{\ket{#1}\bra{#1}}
\newcommand{\braket}[2]{\langle {#1} \ket{#2}}

\begin{document}

\title{{\bf Fully smooth one shot multipartite soft covering 
of quantum states without pairwise independence
}}

\author{Pranab Sen\footnote{
School of Technology and Computer Science, Tata Institute of Fundamental
Research, Mumbai, India. Part of this work was done 
while the author was on sabbatical leave
at the Centre for Quantum Technologies, National University of Singapore.
Email: {\sf pranab.sen.73@gmail.com}.
}
}

\date{}

\maketitle

\begin{abstract}
We provide a powerful machinery to prove fully smooth one shot 
multipartite covering, aka convex split, type results for quantum states.
In the important case of smooth multipartite convex split for classical 
quantum states, aka smooth multipartite soft covering, our 
machinery works even 
when certain marginals of these states do not satisfy pairwise 
independence. The paper \cite{Sen:telescoping} gave the first proof
of fully smooth multipartite convex split by
simplifying and extending a technique called telescoping, developed
originally for convex split by \cite{Cheng:convexsplit}.
However, that work as well as 
all earlier works on convex split
assumed pairwise or even more independence amongst suitable
marginals of the quantum states.

We develop our machinery by leveraging known results from 
\cite{sen:oneshot} involving tilting and augmentation smoothing
of quantum states, combined with a novel observation that a natural
quantum operation `flattening' quantum states actually preserves the 
fidelity. This machinery is powerful enough to lead to non pairwise
independent results as mentioned above.

As an application of our soft covering lemma without pairwise
independence, we prove the `natural' one shot
inner bounds for sending private classical information over a 
quantum wiretap interference channel, even when the classical encoders
at the input lose pairwise independence in their encoding strategies 
to a certain extent. 
This result was unknown earlier even in the classical setting.
\end{abstract}

%%manuscript information
%%received, revised and accepted dates
%%
%\msinfo{19 August 2018}{19 August 2018}{19 August 2018}

%%insert keywords separated by comma
%\keywords{quantum simultaneous decoder; joint typicality; one-shot 
%inner bounds; multiple access channel; network information theory.}

%\setcounter{page}{1001}
%\corres
%\volnum{123}
%\issuenum{4}
%\monthyear{August 2018}
%\pgfirst{1001}
%\pglast{1038}
%\doinum{12.3456/s78910-011-012-3}
%\articleType{}
%%use \articleType{L} for Sadhana Letters, empty/null otherwise

%%The running head information

%\markboth{Pranab Sen}{Unions, intersections and a 
%one-shot quantum joint typicality lemma}

\section{Introduction}
The simultaneous smoothing bottleneck is a famous open problem in 
network quantum
information theory \cite{drescher:simultaneous}. A positive result for 
this problem would imply 
that many union and intersection type arguments carried out, sometimes
implicitly, in network classical information theory extend similarly
to the quantum setting. Most tasks in information theory can be decomposed
into simpler tasks of one of two types: {\em packing} and 
{\em covering}. The earlier work of \cite{sen:oneshot} gave a machinery
for implementing union and intersection for packing tasks in one shot 
network quantum information theory, bypassing simulaneous smoothing.
However it left open the question of implementing union and intersection
for covering tasks, a lacuna which was explicitly pointed out in
\cite{ding:relay}.

Recent exciting works of Cheng, Gao and 
Berta~\cite{Cheng:convexsplit}, and Colomer and 
Winter~\cite{Colomer:decoupling}, have introduced a telescoping cum
mean-zero decomposition technique that bypasses the simultaneous
smoothing bottleneck for intersection arguments for two fundamental 
problems in quantum information theory viz. multipartite covering aka
multipartite convex split, and multipartite decoupling respectively.
However Cheng et al. did not state their
multipartite convex split
results in terms of smooth one shot quantities, leaving several
basic one shot and finite blocklength achievability questions 
in network quantum information theory, e.g. inner bounds for the 
generalised one shot quantum Slepian Wolf \cite{anshu:slepianwolf} problem,
unanswered. Very recently, \cite{Sen:telescoping} showed that
the telescoping technique can in fact be simplified and further extended
in order to prove fully smooth multipartite convex split and decoupling
results for the first time. That work also 
provided several applications of the fully smooth convex 
split and decoupling theorems to important problems in network
quantum information theory.

All the works referenced above as well as all earlier 
works on convex split, whether smooth or non 
smooth, assumed at least pairwise, often full independence
or tensor product structure of certain marginals of quantum states. 
However, all the earlier proof techniques including telescoping
fail when the pairwise independence 
requirement is dropped. Dropping pairwise independence becomes important
for classical quantum soft covering lemmas.
In multipartite soft covering, there are $k$
classical parties, $k \geq 1$. Each party independently samples a set 
of classical 
symbols from a probability distribution.  A $k$-tuple of
symbols from the samples of the $k$ parties is fed to a box whose 
output is a quantum state
depending on the input $k$-tuple. The uniform average of the 
resulting input dependent quantum state is taken
over all $k$ tuples that can be constructed from the $k$ sets sampled
by the $k$ parties. The aim is to ensure that the resulting
sample averaged input dependent quantum state is close in 
trace distance to an `ideal'
fixed quantum state, in expectation over the random choices of the
$k$ parties. Dropping pairwise independence for soft covering means
that a classical party chooses its set of symbols in a pairwise
dependent fashion.

Soft covering lemmas
are used e.g. to prove the privacy requirement for classical messages
for various types of wiretap channels 
\cite{Radhakrishnan:wiretap,Wilde:wiretap}.
Relaxing pairwise independence
in such settings can make computational sense as it will require
senders to use less true random bits, an expensive resource,
in their codebook construction step.

In this paper, we provide a new proof for fully smooth multipartite
convex split of quantum states. Our proof does not use telescoping;
instead it uses a completely different machinery. The advantage
of using this more sophisticated machinery is that in the 
soft covering setting, we can show
for the first time that our fully
smooth inner bounds continue to hold
even when the classical covering parties lose pairwise independence
in a restricted way while choosing their set of samples. We believe 
this is a major
conceptual contribution of our work. 

In order to prove such powerful
and general multipartite fully quantum convex split and soft classical
quantum covering cresults, we develop new
machinery by leveraging known results from 
\cite{sen:oneshot} involving tilting and augmentation smoothing
of quantum states, combined with a novel observation that a natural
quantum operation `flattening' quantum states actually preserves the 
fidelity.
As an application of the power of this machinery, we prove the 
achievability
of a natural rate region for sending private classical information
over a wiretap quantum interference channel extending Chong, Motani,
Garg and ElGamal inner bound \cite{CMGElGamal} for the non-wiretap 
classical interference
channel. Our scheme continues to work even when the classical encodings
at the inputs to the channel lose pairwise independence to a certain
extent. This was not known before, not even in the classical setting.

\subsection{Organisation of the paper}
Section~\ref{sec:CMGcovering} serves as a warmup, where we first
recall the CMG covering problem introduced in \cite{Sen:telescoping}.
That work provided a smooth inner bound for CMG covering using
telescoping. In Section~\ref{sec:CMGcovering},
we will see how the telescoping proof of smooth CMG covering fails
if the pairwise independence
assumption between certain marginals of classical quantum states involved 
in the telescoping proof is removed. 
Some novel results about flattening quantum states, 
Schatten-$\ell_2$ norm and inner product of matrices, and tilting and 
augmentation smoothing of quantum states 
are stated and proved in Section~\ref{sec:prelims}.
They will be required to prove fully
smooth multipartite convex split in Section~\ref{sec:convexsplitflatten}.
The style of proof of convex split in Section~\ref{sec:convexsplitflatten}
serves as a warmup to the proof of the smooth multipartite soft covering
lemma without pairwise independence in 
Section~\ref{sec:coveringnonpairwise}. 
Section~\ref{sec:coveringnonpairwise} also proves a smooth
inner bound for the CMG covering problem that works even when
certain pairwise independence assumptions are relaxed, addressing
the drawback of telescoping pointed out earlier in 
Section~\ref{sec:CMGcovering}.
We use this inner bound for CMG covering in 
Section~\ref{sec:wiretapinterference}, where
we finally
obtain smooth inner bounds for sending private classical information
over a quantum wiretap interference channel, even when certain
pairwise independence requirements at the classical encoders are relaxed.
Section~\ref{sec:newconvexsplit} contains an alternate simpler 
proof of fully smooth multipartite convex split using our flattening 
technique. 
It serves as a warmup to Section~\ref{sec:decoupling}, where we
prove a fully smooth multipartite decoupling theorem using similar
methods. The fully smooth multipartite decoupling theorem of
Section~\ref{sec:decoupling} uses
one extra ancilla qubit per sender unlike the fully smooth multipartite
decoupling theorem in \cite{Sen:telescoping}; however this extra
qubit per sender does not affect any application of decoupling to the
best of our knowledge.
We finally make some concluding remarks and list directions for
further research
in Section~\ref{sec:conclusion}.

\section{Warmup: Telescoping fails without pairwise
independence}
\label{sec:CMGcovering}
The CMG covering problem defined in \cite{Sen:telescoping} lies 
at the heart of obtaining
inner bounds for private classical communication over a quantum 
wiretap interference channel. The name CMG is an
acronym for Chong, Motani and Garg, the discoverers of one of the most
well known inner bounds for the classical non-wiretap interference channel
in the asympotic independent and identically distributed (asymptotic iid)
setting \cite{CMGElGamal}. The paper \cite{sen:interference} first 
showed the achievability of
the same inner bound for sending classical information over a quantum
non-wiretap interference channel in the asymptotic iid setting. 
Subsequently, the achievability in the non-wiretap quantum case in 
the general one shot setting was proved in \cite{sen:simultaneous}.
That paper used the machinery of \cite{sen:oneshot} to bypass 
simultaneous smoothing for intersection problems arising in packing
tasks. Later on in this paper, we will be able to extend the result of
\cite{sen:simultaneous} to the one shot wiretap quantum case by 
combining the packing arguments with our fully smooth CMG covering
lemma without pairwise independence 
(Lemma~\ref{lem:CMGcoveringnonpairwise} below.

Let $\cN$ be a quantum channel with two inputs called Alice and Bob,
and one output called Eve. The intuition here is that Eve is an 
eavesdropper trying to get
information on the classical messages that Alice and Bob are sending 
to some other outputs of 
$\cN$, not described here explicitly as they are  not
relevant for CMG covering. Though $\cN$ has quantum inputs and quantum
outputs, we will assume without loss of generality that the inputs are
classical, because in the inner bound below, there is a standard 
optimisation
step over a choice of all ensembles of input states to Alice and Bob.
Thus, we will henceforth think of $\cN$ as a classical-quantum (cq)
channel with its two classical input alphabets being denoted by
$\cX$, $\cY$, and the quantum output alphabet by $\cE$. 

Let $0 \leq \epsilon \leq 1$. 
On input $(x,y) \in \cX \times \cY$, the channel outputs
a quantum state 
\[
\sigma^E_{xy} := \cN^{A B \rightarrow E}(\rho^A_x \otimes \rho^B_y)
\]
on $E$, where $\{\rho^A_x\}_x$, $\{\rho^B_y\}_y$ are the ensembles of the
so-called encoding quantum states at the quantum inputs $A$, $B$ 
of the channel which is
modelled by a completely positive trace non-increasing superoperator
$\cN: A B \rightarrow E$. 
To define the CMG covering problem, we need to define new alphabets
$\cX'$, $\cY'$, following the scheme of
the original paper \cite{CMGElGamal}. 
We then pick a `control' probability distribution 
\[
p(x',x,y',y) := p(x', x) p(y', y)
\]
on the classical alphabet 
$\cX' \times \cX \times \cY' \times \cY$.
The `control' cq state for the CMG covering problem is now defined as
\[
\sigma^{X' X Y' Y E} := 
\sum_{x',x,y',y} p(x',x,y',y) \ketbra{x',x,y',y}^{X'XY'Y}
     \otimes \sigma^E_{xy}.
\]
Though the state $\sigma^E_{xy}$ depends only on $x$ and $y$, for
later convenience we will write it as
$\sigma^E_{x'xy'y} := \sigma^E_{xy}$. Though this notation seems 
heavier now, it will be extremely useful when we take various marginals
of $\sigma^{X' X Y' Y E}$ in the entropic quantities and the proofs
later on. Thus,
\[
\sigma^{X' X Y' Y E} := 
\sum_{x',x,y',y} p(x',x)p(y',y) \ketbra{x',x,y',y}^{X'XY'Y}
     \otimes \sigma^E_{x'xy'y}.
\]

Let us say Alice and Bob are trying to send a pair of classical messages
$a$ and $b$ respectively to Charlie. To obfuscate them from Eve, Alice
and Bob independently do the following strategy. 
Alice chooses iid samples $x'_1, \ldots, x'_{L'}$ from the
marginal distribution $p^{X'}$.
Conditioned on each sample $x'_{l'} \in \cX'$, Alice chooses iid
samples $x_{l', 1}, \ldots, x_{l', L}$ from the conditioned 
marginal distribution $p^{X} | (X' = x'_{l'}$. Thus, a total of
$L' L$ samples are chosen by Alice.
Similarly, Bob chooses iid samples $y'_1, \ldots, y'_{M'}$ from the
marginal distribution $p^{Y'}$.
Conditioned on each sample $y'_{m'} \in \cY'$, Bob chooses iid
samples $y_{m', 1}, \ldots, y_{m', M}$ from the conditioned 
marginal distribution $p^{Y} | (Y' = y'_{m'}$. Thus, a total of
$M' M$ samples are chosen by Bob. 
Alice inputs a uniformly random sample from her chosen set into the
channel $\cN$. 
Similarly, Bob inputs a uniformly random sample from his chosen set 
into the channel $\cN$. The hope is that this strategy obfuscates
Eve's received state so much that she is unable to figure out 
the actual input message pair $(a,b)$.

Let $\sigma^E$ denote the marginal on $E$ of the control state
$\sigma^{X' X Y' Y E}$. Define
\begin{equation}
\label{eq:CMGdef1}
\sigma^E_{\vec{x'},  \vec{x}, \vec{y'}, \vec{y}} :=
(L' (L+1) M' (M+1))^{-1}
\sum_{a'=1}^{L'} \sum_{a=1}^L
\sum_{b'=1}^{M'} \sum_{b=1}^M
\sigma^E_{x'_{a'} x_{a'a} y'_{b'} y_{b'b}}.
\end{equation}
We would like to show that the  so-called CMG convex split quantity 
defined below is small, when
$L'$, $L$, $M'$, $M$ are suitably large.
\begin{equation}
\label{eq:CMGdef2}
\CMG :=
\E_{\vec{x'}, \vec{x}, \vec{y'}, \vec{y}}[
\|\sigma^E_{\vec{x'},  \vec{x}, \vec{y'}, \vec{y}} - \sigma^E\|_1
],
\end{equation}
the expectation being over the choice of the samples
$\vec{x'} := (x'_1, \ldots, x'_{L'})$,
$\vec{x} := (x_{1,1}, \ldots, x'_{L',L})$,
$\vec{y'} := (y'_1, \ldots, y'_{M'})$,
$\vec{y} := (y_{1,1}, \ldots, y'_{M',M})$ according to the procedure
described above. Recall that for $p > 0$, the 
{\em Schatten-$\ell_p$ norm} of a matrix $M$ is defined as 
$
\|M\|_p := \Tr[(A^\dag A)^{p/2}].
$
The {\em Schatten-$\ell_\infty$ norm}, aka operator norm induced from the
Hilbert space norm, is defined by taking $p \rightarrow +\infty$ in the
above expression, resulting in
$\|M\|_\infty$ being the largest singular value of $M$. For normal
matrices $M$, it equals the largest absolute value of an 
eigenvalue of $M$.

The paper \cite{Sen:telescoping} gave a fully smooth inner bound for
the CMG covering problem using telescoping. The telescoping based proof
in that paper uses expressions like
\begin{equation}
\label{eq:telescopingexample}
\begin{array}{rcl}
\lefteqn{
\E_{x', x, \hat{x}, y', y}\left[
\Tr\left[
(\tsigma^E_{x' x y' y}(8) - 
\tsigma^E_{x' y' y}(8) -
\tsigma^E_{x' x y'}(8) + 
\tsigma^E_{x' y'}(8)) 
\right.
\right. 
} \\
& &
~~~~~~~~~~~~~~~~~~~~~
\left.
\left.
(\tsigma^E_{x' \hat{x} y' y}(8) - 
\tsigma^E_{x' y' y}(8) -
\tsigma^E_{x' \hat{x} y'}(8) + 
\tsigma^E_{x' y'}(8))
\right]
\right],
\end{array}
\end{equation}
in order to bound the CMG quantity in Equation~\ref{eq:CMGdef2} above.
Above, a matrix like $\tsigma^E_{x' x y' y}(8)$ is a perturbed version
of the matrix 
\[
\tsigma^E_{x'x y' y} := 
(\sigma^E)^{-1/2} \circ \sigma^E_{x' x y' y} :=
(\sigma^E)^{-1/2} \sigma^E_{x' x y' y} (\sigma^E)^{-1/2},
\]
where $\sigma^{-1}$ is the so-called Moore-Penrose pseudoinverse of 
$\sigma$ i.e. the inverse on the support of $\sigma$ and the zero
operator orthogonal to the support.
Eight different perturbed versions of $\tsigma^E_{x'x y' y}$ called
$\tsigma^E_{x'x y' y}(i)$, $i \in [8]$, are used 
at different places of the telescoping proof
\cite{Sen:telescoping}
in order to prove fully smooth inner bounds for
CMG covering. 

The telescoping technique crucially requires that a telescoping term
satisfy
the so-called {\em mean zero property} \cite{Sen:telescoping}.
For a term like in Equation~\ref{eq:telescopingexample}
above, various conditions need to be met in order for the mean
zero requirement to be satisfied. One of those conditions is as follows:
\begin{equation}
\label{eq:meanzero}
\begin{array}{rcl}
\lefteqn{
\E_{x', x, \hat{x}, y', y}\left[
\Tr\left[
(\tsigma^E_{x' x y' y}(8) - 
\tsigma^E_{x' y' y}(8) -
\tsigma^E_{x' x y'}(8) + 
\tsigma^E_{x' y'}(8)) 
\right.
\right. 
} \\
&  &
~~~~~~~~~~~~~~~~~~
\left.
\left.
(\tsigma^E_{x' \hat{x} y' y}(8) - 
\tsigma^E_{x' y' y}(8) -
\tsigma^E_{x' \hat{x} y'}(8) + 
\tsigma^E_{x' y'}(8))
\right]
\right] \\
& = &
\E_{x', y', y}\left[
\Tr\left[
\left(
\E_{x|x'}[
\tsigma^E_{x' x y' y}(8) - 
\tsigma^E_{x' y' y}(8) -
\tsigma^E_{x' x y'}(8) + 
\tsigma^E_{x' y'}(8)] 
\right) 
\right.
\right.
\\
&  &
~~~~~~~~~~~~~~~~~~~~~
\left.
\left.
\left(
\E_{\hat{x}|x'}[
\tsigma^E_{x' \hat{x} y' y}(8) - 
\tsigma^E_{x' y' y}(8) -
\tsigma^E_{x' \hat{x} y'}(8) + 
\tsigma^E_{x' y'}(8)]
\right)
\right]
\right] \\
& = &
\E_{x', y', y}\left[
\Tr\left[
(\tsigma^E_{x' y' y}(8) - 
\tsigma^E_{x' y' y}(8) -
\tsigma^E_{x' y'}(8) + 
\tsigma^E_{x' y'}(8))
\right.
\right. \\
& &
~~~~~~~~~~~~~~~~~~~
\left.
\left.
(\tsigma^E_{x' y' y}(8) - 
\tsigma^E_{x' y' y}(8) -
\tsigma^E_{x' y'}(8) + 
\tsigma^E_{x' y'}(8))
\right]
\right] \\
& = &
0.
\end{array}
\end{equation}
The equalities in the above equation only hold when in the obfuscating 
strategy of Alice described earlier,
the distribution of $x_{a' a}$ conditioned
on $x'_{a'}$ is pairwise independent from the distribution of
$x_{a' \hat{a}}$ conditioned on $x'_{a'}$ for any pair $a \neq \hat{a}$.

We can now appreciate the minimum requirement of pairwise independence
in the proofs of both non-smooth as well as smooth covering
lemmas. 
Supppose the pairwise independence condition in Alice's obfuscating
strategy is `slightly broken'. The mean zero
property in Equation~\ref{eq:meanzero} fails.
By {\em slightly broken} we mean that conditioned on a particular
choice of $x'_{a'}$, $x_{a' a}$, $y'_{b'}$, $y_{b'b}$,
\[
\E_{x_{a' \hat{a}}|(x'_{a'} x_{a'a} y'_{b'} y_{b'b})}[
\sigma^E_{x'_{a'} x_{a' \hat{a}} y'_{b} y_{b'b}}] =:
\dot{\sigma}^E_{x'_{a'} x_{a'a} y'_{b} y_{b'b}} \neq
\sigma^E_{x'_{a'} y'_{b} y_{b'b}}, ~~
\dot{\sigma}^E_{x'_{a'} x_{a'a} y'_{b} y_{b'b}} \leq
\epsilon^{-1} \sigma^E_{x'_{a'} y'_{b} y_{b'b}}.
\]
We require that the expectation above depend only on the actual symbols
$x'_{a'}$, $x_{a' a}$, $y'_{b'}$, $y_{b'b}$ and not on their locations
$a'$, $a$, $\hat{a}$, $b'$ or $b$.

This creates immense problems for telescoping. Suppose for simplicity
that only the 
pairwise independence condition involving $X'$ and $X$ is 
slightly broken as described above i.e. conditioned on $x'_{a'}$,
the distribution of $x_{a'a}$ is not independent from the distribution
of $x_{a'\hat{a}}$ for $a \neq \hat{a}$. Other pairs like
$(x'_{a'}, x_{a'a})$ and $(x'_{\hat{a'}}, x_{\hat{a'} \hat{a}})$ 
for $a' \neq \hat{a'}$, or $y_{b'b}$ and $y_{b' \hat{b}}$ conditioned
on $y'_{b'}$ for $b \neq \hat{b}$ etc. continue
to be independent. 
Now writing a term like
\[
(\tsigma^E_{x'_{a'} x_{a'a} y'_{b'} y_{b'b}}(8) - 
\tsigma^E_{x'_{a'} y'_{b'} y_{b'b}}(8) -
\tsigma^E_{x'_{a'} x_{a'a} y'_{b'}}(8) + 
\tsigma^E_{x'_{a'} y'_{b'}}(8)) 
\]
in the telescoping inequalities, requires for example the following
condition (and several other conditions also) to be met in order 
to satisfy the mean zero property. The following condition occurs
in the telescoping proof of \cite{Sen:telescoping} when that proof
considers the mean zero requirement for the case
$a' = \hat{a'}$, $a \neq \hat{a}$,
$b' = \hat{b'}$, $b = \hat{b}$,
\begin{eqnarray*}
\lefteqn{
\E_{x', x, \hat{x}, y', y}\left[
\Tr\left[
(\tsigma^E_{x' x y' y}(8) - 
\tsigma^E_{x' y' y}(8) -
\tsigma^E_{x' x y'}(8) + 
\tsigma^E_{x' y'}(8)) 
(\tsigma^E_{x' \hat{x} y' y}(8) - 
\tsigma^E_{x' y' y}(8) -
\tsigma^E_{x' \hat{x} y'}(8) + 
\tsigma^E_{x' y'}(8))
\right]
\right]
} \\
& = &
\E_{x', x, y', y}\left[
\Tr\left[
\left(
\tsigma^E_{x' x y' y}(8) - 
\tsigma^E_{x' y' y}(8) -
\tsigma^E_{x' x y'}(8) + 
\tsigma^E_{x' y'}(8) 
\right) 
\right.
\right.
\\
&  &
~~~~~~~~~~~~~~~
\left.
\left.
\left(
\E_{\hat{x}|(x' x y' y)}[
\tsigma^E_{x' \hat{x} y' y}(8) - 
\tsigma^E_{x' y' y}(8) -
\tsigma^E_{x' \hat{x} y'}(8) + 
\tsigma^E_{x' y'}(8)]
\right)
\right]
\right] \\
& = &
\E_{x', x, y', y}\left[
\Tr\left[
(\tsigma^E_{x' x y' y}(8) - 
\tsigma^E_{x' y' y}(8) -
\tsigma^E_{x' x y'}(8) + 
\tsigma^E_{x' y'}(8)) 
(\dot{\tsigma}^E_{x' x y' y}(8) - 
\tsigma^E_{x' y' y}(8) -
\dot{\tsigma}^E_{x' x y'}(8) + 
\tsigma^E_{x' y'}(8))
\right]
\right] \\
& \neq &
0,
\end{eqnarray*}
which destroys the mean zero property.
On the other hand, writing a term like
\[
(\tsigma^E_{x'_{a'} x_{a'a} y'_{b'} y_{b'b}}(8) - 
\dot{\tsigma}^E_{x'_{a'} y'_{b'} y_{b'b}}(8) -
\tsigma^E_{x'_{a'} x_{a'a} y'_{b'}}(8) + 
\dot{\tsigma}^E_{x'_{a'} y'_{b'}}(8)) 
\]
in the telescoping inequalities gives an expression like
the following to meet the mean zero requirement in the proof of
\cite{Sen:telescoping}, when $a' = \hat{a'}$, $a \neq \hat{a}$,
$b' = \hat{b'}$, $b = \hat{b}$,
\begin{eqnarray*}
\lefteqn{
\E_{x', x, \hat{x}, y', y}\left[
\Tr\left[
(\tsigma^E_{x' x y' y}(8) - 
\dot{\tsigma}^E_{x' y' y}(8) -
\tsigma^E_{x' x y'}(8) + 
\dot{\tsigma}^E_{x' y'}(8)) 
(\tsigma^E_{x' \hat{x} y' y}(8) - 
\dot{\tsigma}^E_{x' y' y}(8) -
\tsigma^E_{x' \hat{x} y'}(8) + 
\dot{\tsigma}^E_{x' y'}(8))
\right]
\right]
} \\
& = &
\E_{x', x, y', y}\left[
\Tr\left[
\left(
\tsigma^E_{x' x y' y}(8) - 
\dot{\tsigma}^E_{x' y' y}(8) -
\tsigma^E_{x' x y'}(8) + 
\dot{\tsigma}^E_{x' y'}(8) 
\right) 
\right.
\right.
\\
&  &
~~~~~~~~~~~~~~~
\left.
\left.
\left(
\E_{\hat{x}|(x' x y' y)}[
\tsigma^E_{x' \hat{x} y' y}(8) - 
\dot{\tsigma}^E_{x' y' y}(8) -
\tsigma^E_{x' \hat{x} y'}(8) + 
\dot{\tsigma}^E_{x' y'}(8)]
\right)
\right]
\right] \\
& = &
\E_{x', x, y', y}\left[
\Tr\left[
(\tsigma^E_{x' x y' y}(8) - 
\dot{\tsigma}^E_{x' y' y}(8) -
\tsigma^E_{x' x y'}(8) + 
\dot{\tsigma}^E_{x' y'}(8)) 
(\dot{\tsigma}^E_{x' x y' y}(8) - 
\dot{\tsigma}^E_{x' y' y}(8) -
\dot{\tsigma}^E_{x' x y'}(8) + 
\dot{\tsigma}^E_{x' y'}(8))
\right]
\right] \\
& \neq &
0,
\end{eqnarray*}
which again destroys the mean zero property. In fact, there is no
choice of term one could write down for a successful telescoping strategy 
when pairwise 
independence is slightly broken.

The proof of the non-smooth inner bound for CMG covering in 
\cite{Sen:telescoping}
handles the loss of pairwise independence 
without much trouble giving rise to an inequality like
\begin{equation}
\label{eq:nonsmoothinequality}
\log L' + \log M' + \log M >
D_\infty(
\dot{\sigma}^{X' X Y' Y E} \| 
\sigma^{X' X Y' Y} \otimes \sigma^E)
+ \log \epsilon^{-2},
\end{equation}
in the inner bound. 
Above, we have used the 
{\em non-smooth R\'{e}nyi-$\infty$ divergence}, aka non-smooth max
divergence. This and the {\em non-smooth R\'{e}nyi-$2$ divergence} are
defined below.
\begin{definition}
\[
\begin{array}{c}
I_2(X':E)_\sigma := 
D_2(\sigma^{X'E} \| \sigma^{X'} \otimes \sigma^E), \\
I_\infty(X':E)_\sigma := 
D_\infty(\sigma^{X'E} \| \sigma^{X'} \otimes \sigma^E), \\
D_2(\alpha \| \beta) := 
2 \log\|\beta^{-1/4} \alpha \beta^{-1/4}\|_2, ~~
\mbox{if $\supp(\alpha) \leq \supp(\beta)$, $+\infty$ otherwise}, \\
D_\infty(\alpha \| \beta) := 
\log\|\beta^{-1/2} \alpha \beta^{-1/2}\|_\infty, ~~
\mbox{if $\supp(\alpha) \leq \supp(\beta)$, $+\infty$ otherwise}.
\end{array}
\]
\end{definition}
where $\alpha$ is a {\em subnormalised density matrix} and 
$\beta$ is a positive
semidefinite matrix acting on the same Hilbert space. 
The quantity $I_\infty(X' : E)_\sigma$ is known as
{\em non-smooth R\'{e}nyi-$\infty$ mutual information}, aka
non-smooth max mutual information under the joint state $\sigma^{X'E}$. 
The quantity $I_2(X':E)_\sigma$ is known as the
{\em non-smooth R\'{e}nyi-2 mutual information}
under the joint state $\sigma^{X'E}$. 
In this paper,
a subnormalised density matrix means a positive semidefinite matrix with
trace at most one;  a {\em normalised density matrix} means the trace is
exactly one. Note that 
$D_2(\alpha \| \beta) \leq D_\infty(\alpha \| \beta)$, and that
for a cq state $\sigma^{X'E}$, 
\[
I_2(X':E)_\sigma =
\log \E_{x'} \Tr[((\sigma^E)^{-1/4} \sigma^E_{x'} (\sigma^E)^{-1/4})^2] =
\log \E_{x'} \Tr[(\tsigma^E_{x'})^2],
\]
the expecation over $x'$ being taken according to the classical probability
distribution $\sigma^{X'}$.
For later use, we remark that for a cq state $\sigma^{X'E}$,
\[
I_\infty(X':E)_\sigma =
\max_{x'} D_\infty(\sigma^E_{x'} \| \sigma^E).
\]

The rate inequality in Equation~\ref{eq:nonsmoothinequality} above
arises from a term like the one below, when the proof of the non-smooth
inner bound for CMG covering in \cite{Sen:telescoping} considers a
situation where
$a' = \hat{a'}$, $b' = \hat{b'}$, $b = \hat{b}$ and
$a \neq \hat{a}$:
\begin{eqnarray*}
\lefteqn{
\E_{\vec{x'}, \vec{x}, \vec{y'}, \vec{y}}\left[
\Tr[
\tsigma^E_{x'_{a'} x_{a'a} y'_{b'} y_{b'b}}
\tsigma^E_{x'_{\hat{a'}} x_{\hat{a'}\hat{a}} y'_{\hat{b'}} 
	   y_{\hat{b'}\hat{b}}}
]
\right]
} \\
& = &
\E_{x' x \hat{x} y' y}[
\Tr[\tsigma^E_{x' x y' y} \tsigma^E_{x' \hat{x} y' y}]] 
\;=\;
\E_{x' x \hat{x} y' y}[
\Tr[
\sigma^E_{x' x y' y} (\sigma^E)^{-1/2}
\sigma^E_{x' \hat{x} y' y} (\sigma^E)^{-1/2}]] \\
& = &
\E_{x' x y' y}\left[
\Tr\left[
\sigma^E_{x' x y' y} (\sigma^E)^{-1/2}
\E_{\hat{x} | (x' x y' y)}[\sigma^E_{x' \hat{x} y' y}]
(\sigma^E)^{-1/2}
\right]
\right] \\
& = &
\E_{x' x y' y}[
\Tr[
\sigma^E_{x' x y' y} (\sigma^E)^{-1/2}
\dot{\sigma}^E_{x' x y' y} (\sigma^E)^{-1/2}]] \\
& \leq &
2^{D_\infty(\dot{\sigma}^{X' X Y' Y E} \| 
            \sigma^{X' X Y' Y} \otimes \sigma^E)} \cdot
\E_{x' x y' y}[
\Tr[
\sigma^E_{x' x y' y} (\sigma^E)^{-1/2}
\sigma^E (\sigma^E)^{-1/2}]] \\
& = &
2^{D_\infty(\dot{\sigma}^{X' X Y' Y E} \| 
            \sigma^{X' X Y' Y} \otimes \sigma^E)} \cdot
\E_{x' x y' y}[\Tr[\sigma^E_{x' x y' y}]] 
\;=\;
2^{D_\infty(\dot{\sigma}^{X' X Y' Y E} \| 
            \sigma^{X' X Y' Y} \otimes \sigma^E)},
\end{eqnarray*}
where the cq state $\dot{\sigma}^{X' X Y' Y E}$ is defined in the
natural fashion as
\[
\dot{\sigma}^{X' Y' Y E} :=
\sum_{x' x y' y} p(x',x) p(y',y) \ketbra{x',x,y',y}^{X'XY'Y}
\otimes \dot{\sigma}^E_{x'xy'y}.
\]

Observe that if 
\[
\dot{\sigma}^E_{x'xy'y} = 
\sigma^E_{x'y'y},
\]
which is
what we have in the pairwise independent case, 
then,
$
D_\infty(
\dot{\sigma}^{X' X Y' Y E} \| 
\sigma^{X' X Y' Y} \otimes \sigma^E) =
D_\infty(
\sigma^{X' Y' Y E} \| 
\sigma^{X' Y' Y} \otimes \sigma^E).
$
Our  non-smooth rate region then becomes exactly like the non-smooth
rate region for CMG covering in \cite{Sen:telescoping} in the 
pairwise independent case,
except for $D_2(\cdot\|\cdot)$ in \cite{Sen:telescoping}'s rate region
being replaced by $D_\infty(\cdot\|\cdot)$.
Now suppose 
\[
\forall x:
\dot{\sigma}^E_{x'xy'y} < \epsilon^{-1} \sigma^E_{x'y'y},
\]
i.e. the `consequences' of pairwise independence hold up to
a `scale factor' of $\epsilon^{-1}$. Then,
\[
D_\infty(
\dot{\sigma}^{X' X Y' Y E} \| 
\sigma^{X' X Y' Y} \otimes \sigma^E) \leq
D_\infty(
\sigma^{X' Y' Y E} \| 
\sigma^{X' Y' Y} \otimes \sigma^E) + \log \epsilon^{-1}.
\]
Under an appropriate simultaneous smoothing conjecture, such a term
should give rise to an inequality like
\[
\log L' + \log M' + \log M >
D_\infty^{\epsilon}(
\sigma^{X' Y' Y E} \| 
\sigma^{X' Y' Y} \otimes \sigma^E) + \log \epsilon^{-1} +
\mathrm{polylog}(\epsilon^{-1}),
\]
where $0 < \epsilon < 1$ is the so-called {\em smoothing parameter}.
Above we used the smooth R\'{e}nyi-$\infty$ divergence
$D_\infty^\epsilon(\cdot \| \cdot)$.
The $\epsilon$-smooth R\'{e}nyi divergences are defined as follows:
\[
D_2^\epsilon(\alpha \| \beta) 
:=  
\min_{\alpha' \approx_\epsilon \alpha}
D_2(\alpha' \| \beta), ~~
D_\infty^\epsilon(\alpha \| \beta) 
:= 
\min_{\alpha' \approx_\epsilon \alpha}
D_\infty(\alpha' \| \beta), 
\]
where the minimisation is over all subnormalised density matrices 
$\alpha'$ satisfying $\|\alpha' - \alpha\|_1 \leq \epsilon (\Tr\alpha)$. 
The smooth mutual informations are defined accordingly.
\[
I_2^\epsilon(X' : E)_\sigma
:=  
D_2^\epsilon(\sigma^{X'E} \| \sigma^{X'} \otimes \sigma^{E}), \\
I_\infty^\epsilon(X' : E)_\sigma
:= 
D_\infty^\epsilon(\sigma^{X'E} \| \sigma^{X'} \otimes \sigma^{E}).
\]
We remark
that several alternate definitions of non-smooth and smooth R\'{e}nyi
divergences have been provided in the literature. However the work
of \cite{Jain:minimax} shows that all the smooth R\'{e}nyi-$p$ 
divergences for $p > 1$ are essentially equivalent to 
$D^\epsilon_\infty$ up to tweaks in the smoothing parameter $\epsilon$
and dimension independent additive corrections that are polynomial in
$\log \epsilon^{-1}$. So henceforth, our fully smooth 
convex split and soft covering
results shall be stated in terms of $D^\epsilon_\infty$ or 
quantities derived from $D^\epsilon_\infty$ like $I^\epsilon_\infty$.

Another important entropic quantity derived from R\'{e}nyi divergence
is conditional entropy.
The smooth conditional R\'{e}nyi entropy is defined by
\[
H_2^\epsilon(A|R)_\rho :=
-D_2^\epsilon(\rho^{AR} \| \one^A \otimes \rho^R), ~~
\Hmin^\epsilon(A|R)_\rho :=
-D_\infty^\epsilon(\rho^{AR} \| \one^A \otimes \rho^R).
\]
Various alternate definitions of the smooth conditional 
R\'{e}nyi-2 and R\'{e}nyi min entropy
have been given in the literature, but there all equivalent up to 
minor tweaks in the smoothing parameter $\epsilon$ and small dimension
independent additive terms polynomial in $\log \epsilon^{-1}$.
Also, $H_2^\epsilon(A|R)_\rho \approx \Hmin^\epsilon(A|R)_\rho$
\cite{Jain:minimax}.

Now contrast the above expression with the following viz.
\begin{eqnarray*}
\log L' + \log M' + \log M 
& > &
I_2^\epsilon(X'Y'Y:E)_\sigma + \log \epsilon^{-2} 
\;=\;
D_2^\epsilon(\sigma^{X' Y' Y E} \| 
	     \sigma^{X' Y' Y} \otimes \sigma^E) 
+ \log \epsilon^{-2} \\
& \approx &
D_\infty^\epsilon(
\sigma^{X' Y' Y E} \| \sigma^{X' Y' Y} \otimes \sigma^E)
+ \Theta(\log \epsilon^{-2}) 
\end{eqnarray*}
obtained by telescoping in the pairwise independent case 
\cite{Sen:telescoping}.
Thus, simultaenous smoothing is still potentially much more powerful
than telescoping when pairwise independence is slightly broken in the
sense of $\dot{\sigma}^E_{x'xy'y}$ being bounded by 
$\sigma^E_{x'y'y}$ by at most a scale factor 
of $\mathrm{poly}(\epsilon^{-1})$.

As we will see in Section~\ref{sec:coveringnonpairwise}, our 
machinery allows us to
handle the slight breaking of pairwise independence in the sense described
above, achieving the bound in terms of $D_\infty^\epsilon$ that one would
expect from simultaneous smoothing. We remark that sacrificing pairwise
independence may be computationally beneficial for Alice's obfuscation
strategy, as it will require her to use a lesser number of pure random
bits.

\section{Preliminaries}
\label{sec:prelims}

We now state some basic definitions and facts for completeness. First 
is the well known {\em gentle measurement} lemma.
\begin{fact}
\label{fact:gentle}
Let $\rho$ be a normalised density matrix on a Hilbert space. Let
$\Pi$ be a POVM element the same space (i.e. $\zero \leq \Pi \leq \one$) 
satisfying $\Tr[\Pi \rho] \geq 1 - \epsilon$. Then,
\[
\|\rho - \sqrt{\Pi} \, \rho \, \sqrt{\Pi}\|_1 < 2 \sqrt{\epsilon}.
\]
If $\rho$ is not normalised, we can still obtain a good gentle measurement
upper bound by rescaling by $\Tr \rho$.
\end{fact}

The {\em L\"{o}wner partial order} on Hermitian matrices $A$, $B$ on the
same Hilbert space is defined as $A \leq B$ iff $B - A$ is positive
semidefinite. With this notation, we sometimes write $A \geq 0$ to
indicate that $A$ is positive semidefinite.

Every Hermitian positive semidefinite matrix $\rho^A$ on a Hilbert
space $A$, with complex entries, has a {\em length preserving
purification} $\rho^{AR}$ on an extended Hilbert space $A \otimes R$,
where $R$ is another Hilbert space, often called the {\em reference},
of sufficiently large dimension. By a length preserving purification,
we mean that $\rho^{AR}$ is a rank one orthogonal projection, i.e. a
{\em pure state}, whose of same Schatten-$\ell_1$ norm as that of
$\rho^A$. Note that that the Schatten-$\ell_1$ norm of a pure state
$\rho^{AR} = \ketbra{\rho}^{AR}$ is the Euclidean $\ell_2$-length of 
the vector $\ket{\rho^A}$. 
\begin{definition}
The {\em fidelity} between two
positive semidefinite matrices $\rho^A$, $\sigma^A$ on the same Hilbert
space $A$ is defined as
\[
F(\rho^A, \sigma^A) :=
\|\sqrt{\rho}^A \sqrt{\sigma}^A\|_1 =
\max_{\ket{\rho}^{AR}, \ket{\sigma}^{AR}}
|\braket{\rho}{\sigma}|,
\]
where the maximisation is over all length preserving purifications 
$\ket{\rho}^{AR}$, $\ket{\sigma}^{AR}$ of $\rho^A$, $\sigma^A$.
\end{definition}
Given the above definition of fidelity, a straightforward analogue of 
Uhlmann's theorem holds for pairs of positive semidefinite matrices. 
Also, we have the standard inequalities relating fidelity and trace
distance for normalised density matrices.
\begin{fact}
\label{fact:fidelitytracedist}
Let $\rho$, $\sigma$ be normalised density matrices in the same Hilbert
space. Then,
\[
1 - F(\rho,\sigma) \leq
\frac{1}{2} \|\rho - \sigma\|_1 \leq
\sqrt{1 - F(\rho,\sigma)^2}.
\]
\end{fact}

In this paper, all superoperators will be {\em Hermitian preserving} i.e.
they are linear operators on vector spaces of matrices that send
Hermitian matrices to Hermitian matrices. We will only apply superoperators
to Hermitian matrices inside our proofs.

We will need the notion of
`intersection' of two probability subdistributions $p^X$, $q^X$ on the same
alphabet $X$ from \cite{sen:oneshot}.
\begin{equation}
\label{eq:classicalintersection}
(p \cap q)^X(x) := \min\{p(x), q(x)\}.
\end{equation}
It is easy to see that 
\[
\|(p \cap q)^X - p^X\|_1 \leq \|p^X - q^X\|_1, ~~
\|(p \cap q)^X - q^X\|_1 \leq \|p^X - q^X\|_1.
\]

We state the following two folklore facts for completeness, which have 
been commonly used in earlier literature on decoupling.
\begin{fact}
\label{fact:dupuisoperatorineq}
Let $\rho^{AB}$ be a positive semidefinite matrix.
Then $\rho^{AB} \leq |A| (\one^A \otimes \rho^B)$.
\end{fact}
\begin{fact}
\label{fact:CondHminUpperBound}
Let $\rho^{AB}$ be a subnormalised density matrix.
Then $\Hmin^\epsilon(A|B)_\rho \leq \log |A|$.
\end{fact}

We now recall the so-called 
{\em matrix weighted Cauchy-Schwarz inequality}. This inequality has
been extensively used in earlier literature on decoupling.
\begin{fact}
\label{fact:matrixCauchySchwarz}
Let $M$ be Hermitian matrix on a Hilbert space. Let $\sigma$ be a
normalised density matrix on the same space such that
$\supp(M) \subseteq \supp(\sigma)$. Then,
\[
\|M\|_1 \leq \|\sigma^{-1/4} M \sigma^{-1/4}\|_2,
\]
where $\sigma^{-1}$ is the Moore-Penrose pseudoinverse of $\sigma$ as
defined earlier. If $\Tr \sigma \neq 1$, still an appropriate inequality
can be easily obtained by rescaling by $\Tr \sigma$.
\end{fact}

\subsection{Shaved Cauchy-Schwarz inequality for matrix norms}
Though the matrix weighted Cauchy Schwarz inequality is very useful,
it does have a notable drawback in certain applications. Suppose we want 
to show an upper bound for $\|M_1 - M_2\|_1$, where $M_1$, $M_2$
are Hermitian matrices. In several applications, including during
the proof of the smooth multipartite 
convex split lemma in the next section, $M_1$, $M_2$ have incomparable
supports, and it is not clear what normalised density matrix $\sigma$
should be used as the weighting matrix in 
Fact~\ref{fact:matrixCauchySchwarz} so as to get a good upper bound on
$\|M_1 - M_2\|_1$. For several of those applications, including for
our proof in the next section, it turns out that one can apply the
following {\em shaved Cauchy-Schwarz inequality} of 
Proposition~\ref{prop:shavedCauchySchwarz} in order to get
very good upper bounds on $\|M_1 - M_2\|_1$. The term `shaved' refers
to the use of the POVM element $\Pi$ in 
Proposition~\ref{prop:shavedCauchySchwarz} which shaves off the `bad parts'
of $\rho_1$, $\rho_2$ and applies Cauchy Schwarz only on the
good part $\sqrt{\Pi}(\rho_1 - \rho_2)\sqrt{\Pi}$.
\begin{proposition}
\label{prop:shavedCauchySchwarz}
Suppose $\rho_1$, $\rho_2$ are two positive semidefinite matrices on the
same Hilbert space. Let $\Pi$ be a POVM element on the same space.
Suppose $\Tr[\Pi \rho_i] \geq (1 - \epsilon_i) (\Tr \rho_i)$. Then,
\[
\|\rho_1 - \rho_2\|_1 \leq
\sqrt{\Tr \Pi} \cdot \|\rho_1 - \rho_2\|_2 +
(2 \Tr \rho_1) \sqrt{\epsilon_1} +
(2 \Tr \rho_2) \sqrt{\epsilon_2}.
\]
\end{proposition}
\begin{proof}
We just apply Fact~\ref{fact:gentle} and the standard Cauchy-Schwarz
inequality, together with the triangle inequality.
\begin{eqnarray*}
\|\rho_1 - \rho_2\|_1
& \leq &
\|\sqrt{\Pi} (\rho_1 - \rho_2) \sqrt{\Pi} \|_1 + 
\|\rho_1 - \sqrt{\Pi}\,\rho_1\,\sqrt{\Pi}\|_1 +
\|\rho_2 - \sqrt{\Pi}\,\rho_2\,\sqrt{\Pi}\|_1 \\
& \leq &
\sqrt{\Tr \Pi} \cdot \|\sqrt{\Pi} (\rho_1 - \rho_2) \sqrt{\Pi} \|_2 + 
(2 \Tr \rho_1) \sqrt{\epsilon_1} +
(2 \Tr \rho_2) \sqrt{\epsilon_2} \\
& \leq &
\sqrt{\Tr \Pi} \cdot \|\rho_1 - \rho_2\|_2 + 
(2 \Tr \rho_1) \sqrt{\epsilon_1} +
(2 \Tr \rho_2) \sqrt{\epsilon_2}.
\end{eqnarray*}
\end{proof}

\subsection{Projector smoothing of matrix $\|\cdot\|_\infty$ norm}
As we will see below, our novel flattening lemmas reduce considerations
about $D_\infty(\alpha \| \beta)$ to considerations about
$\|\alpha'\|_\infty$, where $\alpha'$ is a matrix related to $\alpha$,
$\beta$. Accordingly, considerations about 
$D_\infty^\epsilon(\alpha \| \beta)$ reduce to considerations about
a smoothed version of $\|\alpha'\|_\infty$. Naively speaking, the smoothed 
version of $\|\alpha'\|_\infty$ that arises from this reduction replaces
positive semidefinite matrix $\alpha'$ by a smoothed positive
semidefinite matrix $\hat{\alpha'}$ close to it leading to
the {\em state smoothed Schatten-$\ell_\infty$ norm}.
\[
\|\alpha'\|_\infty^\epsilon :=
\min_{\hat{\alpha'}} \|\hat{\alpha'}\|_\infty,
\]
where the minimisation is over all positive semidefinite matrices
$\hat{\alpha'}$ satisfying 
$\|\alpha' - \hat{\alpha'}\|_1 \leq \epsilon (\Tr\rho)$.
However for our purposes, as will become clear later, it is better to 
smooth $\alpha'$ by applying a suitable orthogonal projector to it. 
We call this approach {\em projector smoothing} as opposed to 
{\em state smoothing}. The following proposition relates the two
types of smoothing in the context of the Schatten-$\ell_\infty$ norm.
\begin{proposition}
\label{prop:projsmoothing}
Let $0 < \epsilon < 1/4$.
Let $\rho$ be a  positive semidefinite matrix on a Hilbert space.
Let $\rho'$ be a positive
semidefinite matrix in the same space achieving the optimum in 
the definition of $\|\rho\|_\infty^\epsilon$.
Then there exists an orthogonal projector in the same space such that
\[
\Pi \rho = \rho \Pi, ~~
\Tr[\Pi \rho] \geq (1 - \sqrt{\epsilon}) (\Tr \rho), ~~
\|\Pi \rho \Pi\|_\infty \leq (1 + 2\sqrt{\epsilon}) \|\rho'\|_\infty.
= (1 + 2\sqrt{\epsilon}) \|\rho\|_\infty^\epsilon.
\]
\end{proposition}
\begin{proof}
Consider the matrices $\rho$, $\rho'$ expressed in the eigenbasis
$\{\ket{i}\}_i$ of $\rho$. Define the subset 
\[
\mbox{Bad} := \{i: \rho(i,i) > (1 + 2\sqrt{\epsilon}) \rho'(i,i)\}.
\]
Then,
\begin{eqnarray*}
\epsilon (\Tr \rho)
& \geq &
\|\rho - \rho'\|_1 
\;\geq\;
\sum_{i \in \mathrm{Bad}} (\rho(i,i) - \rho'(i,i)) 
\;\geq\;
\sum_{i \in \mathrm{Bad}} 
\rho(i,i) (1 - \frac{1}{1 + 2\sqrt{\epsilon}}) \\
&   =  &
\frac{2\sqrt{\epsilon}}{1 + 2\sqrt{\epsilon}}
\sum_{i \in \mathrm{Bad}}, \\
\implies 
\sum_{i \in \mathrm{Bad}}
\rho(i,i) 
& \leq &
\frac{(1+2\sqrt{\epsilon})\sqrt{\epsilon}}{2} (\Tr \rho)
\;  < \;
\sqrt{\epsilon} (\Tr \rho).
\end{eqnarray*}

Let $\Pi$ be the projector in the eigenbasis of $\rho$ that exactly
skips $i \in \mbox{Bad}$. Then,
$
\Pi \rho = \rho \Pi, 
$
$
\Tr[\Pi \rho] > (1 - \sqrt{\epsilon}) (\Tr \rho),
$
and
\[
\|\Pi \rho \Pi\|_\infty =
\max_{i \not \in \mathrm{Bad}} \rho(i,i) \leq
\max_{i \not \in \mathrm{Bad}} (1 + 2\sqrt{\epsilon}) \rho'(i,i) \leq
(1 + 2\sqrt{\epsilon}) \|\rho'\|_\infty.
\]
This completes the proof of the proposition.
\end{proof}

\subsection{From matrix $\|\cdot\|_2$ norm to inner product}
A crucial technique required in the proof of our
smooth multipartite convex split lemma is a technique to convert
the Schatten-$\ell_2$ norm of a positive semidefinite matrix into 
a Hilbert-Schmidt inner product of two related positive semidefinite 
matrices.
\begin{lemma}
\label{lem:asyml2}
Let $0 < \epsilon < 1/4$.
Suppoose $\rho_1$, $\rho_2$ are positive semidefinite operators on the
same Hilbert space. Let $\|\rho_1 - \rho_2\|_1 \leq \epsilon$.
Then there exists a positive semidefinite matrix $\rho'_1$ on the same
space such that
\[
\rho'_1 \leq \rho_1, ~~
\|\rho'_1 - \rho_1\|_1 < \sqrt{\epsilon}, ~~
\|\rho'_1\|_2^2 \leq (1 + 2\sqrt{\epsilon}) \Tr[\rho_1 \rho_2].
\]
\end{lemma}
\begin{proof}
We follow the lines of the proof of Proposition~\ref{prop:projsmoothing}, 
and work in the eigenbasis $\{\ket{i}\}_i$ of $\rho_1$. 
Define the subset 
\[
\mbox{Bad} := \{i: \rho_1(i,i) > (1 + 2\sqrt{\epsilon}) \rho_2(i,i)\}
\implies
\sum_{i \in \mathrm{Bad}} \rho_1(i,i) < \sqrt{\epsilon} (\Tr \rho).
\]
Define 
\[
\rho'_1 :=
\sum_{i \not \in \mathrm{Bad}} \rho_1(i,i) \ketbra{i} 
\implies
\rho'_1 \leq \rho_1, ~~
\|\rho_1 - \rho'_1\|_1 < \sqrt{\epsilon} (\Tr \rho).
\]
Now,
\begin{eqnarray*}
\|\rho'_1\|_2^2 
& = &
\Tr[(\rho'_1)^2]
\;=\;
\sum_{i \not \in \mathrm{Bad}} \rho_1(i,i)^2 \\
& \leq &
(1 + 2\sqrt{\epsilon}) 
\sum_{i \not \in \mathrm{Bad}} \rho_1(i,i) \rho_2(i,i)
\;\leq\;
(1 + 2\sqrt{\epsilon}) 
\sum_{i} \rho_1(i,i) \rho_2(i,i) \\
& = &
(1 + 2\sqrt{\epsilon}) \Tr[\rho_1 \rho_2].
\end{eqnarray*}
This completes the proof of the lemma.
\end{proof}

\subsection{Tilting and augmentation smoothing}
\label{subsec:tilting}
In this subsection, we state some results about {\em tilting} and
{\em augmentation smoothing} of quantum states from
\cite{sen:oneshot}, adapted to the specific requirements of this
paper.

Consider a normalised density matrix $\rho^{XYM}$. Introduce new
Hilbert spaces $L_X$, $L_Y$ of sufficiently large (how large will
be clarified below) and equal dimension $L$. The {\em augmented state}
is defined as
\[
\rho^{(L_X X) (L_Y Y) M} :=
\frac{\one^{L_X}}{L} \otimes \frac{\one^{L_Y}}{L} \otimes \rho^{XYM}.
\]
As will become clear in Lemma~\ref{lem:smoothing} below, the 
aim of augmentation is to reduce the
Schatten-$\ell_\infty$ norm  of
$\rho^{L_X X}$, $\rho^{L_Y Y}$, which will play a crucial role in
a {\em smoothing property} to be stated below.

Introduce orthogonal copies of the Hilbert space $M$ which we will
label as
$\{M(l_x)\}_{l_x \in [L]}$ and
$\{M(l_y)\}_{l_y \in [L]}$. The orthogonal direct sum
of $M$ together with all its copies is denoted by $\hM$ i.e.
\[
\hM := 
M \oplus 
\left(
\bigoplus_{l_x \in [L]} M(l_x)
\right) \oplus
\left(
\bigoplus_{l_y \in [L]} M(l_y)
\right) 
\cong
M \oplus (M \otimes L_X) \oplus 
(M \otimes L_Y).
\]
We interchangeably think of the augmented state as residing
in $L_X X L_Y Y M$ when we will write
$\rho^{(L_X X) (L_Y Y) M}$, or as residing in $L_X X L_Y Y \hM$ when
we will write $\rho^{(L_X X) (L_Y Y) \hM}$.

Let $0 < \epsilon < 1$. Let $\ket{m}^M$ denote a computational basis
state of $M$. For any fixed $l_x \in [L]$, $l_y \in [L]$, 
the corresponding {\em tilting isometry} embedding $M$ into 
$\hM$ is defined as
\[
T^{M \rightarrow \hM}_{l_x, l_y, \epsilon}(\ket{m}^M) :=
\sqrt{1 - 2 \epsilon} \ket{m}^M +
\sqrt{\epsilon} \ket{m(l_x)}^{M(l_x)} + 
\sqrt{\epsilon} \ket{m(l_y)}^{M(l_y)},
\]
with the extension to the entire vector space $M$ by linearity.
Intuitively, tilting a vector means that we rotate it towards
an orthogonal direction by a small precise amount.

We next define some other tilting maps which are subnormally
scaled isometric embeddings of $M$ into $\hM$. For example,
\[
T^{M \rightarrow \hM}_{l_x, \epsilon}(\ket{m}^M) :=
\sqrt{1 - 2 \epsilon} \ket{m}^M +
\sqrt{\epsilon} \ket{m(l_x)}^{M(l_x)}.
\]
Similarly, we define $T^{M \rightarrow \hM}_{l_y, \epsilon}$. The
last two
of them are isometric embeddings of $M$ into $\hM$ with scale factor
$\sqrt{1 - \epsilon}$. 

Define another tilting map in the analogous fashion:
\[
T^{L_X L_Y M \rightarrow L_X L_Y \hM}_{\epsilon}
(\ket{l_x}^{L_X} \ket{l_y}^{L_X} \ket{m}^M) :=
\ket{l_x}^{L_X} \ket{l_y}^{L_Y} \otimes 
T^{M \rightarrow \hM}_{l_x, l_y, \epsilon}(\ket{m}^M).
\]
The tilting maps $T^{L_X M \rightarrow L_X \hM}_{\epsilon}$,
$T^{L_Y M \rightarrow L_Y \hM}_{\epsilon}$ are defined similarly. 

Finally, we define the tilting map:
\[
T^{L_X X L_Y Y M \rightarrow L_X X L_Y Y \hM}_{\epsilon} :=
\I^{X} \otimes \I^{Y} \otimes 
T^{L_X L_Y M \rightarrow L_X L_Y \hM}_{\epsilon}.
\]
The tilting maps $T^{L_X X M \rightarrow L_X X \hM}_{\epsilon}$,
$T^{L_Y Y M \rightarrow L_Y Y \hM}_{\epsilon}$ are then defined 
similarly.
We remark that all the above tilting maps are isometric embeddings
with scale factors $\sqrt{1 - \epsilon}$ or $1$.
Also, all the tilting maps defined above
can be extended from spaces of
vectors to spaces of Hermitian matrices in the natural fashion, becoming
superoperators if need be.

We remark that the definition of tilting given above is slightly simpler
than the definition in \cite{sen:oneshot}. This simplification
makes the proofs of subsequent statements much shorter; nevertheless
it is still strong enough for all known applications.

Let $\Pi^{XYM}(1)$, $\Pi^{XYM}(2)$, $\Pi^{XYM}(3)$ be 
certain orthogonal projectors in $XYM$ satisfying
\begin{equation}
\label{eq:typicalproj}
\Tr[\Pi^{XYM}(i) \rho^{XYM}] \geq 1 - \sqrt{\epsilon}, ~~~
\forall i \in [3].
\end{equation}
In the context of smooth multipartite convex split, these projectors 
will be related to optimising states
in the definitions of certain smooth entropic quantities. 
We note that, in general, there is no relation amongst them.

Define the {\em $\epsilon$-tilted approximate intersection} of the 
projectors
$\Pi^{XYM}(1)$, $\Pi^{XYM}(2)$, $\Pi^{XYM}(3)$ as follows:
\begin{equation}
\label{eq:approxintersection}
\begin{array}{rcl}
\hPi^{L_X X L_Y Y \hM} 
& := &
\one^{L_X X L_Y Y \hM} - {} \\
&   &
~~~
\spanning\{
(T^{L_X X M \rightarrow L_X X \hM}_{\epsilon^{1/4}} \otimes \I^{L_Y Y})(
\one^{L_X L_Y} \otimes (\one^{XYM} - \Pi^{XYM}(1))), \\
&   &
~~~~~~~~~~~~~~~~
(T^{L_Y Y M \rightarrow L_Y Y \hM}_{\epsilon^{1/4}} \otimes \I^{L_X X})(
\one^{L_X L_Y} \otimes (\one^{XYM} - \Pi^{XYM}(2))), \\
&   &
~~~~~~~~~~~~~~~~
T^{L_X X L_Y Y M \rightarrow L_X X L_Y Y \hM}_{\epsilon^{1/4}}(
\one^{L_X L_Y} \otimes (\one^{XYM} - \Pi^{XYM}(3)))
\},
\end{array}
\end{equation}
where the notation `span' above means that we take the orthogonal
projection onto the vector space span of the supports of the 
operators in the expression.
Intuitively, we are just defining intersection via
de Morgan's rule i.e.
intersection is the complement of the union of complements. 
Complement of a subspace, or rather of the orthogonal projection onto
that subspace, has a precise meaning in terms of matrices.
However, there is a problem defining the  union of subspaces
as explained in more detail in \cite{sen:oneshot}; naively
taking the span of subspaces often results in a catastrophe as it
can easily be the entire Hilbert space!
So, we need the
more sophisticated approach of \cite{sen:oneshot} viz. tilt the
subspaces into orthogonal directions, then take their span, and then
take the complement.

Following \cite{sen:oneshot} again, we now define a normalised
followed by a subnormalised density matrix. Both of them are
close to the augmented normalised density matrix $\rho^{L_X X L_Y Y M}$.
\begin{equation}
\label{eq:hrhohhrho}
\begin{array}{rcl}
\hrho^{L_X X L_Y Y \hM}
& := &
T^{L_X X L_Y Y M \rightarrow L_X X L_Y Y \hM}_{\epsilon^{1/4}}(
\rho^{L_X X L_Y Y M}), \\
\hhrho^{L_X X L_Y Y \hM}
& := &
\hPi^{L_X X L_Y Y \hM} \circ \rho^{L_X X L_Y Y \hM}.
\end{array}
\end{equation}
These two states satisfy the following properties.
\begin{lemma}
\label{lem:tiltclose}
$
\|\hrho^{L_X X L_Y Y \hM} - \rho^{L_X X L_Y Y \hM}\|_1 =
2\sqrt{2} \cdot \epsilon^{1/8}.
$
\end{lemma}
\begin{proof}
Easily follows from the definition of the tilting isometry.
\end{proof}
\begin{lemma}
\label{lem:matrixtiltedspan}
$
\Tr[\hhrho^{L_X X L_Y Y \hM}] \geq
1 - 121 \epsilon^{1/4}.
$
\end{lemma}
\begin{proof}{\bf (Sketch)}
Easily follows from Equation~\ref{eq:typicalproj} and 
Proposition~3 of \cite{sen:oneshot} regarding the 
`union' properties of the so-called {\em matrix tilted span}.
\end{proof}
\begin{corollary}
\label{cor:matrixtiltedspan}
$
\|\hhrho^{L_X X L_Y Y \hM} - \rho^{L_X X L_Y Y \hM}\|_1 <
22 \cdot \epsilon^{1/8}.
$
\end{corollary}
\begin{proof}
Follows by applying Fact~\ref{fact:gentle} to 
Lemma~\ref{lem:matrixtiltedspan}.
\end{proof}
\begin{corollary}
\label{cor:triangleineq}
$
\|\hrho^{L_X X L_Y Y \hM} - \hhrho^{L_X X L_Y Y \hM}\|_1 <
25 \cdot \epsilon^{1/8}.
$
\end{corollary}
\begin{proof}
Triangle inequality and Lemma~\ref{lem:tiltclose} and
Corollary~\ref{cor:matrixtiltedspan}.
\end{proof}

We now state an important lemma that follows from the 
smoothing properties of a tilted augmented state shown in
\cite{sen:oneshot}. It essentially states that tracing out 
$L_X$ removes the tilt in $\hM$ along $l_x$ in terms of the
Schatten $\|\cdot\|_\infty$ norm, and similarly for tracing out $L_Y$.
\begin{lemma}
\label{lem:smoothing}
For any $l_x, l_y \in [L]$ and positive semidefinite matrix $\sigma^M$,
\begin{eqnarray*}
\|
L^{-1} \sum_{l'_y \in [L]} 
T_{l_x,l'_y,\epsilon^{1/4}}^{M \rightarrow \hM}(\sigma^M) -
T_{l_x,\epsilon^{1/4}}^{M \rightarrow \hM}(\sigma^M)
\|_\infty 
& < &
4 \epsilon^{1/8} L^{-1/2} (\Tr \sigma), \\
\|
L^{-1} \sum_{l'_x \in [L]} 
T_{l'_x,l_y,\epsilon^{1/4}}^{M \rightarrow \hM}(\sigma^M) -
T_{l_y,\epsilon^{1/4}}^{M \rightarrow \hM}(\sigma^M)
\|_\infty 
& < &
4 \epsilon^{1/8} L^{-1/2} (\Tr \sigma), \\
\|
L^{-2} \sum_{l'_x, l'_y \in [L]} 
T_{l'_x,l'_y,\epsilon^{1/4}}^{M \rightarrow \hM}(\sigma^M) -
\sigma^M
\|_\infty 
& < &
8 \epsilon^{1/8} L^{-1/2} (\Tr \sigma).
\end{eqnarray*}
As a consequence, for any augmented normalised density matrix 
$\rho^{L_X X L_Y Y M}$,
\begin{eqnarray*}
\|\hrho^{L_X X \hM} - 
  T_{\epsilon^{1/4}}^{L_X X M \rightarrow L_X X \hM}(\rho^{L_X X M})
\|_\infty
& < &
4 \epsilon^{1/8} L^{-3/2}, \\
\|\hrho^{L_Y Y \hM} - 
  T_{\epsilon^{1/4}}^{L_Y Y M \rightarrow L_Y Y \hM}(\rho^{L_Y Y M})
\|_\infty
& < &
4 \epsilon^{1/8} L^{-3/2}, \\
\|\hrho^{\hM} - \rho^{M}\|_\infty
& < &
8 \epsilon^{1/8} L^{-1/2}.
\end{eqnarray*}
\end{lemma}
\begin{proof}
Follows from Proposition~4 of \cite{sen:oneshot}.
\end{proof}

\subsection{Flattening quantum states}
In this subsection, we shall show that a natural quantum operation
called {\em flattening}, aka the `brother construction', 
actually preserves the fidelity between quantum
states. The {\em  brother construction} is a well known argument
in classical probability theory that is sometimes used to flatten a 
specific probability distribution
by enlarging its support.  It works as follows: Suppose $p^X$ is a 
probability distribution on the alphabet $X$. Assume for simplicity
that probability of each symbol in $X$ is a rational number with 
denominator $F_p$. We note that taking $F_p$ to be large enough 
approximates $p^X$ quite well; this point will be made more precise
(and also quantum!) when we actually state our technical lemmas below.
So for now we assume that $p(x) = \frac{a_x}{F_p}$ for all $x \in X$.
We take a new alphabet $L$ large enough, and define a probability
distribution $p^X$ on $X L$ as follows:
\[
p^{XL}(x,l) := 
\left\{
\begin{array}{l l}
\frac{1}{F_p} & l \leq a_x, \\
& \\
0             & l > a_x.
\end{array}
\right.
\]
It is now easy to see that if $p^{XL}(x,l) \neq 0$ for some $(x,l)$,
then $p^{XL}(x,l) = \frac{1}{F_p}$. In other words, $p^{XL}$ is 
{\em flat on its support}. For a fixed $x$, the set of pairs
$\{(x,l): l \leq a_x\}$ are called the `brothers' of $x$.

Note that the above brothers construction was tailored for a specific
probability distribution $p^X$. It will flatten only $p^X$. A lesser
known fact about the construction is that flattening according to 
$p^X$ nevertheless preserves the $\ell_1$-distance between any
two probability distributions $q_1^X$, $q_2^X$. Flattening sometimes
leads to cleaner proofs in classical probability theory. As we will
see later in our proof of our smooth multipartite convex split
lemma, the quantum analogue of flattening plays
a crucial role in overcoming bottleneck due to non-commutativity of
certain operators involved in the proof.

Flattening ideas have been used in quantum information theory before.
Anshu and Jain \cite{decoupling_convexsplit} take inspiration 
from the classical
brothers construction and defines a unitary quantum flattening
operation. However in order to define their unitary, they need to 
add a small ancilla in addition to the `brother' Hilbert space. That
ancilla is initialised with a suitable {\em embezzling pure state},
which is almost unchanged at the end of the unitary operation. Anshu
and Jain use this unitary in the context of one of their more
economical convex split lemmas.
The extra
ancilla containing the embezzling state does not hurt their application.
However, such an ancilla creates serious problems in our attempt to
flatten a quantumn state in order overcome the non-commutativity 
bottleneck amongst certain operators in our smooth multipartite convex
split proof. Hence, we cannot use the unitary flattening operation of
\cite{decoupling_convexsplit} in this paper.

It turns out that our final solution for this problem is even simpler 
than the unitary
flattening operation of \cite{decoupling_convexsplit}. We take the vanilla brothers
construction outlined above and convert it in the most naive fashion into
a CPTP superoperator acting on quantum 
states. Our flattening superoperator is not a unitary operation. 
It does not use embezzling states.
As expected it is tailored to flatten a specific density matrix
$\sigma^A$ on the Hilbert space $A$. Serendipitously, it preserves the 
fidelity between any two
quantum states $\rho_1^A$, $\rho_2^A$, though not the trace distance.
Nevertheless, this novel observation suffices to meet all our 
requirements. We consider
this simple but powerful operation to be a major conceptual contribution
of this work.

We now describe our quantum flattening operation. Let $\sigma^A$ be
a positive semidefinite operator on the Hilbert space $A$. Let
its eigenvalues and eigenvectors be 
$\{(\sigma(a,a), \ket{a}^A)\}_a$. Let $0 < \delta < 1$. Let $F_\sigma$
be the smallest positive integer such that there exist integers
$\{f(a)\}_a$ with the property that
\[
\forall a: ~~
\sigma(a,a) \leq \frac{f(a)}{F_\sigma} \leq \sigma(a,a) (1+\delta).
\]
Then the {\em flattening superoperator with respect to $\sigma^A$} is
defined as
\[
\cF_\sigma^{A \rightarrow AL}(\rho^A) := 
\sum_{l = 1}^L 
A_l^{A \rightarrow AL} 
(\rho^A) 
(A_l^\dag)^{AL \rightarrow A},
\]
where $\rho^A$ is any Hermitian matrix on $A$, $L := \max_a f(a)$ and for
each $l \in [L]$, $A_l$ is a linear map from the Hilbert space $A$ 
to the Hilbert space $A L := A \otimes \C^L$ defined as:
\[
A_l^{A \rightarrow AL}(\ket{a}^A) := 
\left\{
\begin{array}{l l}
\frac{1}{\sqrt{f(a)}} \ket{a}^A \ket{l}^L & l \leq f(a), \\
& \\
0 & l > f(a),
\end{array}
\right.
\]
with the extension to the full vector space $A$ by linearity. 

To get a
better geometric picture of $\cF_\sigma^{A \rightarrow AL}$, it helps 
to order the
eigenvectors of $\sigma^A$ as $\{\ket{a_1}, \ket{a_2}, \ldots, \}$ in
non-decreasing order of $f(a)$. Introduce new notation $A(1) \equiv A$.
Consider fresh copies $A(l)$, $2 \leq l \leq L$ of the Hilbert space $A$.
We can then write 
\[
AL = A \otimes \C^L = \bigoplus_{l=1}^L A(l).
\]
Define $f(a_0) \equiv 0$.
For $i \in [|A|]$, define
$
A_{(i)} := 
\spanning\{\ket{a_i}, \ket{a_{i+1}}, \ldots, \ket{a_{|A|}}\}.
$
For $i \in [A]$ and $l \in [L]$, define $T_{i,l}$ to be the identity 
isometric embedding of $A_{(i)}$ into $A(l)$. Depending on the 
context, $T_{i,l}$ can also be thought of as the identity superoperator
mapping matrices on $A_{(i)}$ to matrices on $A(l)$.
Then, 
\[
\cF_\sigma^{A \rightarrow AL} =
\bigoplus_{i=1}^{|A|} \bigoplus_{l = f(a_{i-1})+1}^{f(a_i)}
T_{i,l}^{A_{(i)} \rightarrow A(l)},
\]
i.e., $\cF_\sigma^{A \rightarrow AL}$ is an orthogonal direct sum
of projection superoperators.

We now prove the following lemma.
\begin{lemma}
The superoperator $\cF_\sigma^{A \rightarrow AL}$ is completely
positive and trace preserving.
\end{lemma}
\begin{proof}
The Kraus representation of $\cF_\sigma^{A \rightarrow AL}$ given above
immediately implies it is completely positive. To check the trace
preserving property, we work in the basis of $A$ given by the
eigenvectors of $\sigma^A$ and obtain,
\begin{eqnarray*}
\sum_{l=1}^L A_l^\dag A_l(a,a')
& = &
\bra{a} (\sum_{l=1}^L A_l^\dag A_l) \ket{a'}
\;=\;
\sum_{l=1}^{\min\{f(a'), f(a)\}} 
\frac{1}{\sqrt{f(a') f(a)}} 
\bra{a,l} \ket{a',l} \\
& = &
\delta_{a a'} \sum_{l=1}^{f(a)} \frac{1}{f(a)} 
\;=\;
\delta_{a a'}, \\
\implies 
\sum_{l=1}^L A_l^\dag A_l(a,a')
& = &
\one^A.
\end{eqnarray*}
This completes the proof of the lemma.
\end{proof}

Define the subspace $F'_\sigma$ of $AL$ by
\[
F'_\sigma := \spanning\{\ket{a}^A \ket{l}^L: l \leq f(a)\}.
\]
In other words, the image of $\cF_\sigma^{A \rightarrow AL}$ is contained
in the span of Hermitian operators with support in $F'_\sigma$.
The next lemma shows that the dimension of $F'_\sigma$ is approximately
$F_\sigma (\Tr\sigma)$.
\begin{lemma}
\label{lem:flatteningdenom}
$
(\Tr\sigma) F_\sigma \leq |F'_\sigma| \leq
(1+\delta) (\Tr\sigma) F_\sigma.
$
\end{lemma}
\begin{proof}
\begin{eqnarray*}
|F'_\sigma|
& = &
\sum_a f(a) 
\;=\;
F_\sigma \sum_a \frac{f(a)}{F_\sigma}, \\
F_\sigma \sum_a \sigma(a,a) 
& \leq &
F_\sigma \sum_a \frac{f(a)}{F_\sigma}
\;\leq\;
F_\sigma \sum_a \sigma(a,a) (1+\delta) \\
\implies
F_\sigma (\Tr\sigma)
& \leq &
|F'_\sigma| 
\;\leq\;
(1+\delta) F_\sigma (\Tr\sigma).
\end{eqnarray*}
This completes the proof of the lemma.
\end{proof}

We are now finally in a position to prove that 
$\cF_\sigma^{A \rightarrow AL}$ flattens $\sigma^A$.
\begin{proposition}
\label{prop:flattening}
Let $\sigma^A$ be a positive semidefinite operator on $A$. Then,
\[
(1+\delta)^{-1} \frac{\one^{F'_\sigma}}{F_\sigma} \leq
\cF_\sigma^{A \rightarrow AL}(\sigma) \leq
\frac{\one^{F'_\sigma}}{F_\sigma}.
\]
\end{proposition}
\begin{proof}
\begin{eqnarray*}
\cF^{A \rightarrow AL}_\sigma(\sigma)
& = &
\sum_a \sigma(a,a) \cF^{A \rightarrow AL}_\sigma(\ketbra{a}) 
\;=\;
\sum_a \sigma(a,a) \sum_{l=1}^L A_l \ketbra{a} A_l^\dag 
\;=\;
\sum_a \frac{\sigma(a,a)}{f(a)} \ketbra{a}^A
\sum_{l=1}^{f(a)} \ketbra{l}^L.
\end{eqnarray*}
But,
\begin{eqnarray*}
\sigma(a,a)
& \leq &
\frac{f(a)}{F_\sigma} 
\;\leq\;
(1+\delta) \sigma(a,a) \\
\implies
\frac{(1+\delta)^{-1}}{F_\sigma}
& \leq &
\frac{\sigma(a,a)}{f(a)}
\;\leq\;
\frac{1}{F_\sigma}.
\end{eqnarray*}
This implies
\begin{eqnarray*}
\frac{(1+\delta)^{-1}}{F_\sigma}
\sum_a \ketbra{a}^L
\sum_{l=1}^{f(a)} \ketbra{l}^L
& \leq &
\cF_\sigma^{A \rightarrow AL}(\sigma^A)
\;\leq\;
\frac{1}{F_\sigma}
\sum_a \ketbra{a}^L
\sum_{l=1}^{f(a)} \ketbra{l}^L.
\end{eqnarray*}
But, 
$
\sum_a \ketbra{a}^L
\sum_{l=1}^{f(a)} \ketbra{l}^L = \one^{F'_\sigma}.
$
This completes the proof of the proposition.
\end{proof}

We now show the important property that flattening does not change the
length preserving fidelity.
\begin{proposition}
\label{prop:fidelitypreserving}
Let $\sigma^A$, $\rho_1^{AB}$, $\rho_2^{AB}$ be positive semidefinite
operators on their respective Hilbert spaces. Then,
\[
F((\cF_\sigma^{A \rightarrow LA} \otimes \I^B)(\rho_1^{AB}),
  (\cF_\sigma^{A \rightarrow LA} \otimes \I^B)(\rho_2^{AB})) =
F(\rho_1^{AB}, \rho_2^{AB}).
\]
\end{proposition}
\begin{proof}
Consider the following unitary Stinespring dilation 
$U^{A \rightarrow A L L'}$ of $\cF_\sigma^{A \rightarrow A L}$, where
$A L L' := A \otimes \C^L \otimes \C^L$.
\[
U: \ket{a}^A \mapsto
\ket{a}^A \otimes
\left(\frac{1}{\sqrt{f(a)}}\sum_{l=1}^{f(a)} \ket{l}^L \ket{l}^{L'}\right).
\]
Define for $i = 1, 2$,
$
\sigma_i^{L' L A B} := 
(U^{A \rightarrow ALL'} \otimes \one^B) \circ \rho_i^{AB}.
$
Since unitary superoperators preserve fidelity, we have
\[
F(\rho_1^{AB}, \rho_2^{AB}) = F(\sigma_1^{L'L AB}, \sigma_2^{L'L AB}).
\]

Let $\ket{\sigma_1}^{L'L AB R}$, $\ket{\sigma_2}^{L'L AB R}$ be 
length preserving purifications of $\sigma_1^{L'L AB}$, 
$\sigma_2^{L'L AB}$ achieving their fidelity i.e.
\[
F(\sigma_1^{L'L AB}, \sigma_2^{L'L AB}) =
F(\sigma_1^{L'L AB R}, \sigma_2^{L'L AB R}) =
\braket{\sigma_1}{\sigma_2}.
\]
Then,
\[
\ket{\sigma_1}^{L'LABR} =
\sum_{l=1}^L \ket{l}^{L'} \ket{l}^L \ket{v_1(l)}^{ABR},
~~~
\ket{\sigma_2}^{L'LABR} =
\sum_{l=1}^L \ket{l}^{L'} \ket{l}^L \ket{v_2(l)}^{ABR},
\]
where $\{\ket{v_1(l)}^{ABR}\}_l$, $\{\ket{v_2(l)}^{ABR}\}_l$ are in general
unnormalised, non-orthgonal pure states.

By the definition of the flattening operator $\cF_\sigma$,
$
\sigma_i^{LAB} :=
(\cF_\sigma^{A \rightarrow AL} \otimes \I^B)(\rho_i^{AB}),
$
$i = 1, 2$ 
have an orthogonal direct sum structure indexed by $l \in L$ i.e.
\[
\sigma_i^{LAB} =
\bigoplus_{l=1}^L \ketbra{l}^L \otimes \sigma_i(l)^{AB},
\]
for some positive semidefinite operators $\sigma_i(l)^{AB}$. In other 
words, $\sigma_1^{LAB}$, $\sigma_2^{LAB}$ are classical on $L$. 
As a result,
$
\|\ket{v_i(l)}^{ABR}\|_2^2  = \Tr[\sigma_i(l)^{AB}],
$
and $\ket{v_i(l)}^{ABR}$ is a length preserving purification of
$\sigma_i(l)^{AB}$. Moreover,
\[
F(\sigma_1(l)^{AB}, \sigma_2(l)^{AB}) =
\braket{v_1(l)}{v_2(l)}.
\]
Otherwise, the maximum overlap of length preserving purifications required
by the fidelity of positive semidefinite operators $\sigma_1^{LAB}$,
$\sigma_2^{LAB}$ will not be achieved.

We now investigate the fidelity of positive operators $\sigma_1^{LAB}$ and
$\sigma_2^{LAB}$. 
Take $\ket{\sigma_1}^{L'LABR}$ as the purification of
$\sigma_1^{LAB}$. 
Let $\ket{\sigma'_2}^{L'LABR}$ be a fidelity achieving length preserving
purification of $\sigma_2^{LAB}$. Since the flattening operation is
trace preserving, we have 
$
\Tr[\sigma_i^{LAB}] = \Tr[\rho_i^{AB}],
$
$i = 1, 2$. Let
\[
\ket{\sigma'_2}^{L'LABR} =
\sum_{l',l=1}^L \ket{l'}^{L'} \ket{l}^L \ket{v_2(l',l)}^{ABR},
\]
Then the unnormalised pure state
\[
\ket{\sigma'_2(l)}^{L'ABR} :=
\sum_{l'=1}^L \ket{l'}^{L'} \ket{v_2(l',l)}^{ABR}
\]
is a length preserving purification of $\sigma_2(l)^{AB}$ i.e.
$
\|\ket{\sigma_2'(l)}^{L'ABR}\|_2^2  = \Tr[\sigma_2(l)^{AB}].
$
One way to maximise the overlap between 
$\ket{\sigma_1}^{L'LABR}$ and $\ket{\sigma'_2}^{L'LABR}$ is to set
$\ket{v_2'(l',l)} = 0$ if $l' \neq l$ and
$\ket{v_2'(l,l)} = \ket{v_2(l)}$.
In other words, we can take 
$\ket{\sigma_2'}^{L'LABR} = \ket{\sigma_2}^{L'LABR}$. 

Thus,
\[
F(\sigma_1^{LAB}, \sigma_2^{LAB}) = 
\braket{\sigma_1}{\sigma_2} =
F(\sigma_1^{L'LAB}, \sigma_2^{L'LAB}) = 
F(\rho_1^{AB}, \rho_2^{AB}),
\]
completing the proof of the proposition.
\end{proof}

\section{Fully smooth multipartite convex split via flattening}
\label{sec:convexsplitflatten}
In this section we prove a smooth multipartite convex split lemma,
Lemma~\ref{lem:convexsplitflattening},
using flattening, tilting and augmentation smoothing. At first glance,
Lemma~\ref{lem:convexsplitflattening} looks like
another version of the smooth multipartite convex split lemma proved
in \cite{Sen:telescoping}, with the main difference being that
Lemma~\ref{lem:convexsplitflattening} states its
inner bounds stated in terms of $D^\epsilon_\infty$ instead of
$D^\epsilon_2$ in \cite{Sen:telescoping}. One may wonder about 
the need of giving another proof of smooth multipartite convex split,
and that too a much longer and more complicated proof. There seems to be 
no advantage in doing
so; the $D^\epsilon_\infty$ formulation is not any stronger. In fact
it is likely
slightly worse than the $D^\epsilon_2$ formulation. However, this 
exercise
paves the way to prove a non pairwise independent smooth multipartite
soft covering
lemma for classical quantum states in the next section.
This is where our more sophisticated
machinery wins out as proving such a result is beyond the reach of
telescoping. 
We will assume in this section, without loss of generality, 
that the density matrix $\rho^{XYM}$ in
the statement of convex split is normalised, by dividing by $\Tr\rho$.

We first state the version of the smooth multipartite convex split
lemma proved via flattening, tilting and augmentation smoothing.
\begin{lemma}
\label{lem:convexsplitflattening}
Let $\rho^{XYM}$ be a subnormalised density matrix.
Let $\alpha^X$, $\beta^Y$ be normalised
density matrices such that $\supp(\rho^X) \leq \supp(\alpha^X)$ and 
$\supp(\rho^Y) \leq \supp(\beta^Y)$. 
Define the $A$-fold tensor product states
$
\alpha^{X^A} := (\alpha^X)^{\otimes A}, 
\beta^{Y^B} := (\beta^Y)^{\otimes B}, 
$
and the $(A-1)$-fold tensor product states $\alpha^{X^{-a}}$, 
$\beta^{Y^{-b}}$ for any $a \in [A]$, $b \in [B]$.
Define the {\em convex split state}
\[
\sigma^{X^A Y^B M} :=
(AB)^{-1} 
\sum_{a=1}^A \sum_{b=1}^B
\rho^{X_a Y_b M} \otimes \alpha^{X^{-a}} \otimes \beta^{Y^{-b}},
\]
and the fully decoupled state 
\[
\tau^{X^A Y^B M} := \alpha^{X^A} \otimes \beta^{Y^B} \otimes \rho^M.
\]
Suppose
\begin{eqnarray*}
\log A 
& > &
D^\epsilon_\infty(\rho^{XM} \| \alpha^X \otimes \rho^M) +
\log \epsilon^{-1/20}, \\
\log B 
& > &
D^\epsilon_\infty(\rho^{YM} \| \beta^Y \otimes \rho^M) +
\log \epsilon^{-1/20}, \\
\log A + \log B 
& > &
D^\epsilon_\infty(\rho^{XYM} \| \alpha^X \otimes \beta^Y \otimes \rho^M) +
\log \epsilon^{-1/20}.
\end{eqnarray*}
Then,
$
\|\sigma^{X^A Y^B M} - \tau^{X^A Y^B M}\|_1 < 10 \epsilon^{1/64} (\Tr\rho).
$
\end{lemma}

We now proceed with the proof of Lemma~\ref{lem:convexsplitflattening}.
From the definitions of smooth $D^\epsilon_\infty(\cdot \| \cdot)$,
we have subnormalised density matrices $\rho^{XYM}(3)$, $\rho^{XM}(1)$
and $\rho^{YM}(2)$ satisfying
\begin{equation}
\label{eq:convsplit1}
\|\rho^{XYM} - \rho^{XYM}(3)\|_1 \leq \epsilon, ~~
\|\rho^{XM} - \rho^{XM}(1)\|_1 \leq \epsilon, ~~
\|\rho^{YM} - \rho^{YM}(2)\|_1 \leq \epsilon, 
\end{equation}
such that
\begin{equation}
\label{eq:convsplit2}
\begin{array}{rcl}
\rho^{XYM}(3)
& \leq &
2^{
D^\epsilon_\infty(\rho^{XYM} \| \alpha^X \otimes \beta^Y \otimes \rho^M)
} (\alpha^X \otimes \beta^Y \otimes \rho^M), \\
\rho^{XM}(1)
& \leq &
2^{D^\epsilon_\infty(\rho^{XM} \| \alpha^X \otimes \rho^M)}
(\alpha^X \otimes \rho^M), \\
\rho^{YM}(2)
& \leq &
2^{D^\epsilon_\infty(\rho^{YM} \| \beta^Y \otimes \rho^M)}
(\beta^Y \otimes \rho^M).
\end{array}
\end{equation}

Apply the flattening CPTP superoperator to get the normalised 
density matrix
\[
\rho^{L_X X L_Y Y L_M M} :=
(\cF_\alpha^{X \rightarrow L_X X} \otimes
 \cF_\beta^{Y \rightarrow L_Y Y} \otimes \cF_\rho^{M \rightarrow L_M M})
(\rho^{XYM}).
\]
Observe that 
\[
\rho^{L_X X L_M M} =
(\cF_\alpha^{X \rightarrow L_X X} \otimes \cF_\rho^{M \rightarrow L_M M})
(\rho^{XM}), ~~
\rho^{L_Y Y L_M M} =
(\cF_\beta^{Y \rightarrow L_Y Y} \otimes \cF_\rho^{M \rightarrow L_M M})
(\rho^{YM}).
\]
Define $\rho^{L_X X L_Y Y L_M M}(3)$, $\rho^{L_X X L_M M}(1)$,
$\rho^{L_Y Y L_M M}(2)$ via flattening in the natural fashion. Then
we use Equation~\ref{eq:convsplit1} and conclude,
\begin{equation}
\label{eq:convsplit3}
\begin{array}{c}
\|\rho^{L_X X L_Y Y L_M M} - \rho^{L_X X L_Y Y L_M M}(3)\|_1 
< \epsilon, \\
\|\rho^{L_X X L_M M} - \rho^{L_X X L_M M}(1)\|_1 
< \epsilon, ~~
\|\rho^{L_Y Y L_M M} - \rho^{L_Y Y L_M M}(2)\|_1 
< \epsilon,
\end{array}
\end{equation}
since a CPTP superoperator cannot increase the trace distance.

By Lemma~\ref{lem:flatteningdenom} and Proposition~\ref{prop:flattening}, 
we have
\begin{equation}
\label{eq:convsplit4}
\begin{array}{c}
(1+\delta)^{-2} \frac{\one^{F'_\alpha}}{|F'_\alpha|} \leq
\cF_\alpha^{X \rightarrow L_X X}(\alpha^X) \leq
(1+\delta) \frac{\one^{F'_\alpha}}{|F'_\alpha|}, \\
(1+\delta)^{-2} \frac{\one^{F'_\beta}}{|F'_\beta|} \leq
\cF_\beta^{Y \rightarrow L_Y Y}(\beta^Y) \leq
(1+\delta) \frac{\one^{F'_\beta}}{|F'_\beta|},  \\
(1+\delta)^{-2} \frac{\one^{F'_\rho}}{|F'_\rho|}.
\cF_\rho^{M \rightarrow L_M M}(\rho^M) \leq
(1+\delta) \frac{\one^{F'_\rho}}{|F'_\rho|}.
\end{array}
\end{equation}
Because a CPTP superoperator respects the L\"{o}wner partial order on
Hermitian matrices, combining Equations~\ref{eq:convsplit2} and
\ref{eq:convsplit4} gives us
\begin{equation}
\label{eq:convsplit5}
\begin{array}{rcl}
\rho^{L_X X L_Y Y L_M M}(3)
& \leq &
(1+\delta)^3
2^{
D^\epsilon_\infty(\rho^{XYM} \| \alpha^X \otimes \beta^Y \otimes \rho^M)
} 
\frac{\one^{F'_\alpha} \otimes \one^{F'_\beta} \otimes \one^{F'_\rho}}
     {|F'_\alpha| |F'_\beta| |F'_\rho|}, \\
\rho^{L_X X L_M M}(1)
& \leq &
(1+\delta)^2
2^{D^\epsilon_\infty(\rho^{XM} \| \alpha^X \otimes \rho^M)}
\frac{\one^{F'_\alpha} \otimes \one^{F'_\rho}} {|F'_\alpha| |F'_\rho|}, \\
\rho^{L_Y Y L_M M}(2)
& \leq &
(1+\delta)^2
2^{D^\epsilon_\infty(\rho^{YM} \| \beta^Y \otimes \rho^M)}
\frac{\one^{F'_\beta} \otimes \one^{F'_\rho}} {|F'_\beta| |F'_\rho|}.
\end{array}
\end{equation}
In particular,
\[
\supp(\rho^{L_X X L_Y Y L_M M}(3) \leq 
F'_\alpha \otimes F'_\beta \otimes F'_\rho, ~
\supp(\rho^{L_X X L_M M}(1) \leq 
F'_\alpha \otimes F'_\rho, ~
\supp(\rho^{L_Y Y L_M M}(2) \leq 
F'_\beta \otimes F'_\rho.
\]

Applying Proposition~\ref{prop:projsmoothing} to
Equations~\ref{eq:convsplit3}, \ref{eq:convsplit5} tells us that 
there exists orthogonal
projectors $\Pi^{L_X X L_Y Y L_M M}(3)$, $\Pi^{L_X X L_M M}(1)$,
$\Pi^{L_Y Y L_M M}(2)$ such that
\begin{equation}
\label{eq:convsplit6}
\begin{array}{rcl}
\Tr[\Pi^{L_X X L_Y Y L_M M}(3) \rho^{L_X X L_Y Y L_M M}] 
& \geq & 
1 - \sqrt{\epsilon}, \\
\Tr[\Pi^{L_X X L_M M}(1) \rho^{L_X X L_M M}] 
& \geq & 
1 - \sqrt{\epsilon}, \\
\Tr[\Pi^{L_Y Y L_M M}(2) \rho^{L_Y Y L_M M}] 
& \geq & 
1 - \sqrt{\epsilon}, \\
\Pi^{L_X X L_Y Y L_M M}(3) \rho^{L_X X L_Y Y L_M M} 
& = &
\rho^{L_X X L_Y Y L_M M} \Pi^{L_X X L_Y Y L_M M}(3), \\
\Pi^{L_X X L_M M}(1) \rho^{L_X X L_M M}
& = &
\rho^{L_X X L_M M} \Pi^{L_X X L_M M}(1), \\
\Pi^{L_Y Y L_M M}(2) \rho^{L_Y Y L_M M}
& = &
\rho^{L_Y Y L_M M} \Pi^{L_Y Y L_M M}(2), \\
\|\Pi^{L_X X L_Y Y L_M M}(3) \circ \rho^{L_X X L_Y Y L_M M}\|_\infty 
& \leq &
\frac{
(1+\delta)^3 (1 + 2\sqrt{\epsilon})
2^{
D^\epsilon_\infty(\rho^{XYM} \| \alpha^X \otimes \beta^Y \otimes \rho^M)
}}{|F'_\alpha| |F'_\beta| |F'_\rho|}, \\
\|\Pi^{L_X X L_M M}(1) \circ \rho^{L_X X L_M M}\|_\infty 
& \leq &
\frac{
(1+\delta)^2 (1 + 2\sqrt{\epsilon})
2^{
D^\epsilon_\infty(\rho^{XM} \| \alpha^X \otimes \rho^M)
}}{|F'_\alpha| |F'_\rho|}, \\
\|\Pi^{L_Y Y L_M M}(2) \circ \rho^{L_Y Y L_M M}\|_\infty 
& \leq &
\frac{
(1+\delta)^2 (1 + 2\sqrt{\epsilon})
2^{
D^\epsilon_\infty(\rho^{YM} \| \beta^Y \otimes \rho^M)
}}{|F'_\beta| |F'_\rho|}.
\end{array}
\end{equation}

Propositions \ref{prop:flattening} and
\ref{prop:fidelitypreserving}, and Lemma~\ref{lem:flatteningdenom} 
taken together allow us to prove Lemma~\ref{lem:flatteningfidelity}
below. Lemma~\ref{lem:flatteningfidelity}
will be used as the first step in the proof of 
Lemma~\ref{lem:convexsplitflattening}.
We note that $\delta$ can be taken as small as we please
at the expense of increasing $|F'_\sigma|$. Essentially, $\delta$ can be
treated as a free parameter.
\begin{lemma}
\label{lem:flatteningfidelity}
Let $A$, $B$ be positive integers.
Let $\alpha^X$, $\beta^Y$, $\rho^M$, $\sigma^{X^A Y^B M}$
be normalised density matrices. 
Define the normalised density matrix
$
\tau^{X^A Y^B M} := 
(\alpha^X)^{\otimes A} \otimes
(\beta^Y)^{\otimes B} \otimes
\rho^M.
$ 
Then,
\begin{eqnarray*}
\lefteqn{
F(((\cF_\alpha^{X \rightarrow L_X X})^{\otimes A} \otimes 
   (\cF_\beta^{Y \rightarrow L_Y Y})^{\otimes B} \otimes
   \cF_\rho^{M \rightarrow L_M M})(\sigma^{X^A Y^B M}),
} \\
&  &
~~~~~~~
  ((\cF_\alpha^{X \rightarrow L_X X})^{\otimes A} \otimes 
   (\cF_\beta^{Y \rightarrow L_Y Y})^{\otimes B} \otimes
   \cF_\rho^{M \rightarrow L_M M})(\tau^{X^A Y^B M})) \\
& = &
F(\sigma^{X^A Y^B M}, \tau^{X^A Y^B M}), \\
\lefteqn{
(1+\delta)^{-(A+B+1)}
\left(
\left(\frac{\one^{F'_\alpha}}{|F'_\alpha|}\right)^{\otimes A} \otimes 
\left(\frac{\one^{F'_\beta}}{|F'_\beta|}\right)^{\otimes B} \otimes 
\frac{\one^{F'_\rho}}{|F'_\rho|}
\right)
} \\
& \leq &
((\cF_\alpha^{X \rightarrow L_X X})^{\otimes A} \otimes 
 (\cF_\beta^{Y \rightarrow L_Y Y})^{\otimes B} \otimes
 \cF_\rho^{M \rightarrow L_M M})(\tau^{X^A Y^B M}) \\
& \leq &
(1+\delta)^{A+B+1}
\left(
\left(\frac{\one^{F'_\alpha}}{|F'_\alpha|}\right)^{\otimes A} \otimes 
\left(\frac{\one^{F'_\beta}}{|F'_\beta|}\right)^{\otimes B} \otimes 
\frac{\one^{F'_\rho}}{|F'_\rho|}
\right),
\end{eqnarray*}
where $L_X$, $L_Y$, $L_M$ are the additional Hilbert spaces required for 
the respective flattening superoperators.
\end{lemma}

In order to prepare for the proof of 
Lemma~\ref{lem:convexsplitflattening},
we will assume that the states $\sigma^{X^A Y^B M}$ and $\tau^{X^A Y^B M}$
have been flattened as in Lemma~\ref{lem:flatteningfidelity}, where
the positive integers $A$, $B$ come from the statement of 
Lemma~\ref{lem:convexsplitflattening}.
We will call the flattened states
$\sigma^{(L_X X)^A (L_Y Y)^B L_M M}$ and 
$\tau^{(L_X X)^A (L_Y Y)^B L_M M}$. We will choose the 
positive integers $F_\alpha$, $F_\beta$, $F_\rho$ large enough, and
correspondingly the dimensions of $L_X$, $L_Y$, $L_M$ large enough so
that $\delta \leq \frac{\sqrt{\epsilon}}{2(A+B+1)}$. Also to make the
notation lighter, henceforth we will always work with the flattened
states and the Hilbert spaces $L_X X$, $L_Y Y$ and $L_M M$. Hence from
now on we redefine $X \equiv L_X X$, $Y \equiv L_Y Y$ and 
$M \equiv L_M M$. Under this redefinition, we restate 
Equation~\ref{eq:convsplit6} as 
Equation~\ref{eq:convsplit7} below. That is, we can assume that there
are projectors $\Pi^{XYM}(3)$, $\Pi^{XM}(1)$ and $\Pi^{YM}(2)$ such
that
\begin{equation}
\label{eq:convsplit7}
\begin{array}{rcl}
\Tr[\Pi^{X Y M}(3) \rho^{X Y M}] 
& \geq & 
1 - \sqrt{\epsilon}, \\
\Tr[\Pi^{X M}(1) \rho^{X M}] 
& \geq & 
1 - \sqrt{\epsilon}, \\
\Tr[\Pi^{Y M}(2) \rho^{Y M}] 
& \geq & 
1 - \sqrt{\epsilon}, \\
\Pi^{X Y M}(3) \rho^{X Y M} 
& = &
\rho^{X Y M} \Pi^{X Y M}(3), \\
\Pi^{X M}(1) \rho^{X M}
& = &
\rho^{X M} \Pi^{X M}(1), \\
\Pi^{Y M}(2) \rho^{Y M}
& = &
\rho^{Y M} \Pi^{Y M}(2), \\
\|\Pi^{X Y M}(3) \rho^{X Y M}\|_\infty 
& \leq &
\frac{
(1 + 3\sqrt{\epsilon})
2^{
D^\epsilon_\infty(\rho^{XYM} \| \alpha^X \otimes \beta^Y \otimes \rho^M)
}}{|F'_\alpha| |F'_\beta| |F'_\rho|}, \\
\|\Pi^{X M}(1) \rho^{X M}\|_\infty 
& \leq &
\frac{
(1 + 3\sqrt{\epsilon})
2^{
D^\epsilon_\infty(\rho^{XM} \| \alpha^X \otimes \rho^M)
}}{|F'_\alpha| |F'_\rho|}, \\
\|\Pi^{Y M}(2) \rho^{Y M}\|_\infty 
& \leq &
\frac{
(1 + 3\sqrt{\epsilon})
2^{
D^\epsilon_\infty(\rho^{YM} \| \beta^Y \otimes \rho^M)
}}{|F'_\beta| |F'_\rho|} \\
F'_\alpha
& \leq &
X, ~~
F'_\beta
\;\leq\;
Y, ~~
F'_\rho
\;\leq\;
M, \\
(1 - \sqrt{\epsilon}) \frac{\one^{F'_\alpha}}{|F'_\alpha|}
& \leq &
\alpha^X
\;\leq\;
(1 + \sqrt{\epsilon}) \frac{\one^{F'_\alpha}}{|F'_\alpha|}, \\
(1 - \sqrt{\epsilon}) \frac{\one^{F'_\beta}}{|F'_\beta|}
& \leq &
\beta^Y
\;\leq\;
(1 + \sqrt{\epsilon}) \frac{\one^{F'_\beta}}{|F'_\beta|}, \\
(1 - \sqrt{\epsilon}) \frac{\one^{F'_\rho}}{|F'_\rho|}
& \leq &
\rho^M
\;\leq\;
(1 + \sqrt{\epsilon}) \frac{\one^{F'_\rho}}{|F'_\rho|},\\
(1 - \sqrt{\epsilon})
\frac{(\one^{F'_\alpha})^{\otimes A} \otimes 
      (\one^{F'_\beta})^{\otimes B} \otimes 
      (\one^{F'_\rho})
     }{|F'_\alpha|^A |F'_\beta|^B |F'_\rho|}
& \leq &
\tau^{X^A Y^B M}
\;\leq\;
(1 + \sqrt{\epsilon})
\frac{(\one^{F'_\alpha})^{\otimes A} \otimes 
      (\one^{F'_\beta})^{\otimes B} \otimes 
      (\one^{F'_\rho})
     }{|F'_\alpha|^A |F'_\beta|^B |F'_\rho|}.
\end{array}
\end{equation}
In particular, 
\[
\supp(\rho^{X Y M}) \leq
F'_\alpha \otimes F'_\beta \otimes F'_\rho, ~~
\supp(\sigma^{X^A Y^B M}) \leq
(F'_\alpha)^{\otimes A} \otimes (F'_\beta)^{\otimes B} \otimes
F'_\rho =
\supp(\tau^{X^A Y^B M}).
\]
To finish the proof of Lemma~\ref{lem:convexsplitflattening}, we 
only need to show 
\begin{equation}
\label{eq:convsplit8}
\|\sigma^{X^A Y^B M} - \tau^{X^A Y^B M}\|_1 \leq 20 \epsilon^{1/32},
\end{equation}
because of Lemma~\ref{lem:flatteningfidelity} and
Fact~\ref{fact:fidelitytracedist}.

In order to prove the above inequality, we augment $X$ to $L_X X$,
$Y$ to $L_Y Y$ as in Section~\ref{subsec:tilting}. Note that 
$L_X$, $L_Y$ are the
additional Hilbert spaces required for augmentation and have nothing
to do the $L_X$, $L_Y$ used for flattening earlier. In fact, the
current definition of the Hilbert space $X$ is actually the old $X$
tensored with the old flattening $L_X$; same for the current definition
of $Y$. Recall that the new spaces $L_X$, $L_Y$ required for augmentation
have the same dimension $L$. The value of $L$ will be chosen later,
and it will be sufficiently large for our later purposes. Keeping
augmentation in mind, we define
\[
\alpha^{L_X X} := \frac{\one^{L_X}}{L} \otimes \alpha^X, ~
\beta^{L_Y Y} := \frac{\one^{L_Y}}{L} \otimes \beta^Y, ~~
\rho^{L_X X L_Y Y M} := 
\frac{\one^{L_X}}{L} \otimes \frac{\one^{L_Y}}{L} \otimes 
\rho^{XYM}.
\]
Then, the $A$-fold tensor product state $\alpha^{(L_X X)^A}$ and
the $(|A|-1)$-fold tensor product state $\alpha^{(L_X X)^{-a}}$ for
any $a \in [A]$ can be defined as before in the natural fashion.
A similar comment holds for $\beta^{(L_Y Y)^B}$ and 
$\beta^{(L_Y Y)^{-b}}$ for any $b \in [B]$.

As in Section~\ref{subsec:tilting}, we enlarge the
Hilbert space $M$ into a larger space $\hM$. Because of this containment,
the state
$\rho^{L_X X L_Y Y \hM}$ is identical to $\rho^{L_X X L_Y Y M}$; which
notation we use depends on the application. Define the normalised
state $\hrho^{L_X X L_Y Y \hM}$ and subnormalised state
$\hhrho^{L_X X L_Y Y \hM}$ as in  Equation~\ref{eq:hrhohhrho} via
the projectors $\Pi^{XYM}(3)$, 
$\Pi^{XYM}(2) := \one^X \otimes \Pi^{YM}(2)$,
$\Pi^{XYM}(1) := \one^Y \otimes \Pi^{XM}(1)$.
Define
\begin{equation}
\label{eq:convsplit9}
\begin{array}{rcl}
\sigma^{(L_X X)^A (L_Y Y)^B \hM}
& := &
\sigma^{(L_X X)^A (L_Y Y)^B M} \\
& = &
(AB)^{-1} 
\sum_{a=1}^A \sum_{b=1}^B
\rho^{(L_X X)_a (L_Y Y)_b M} \otimes 
\alpha^{(L_X X)^{-a}} \otimes \beta^{(L_Y Y)^{-b}}, \\
\hsigma^{(L_X X)^A (L_Y Y)^B \hM}
& := &
(AB)^{-1} 
\sum_{a=1}^A \sum_{b=1}^B
\hrho^{(L_X X)_a (L_Y Y)_b \hM} \otimes 
\alpha^{(L_X X)^{-a}} \otimes \beta^{(L_Y Y)^{-b}}, \\
\hhsigma^{(L_X X)^A (L_Y Y)^B \hM}
& := &
(AB)^{-1} 
\sum_{a=1}^A \sum_{b=1}^B
\hhrho^{(L_X X)_a (L_Y Y)_b \hM} \otimes 
\alpha^{(L_X X)^{-a}} \otimes \beta^{(L_Y Y)^{-b}}, \\
\tau^{(L_X X)^A (L_Y Y)^B \hM}
& := &
\tau^{(L_X X)^A (L_Y Y)^B M}
\;=\;
\alpha^{(L_X X)^A} \otimes
\beta^{(L_Y Y)^B} \otimes
\rho^{L_M M}.
\end{array}
\end{equation}
Observe that 
\begin{eqnarray*}
\|\sigma^{(L_X X)^A (L_Y Y)^B \hM} - \tau^{(L_X X)^A (L_Y Y)^B \hM}\|_1 
& \equiv &
\|\sigma^{(L_X X)^A (L_Y Y)^B M} - \tau^{(L_X X)^A (L_Y Y)^B M}\|_1 \\
& = &
\|\sigma^{X^A Y^B M} - \tau^{X^A Y^B M}\|_1.
\end{eqnarray*}
Note that $\sigma^{(L_X X)^A (L_Y Y)^B \hM}$,
$\tau^{(L_X X)^A (L_Y Y)^B \hM}$, 
$\hsigma^{(L_X X)^A (L_Y Y)^B \hM}$ are normalised states and
$\hhsigma^{(L_X X)^A (L_Y Y)^B \hM}$ is subnormalised.
So to prove Lemma~\ref{lem:convexsplitflattening}, we just need to
show that
\begin{equation}
\label{eq:convsplit10}
\|\sigma^{(L_X X)^A (L_Y Y)^B \hM} - \tau^{(L_X X)^A (L_Y Y)^B \hM}\|_1 
\leq 17 \epsilon^{1/32},
\end{equation}
because of Equations~\ref{eq:convsplit8}, \ref{eq:convsplit9}.

Define the tensor product spaces
\[
L_X \otimes F'_\alpha =: L_X F'_\alpha \leq L_X X, ~~
L_Y \otimes F'_\beta =: L_Y F'_\beta \leq L_Y Y.
\]
Observe that
\begin{equation}
\label{eq:convsplit11}
\begin{array}{rcl}
(1 - \sqrt{\epsilon})
\frac{(\one^{L_X F'_\alpha})^{-a} \otimes 
      (\one^{L_Y F'_\beta})^{-b}
     }{(L |F'_\alpha|)^{A-1} (L |F'_\beta|)^{B-1}}
& \leq &
\alpha^{X^{-a}} \otimes \beta^{Y^{-b}} \\
& \leq &
(1 + \sqrt{\epsilon})
\frac{(\one^{L_X F'_\alpha})^{-a} \otimes 
      (\one^{L_Y F'_\beta})^{-b} 
     }{(L |F'_\alpha|)^{A-1} (L |F'_\beta|)^{B-1}}, \\
(1 - \sqrt{\epsilon})
\frac{(\one^{L_X F'_\alpha})^{\otimes A} \otimes 
      (\one^{L_Y F'_\beta})^{\otimes B} \otimes 
      (\one^{F'_\rho})
     }{(L |F'_\alpha|)^A (L |F'_\beta|)^B |F'_\rho|}
& \leq &
\tau^{(L_X X)^A (L_Y Y)^B \hM} \\
& \leq &
(1 + \sqrt{\epsilon})
\frac{(\one^{L_X F'_\alpha})^{\otimes A} \otimes 
      (\one^{L_Y F'_\beta})^{\otimes B} 
      (\one^{F'_\rho})
     }{(L |F'_\alpha|)^A (L |F'_\beta|)^B |F'_\rho|},
\end{array}
\end{equation}
for any $a \in [A]$, $b \in [B]$.

We now prove four lemmas which are crucially required in the following
arguments, leading to the proof of
Lemma~\ref{lem:convexsplitflattening}.
\begin{lemma}
\label{lem:hrhohhrho3}
\[
\Tr[\hrho^{L_X X L_Y Y \hM} \hhrho^{L_X X L_Y Y \hM}] \leq
\frac{(1+3\sqrt{\epsilon})
      2^{D^\epsilon_\infty(\rho^{XYM} \| 
			   \alpha^X \otimes \beta^Y \otimes \rho^M)}
     }{(L |F'_\alpha|) (L F'_\beta) |F'_\rho|}.
\]
\end{lemma}
\begin{proof}
Using Equations~\ref{eq:hrhohhrho}, \ref{eq:approxintersection},
\ref{eq:convsplit7}, we get
\begin{eqnarray*}
\lefteqn{
\Tr[\hrho^{L_X X L_Y Y \hM} \hhrho^{L_X X L_Y Y \hM}]
} \\
& = &
\Tr[\hrho^{L_X X L_Y Y \hM} 
    (\hPi^{L_X X L_Y Y \hM} \circ \rho^{L_X X L_Y Y \hM})] \\
& = &
\Tr[T^{L_X X L_Y Y M \rightarrow L_X X L_Y Y \hM}_{\epsilon^{1/4}}(
	\rho^{L_X X L_Y Y M}) \\
&  &
~~~~~~~~
    ((\one^{L_X X L_Y Y \hM} - 
      T^{L_X X L_Y Y M \rightarrow L_X X L_Y Y \hM}_{\epsilon^{1/4}}(
         \one^{L_X L_Y} \otimes (\one^{XYM} - \Pi^{XYM}(3)))
     ) \\
&   &
~~~~~~~~~~~~~~~
{} \circ (\hPi^{L_X X L_Y Y \hM} \circ \rho^{L_X X L_Y Y \hM})
    )
   ] \\
& = &
\Tr[(T^{L_X X L_Y Y M \rightarrow L_X X L_Y Y \hM}_{\epsilon^{1/4}}(
	\one^{L_X X L_Y Y M}) \circ
     T^{L_X X L_Y Y M \rightarrow L_X X L_Y Y \hM}_{\epsilon^{1/4}}(
	\rho^{L_X X L_Y Y M})
    ) \\
&  &
~~~~~~~~
    ((\one^{L_X X L_Y Y \hM} - 
      T^{L_X X L_Y Y M \rightarrow L_X X L_Y Y \hM}_{\epsilon^{1/4}}(
         \one^{L_X L_Y} \otimes (\one^{XYM} - \Pi^{XYM}(3)))
     ) \\
&   &
~~~~~~~~~~~~~~~
{} \circ \hhrho^{L_X X L_Y Y \hM}
    )
   ] \\
& = &
\Tr[T^{L_X X L_Y Y M \rightarrow L_X X L_Y Y \hM}_{\epsilon^{1/4}}(
	\rho^{L_X X L_Y Y M}) \\
&  &
~~~~~~~~
    (T^{L_X X L_Y Y M \rightarrow L_X X L_Y Y \hM}_{\epsilon^{1/4}}(
	\one^{L_X X L_Y Y M}) - {} \\
&  &
~~~~~~~~~~~~~~~~
      T^{L_X X L_Y Y M \rightarrow L_X X L_Y Y \hM}_{\epsilon^{1/4}}(
         \one^{L_X L_Y} \otimes (\one^{XYM} - \Pi^{XYM}(3)))
    ) \circ \hhrho^{L_X X L_Y Y \hM}
    )
   ] \\
& = &
\Tr[T^{L_X X L_Y Y M \rightarrow L_X X L_Y Y \hM}_{\epsilon^{1/4}}(
	(\frac{\one^{L_X L_Y}}{L^2}) \otimes \rho^{XYM}) \\
&  &
~~~~~~~~~~
    (T^{L_X X L_Y Y M \rightarrow L_X X L_Y Y \hM}_{\epsilon^{1/4}}(
         \one^{L_X L_Y} \otimes \Pi^{XYM}(3)) \circ 
           \hhrho^{L_X X L_Y Y \hM}
    )
   ] \\
& = &
\Tr[T^{L_X X L_Y Y M \rightarrow L_X X L_Y Y \hM}_{\epsilon^{1/4}}(
	(\frac{\one^{L_X L_Y}}{L^2}) \otimes 
	(\Pi^{XYM}(3) \circ \rho^{XYM})) 
    \hhrho^{L_X X L_Y Y \hM}
    )
   ] \\
& \leq &
\|T^{L_X X L_Y Y M \rightarrow L_X X L_Y Y \hM}_{\epsilon^{1/4}}(
	(\frac{\one^{L_X L_Y}}{L^2}) \otimes 
	(\Pi^{XYM}(3) \circ \rho^{XYM}))
\|_\infty \cdot \|\hhrho^{L_X X L_Y Y \hM}\|_1 \\
& = &
\|(\frac{\one^{L_X L_Y}}{L^2}) \otimes 
  (\Pi^{XYM}(3) \circ \rho^{XYM}) \|_\infty 
\cdot \|\hhrho^{L_X X L_Y Y \hM}\|_1 \\
& = &
L^{-2} \cdot \|\Pi^{XYM}(3) \circ \rho^{XYM}\|_\infty 
\cdot \|\hhrho^{L_X X L_Y Y \hM}\|_1 
\;\leq\;
\frac{(1+3\sqrt{\epsilon}) 
      2^{D^\epsilon_\infty(\rho^{XYM} \|
			   \alpha^X \otimes \beta^Y \otimes \rho^M)}
     }{(L |F'_\alpha|) (L |F'_\beta|) |F'_\rho|}.
\end{eqnarray*}
Above, we used the fact that
\[
\one^{L_X X L_Y Y \hM} - 
T^{L_X X L_Y Y M \rightarrow L_X X L_Y Y \hM}_{\epsilon^{1/4}}(
     \one^{L_X L_Y} \otimes (\one^{XYM} - \Pi^{XYM}(3))) \geq
\hPi^{L_X X L_Y Y \hM}
\]
in the second equality,
\begin{eqnarray*}
\lefteqn{
(\one^{L_X X L_Y Y \hM} - 
T^{L_X X L_Y Y M \rightarrow L_X X L_Y Y \hM}_{\epsilon^{1/4}}(
     \one^{L_X L_Y} \otimes (\one^{XYM} - \Pi^{XYM}(3)))) 
} \\
&   &
~~~~~~
(T^{L_X X L_Y Y M \rightarrow L_X X L_Y Y \hM}_{\epsilon^{1/4}}(
	\one^{L_X X L_Y Y M})
)\\
& = &
T^{L_X X L_Y Y M \rightarrow L_X X L_Y Y \hM}_{\epsilon^{1/4}}(
	\one^{L_X X L_Y Y M}) \\
& &
~~~
{} - 
      T^{L_X X L_Y Y M \rightarrow L_X X L_Y Y \hM}_{\epsilon^{1/4}}(
         \one^{L_X L_Y} \otimes (\one^{XYM} - \Pi^{XYM}(3)))
\end{eqnarray*}
in the fourth equality, 
$T^{L_X X L_Y Y M \rightarrow L_X X L_Y Y \hM}_{\epsilon^{1/4}}$ is
an isometry in the sixth and seventh equalities. 
The proof of the lemma is now complete.
\end{proof}

\begin{lemma}
\label{lem:hrhohhrho1}
\[
\Tr[\hrho^{L_X X \hM} \hhrho^{L_X X \hM}] \leq
\frac{4 \epsilon^{1/8}}{L^{3/2}} +
\frac{(1+3\sqrt{\epsilon})
      2^{D^\epsilon_\infty(\rho^{XM} \| \alpha^X \otimes \rho^M)}
     }{(L |F'_\alpha|) |F'_\rho|}.
\]
\end{lemma}
\begin{proof}
Using Equations~\ref{eq:hrhohhrho}, \ref{eq:approxintersection},
\ref{eq:convsplit7} and Lemma~\ref{lem:smoothing}, we get
\begin{eqnarray*}
\lefteqn{
\Tr[\hrho^{L_X X \hM} \hhrho^{L_X X \hM}]
} \\
& = &
\Tr[(\hrho^{L_X X \hM} \otimes \one^{L_Y Y}) \hhrho^{L_X X L_Y Y \hM}] \\
& = &
\Tr[((T^{L_X X M \rightarrow L_X X \hM}_{\epsilon^{1/4}}(\rho^{L_X X M}))
     \otimes \one^{L_Y Y}) \hhrho^{L_X X L_Y Y \hM}
   ] \\
&  &
~~~
{} +
\Tr[((\hrho^{L_X X \hM} - 
      T^{L_X X M \rightarrow L_X X \hM}_{\epsilon^{1/4}}(\rho^{L_X X M})
     )\otimes \one^{L_Y Y}) \hhrho^{L_X X L_Y Y \hM}
   ] \\
& = &
\Tr[((T^{L_X X M \rightarrow L_X X \hM}_{\epsilon^{1/4}}(\rho^{L_X X M}))
     \otimes \one^{L_Y Y}) \\
&  &
~~~~~~~~
    ((\one^{L_X X L_Y Y \hM} - 
      (T^{L_X X M \rightarrow L_X X \hM}_{\epsilon^{1/4}} \otimes
       \one^{L_Y Y})(\one^{L_X L_Y} \otimes 
	             (\one^{XYM} - \one^Y \otimes \Pi^{XM}(1)))
     ) \\
&   &
~~~~~~~~~~~~~~~
{} \circ (\hPi^{L_X X L_Y Y \hM} \circ \rho^{L_X X L_Y Y \hM})
    )
   ] \\
&  &
~~~
{} +
\Tr[(\hrho^{L_X X \hM} - 
      T^{L_X X M \rightarrow L_X X \hM}_{\epsilon^{1/4}}(\rho^{L_X X M})
    ) \hhrho^{L_X X \hM}
   ] \\
& = &
\Tr[(((T^{L_X X M \rightarrow L_X X \hM}_{\epsilon^{1/4}}(
	\rho^{L_X X M})
      ) \\
&  &
~~~~~~~~
      (\one^{L_X X \hM} - 
       T^{L_X X M \rightarrow L_X X \hM}_{\epsilon^{1/4}}(
          \one^{L_X} \otimes (\one^{XM} - \Pi^{XM}(1)))
      )
     ) \otimes \one^{L_Y Y}
    ) \circ \hhrho^{L_X X L_Y Y \hM}
   ] \\
&  &
~~~
{} +
\Tr[(\hrho^{L_X X \hM} - 
      T^{L_X X M \rightarrow L_X X \hM}_{\epsilon^{1/4}}(\rho^{L_X X M})
    ) \hhrho^{L_X X \hM}
   ] \\
& = &
\Tr[((T^{L_X X M \rightarrow L_X X \hM}_{\epsilon^{1/4}}(
	\rho^{L_X X M})
     ) \\
&  &
~~~~~~~~
     (\one^{L_X X \hM} - 
       T^{L_X X M \rightarrow L_X X \hM}_{\epsilon^{1/4}}(
          \one^{L_X} \otimes (\one^{XM} - \Pi^{XM}(1)))
     )
    ) \circ \hhrho^{L_X X \hM}
   ] \\
&  &
~~~
{} +
\Tr[(\hrho^{L_X X \hM} - 
      T^{L_X X M \rightarrow L_X X \hM}_{\epsilon^{1/4}}(\rho^{L_X X M})
    ) \hhrho^{L_X X \hM}
   ] \\
& \leq &
\Tr[((T^{L_X X M \rightarrow L_X X \hM}_{\epsilon^{1/4}}(
	\rho^{L_X X M})
     ) \\
&  &
~~~~~~~~
     (\one^{L_X X \hM} - 
       T^{L_X X M \rightarrow L_X X \hM}_{\epsilon^{1/4}}(
          \one^{L_X} \otimes (\one^{XM} - \Pi^{XM}(1)))
     )
    ) \circ \hhrho^{L_X X \hM}
   ] \\
&  &
~~~
{} +
\|\hrho^{L_X X \hM} - 
  T^{L_X X M \rightarrow L_X X \hM}_{\epsilon^{1/4}}(\rho^{L_X X M})
\|_\infty \cdot \|\hhrho^{L_X X \hM}\|_1 \\
& \leq &
\Tr[((T^{L_X X M \rightarrow L_X X \hM}_{\epsilon^{1/4}}(
	\rho^{L_X X M})
     ) \\
&  &
~~~~~~~~
     (\one^{L_X X \hM} - 
       T^{L_X X M \rightarrow L_X X \hM}_{\epsilon^{1/4}}(
          \one^{L_X} \otimes (\one^{XM} - \Pi^{XM}(1)))
     )
    ) \circ \hhrho^{L_X X \hM}
   ] + 4 \epsilon^{1/8} L^{-3/2} \\
& = &
\Tr[((T^{L_X X M \rightarrow L_X X \hM}_{\epsilon^{1/4}}(
	\one^{L_X X M}) \circ
      T^{L_X X M \rightarrow L_X X \hM}_{\epsilon^{1/4}}(
	\rho^{L_X X M})
     ) \\
&  &
~~~~~~~~
     (\one^{L_X X \hM} - 
       T^{L_X X M \rightarrow L_X X \hM}_{\epsilon^{1/4}}(
          \one^{L_X} \otimes (\one^{XM} - \Pi^{XM}(1)))
     )
    ) \circ \hhrho^{L_X X \hM}
   ] + 4 \epsilon^{1/8} L^{-3/2} \\
& = &
\Tr[T^{L_X X M \rightarrow L_X X \hM}_{\epsilon^{1/4}}(\rho^{L_X X M}) \\
&  &
~~~~~~~~
    ((T^{L_X X M \rightarrow L_X X \hM}_{\epsilon^{1/4}}(\one^{L_X X M})
       - T^{L_X X M \rightarrow L_X X \hM}_{\epsilon^{1/4}}(
            \one^{L_X} \otimes (\one^{XM} - \Pi^{XM}(1)))
     ) \circ \hhrho^{L_X X \hM}
    )
   ] \\
&   &
~~~
{} + 4 \epsilon^{1/8} L^{-3/2} \\
& = &
\Tr[T^{L_X X M \rightarrow L_X X \hM}_{\epsilon^{1/4}}(
	(\frac{\one^{L_X}}{L}) \otimes \rho^{XM}) \\
&  &
~~~~~~~~
    ((T^{L_X X M \rightarrow L_X X \hM}_{\epsilon^{1/4}}(
            \one^{L_X} \otimes \Pi^{XM}(1))
     ) \circ \hhrho^{L_X X \hM}
    )
   ] + 4 \epsilon^{1/8} L^{-3/2} \\
& = &
\Tr[(T^{L_X X M \rightarrow L_X X \hM}_{\epsilon^{1/4}}(
	(\frac{\one^{L_X}}{L}) \otimes (\Pi^{XM}(1) \circ \rho^{XM})) 
     \hhrho^{L_X X \hM}
    )
   ] + 4 \epsilon^{1/8} L^{-3/2} \\
& \leq &
\|T^{L_X X M \rightarrow L_X X \hM}_{\epsilon^{1/4}}(
	(\frac{\one^{L_X}}{L}) \otimes (\Pi^{XM}(1) \circ \rho^{XM}))
\|_\infty \cdot \|\hhrho^{L_X X \hM}\|_1 + 4 \epsilon^{1/8} L^{-3/2} \\
& = &
\|(\frac{\one^{L_X}}{L}) \otimes (\Pi^{XM}(1) \circ \rho^{XM})\|_\infty 
\cdot \|\hhrho^{L_X X \hM}\|_1 + 4 \epsilon^{1/8} L^{-3/2} \\
& = &
L^{-1} \cdot \|\Pi^{XM}(1) \circ \rho^{XM}\|_\infty 
\cdot \|\hhrho^{L_X X \hM}\|_1 + 4 \epsilon^{1/8} L^{-3/2} \\
& \leq &
\frac{(1+3\sqrt{\epsilon}) 
      2^{D^\epsilon_\infty(\rho^{XM} \| \alpha^X \otimes \rho^M)}
     }{(L |F'_\alpha|) |F'_\rho|} + 4 \epsilon^{1/8} L^{-3/2}.
\end{eqnarray*}
Above, we used the fact that
\[
\one^{L_X X L_Y Y \hM} - 
T^{L_X X L_Y Y M \rightarrow L_X X L_Y Y \hM}_{\epsilon^{1/4}}(
     \one^{L_X L_Y} \otimes (\one^{XYM} - \one^Y \Pi^{XM}(1))) \geq
\hPi^{L_X X L_Y Y \hM}
\]
in the third equality,
\begin{eqnarray*}
\lefteqn{
(\one^{L_X X \hM} - 
 T^{L_X X M \rightarrow L_X X \hM}_{\epsilon^{1/4}}(
     \one^{L_X} \otimes (\one^{XM} - \Pi^{XM}(1)))
) 
(T^{L_X X M \rightarrow L_X X \hM}_{\epsilon^{1/4}}(\one^{L_X X M}))
} \\
& = &
T^{L_X X M \rightarrow L_X X \hM}_{\epsilon^{1/4}}(
	\one^{L_X X M}) - 
T^{L_X X M \rightarrow L_X X \hM}_{\epsilon^{1/4}}(
         \one^{L_X} \otimes (\one^{XM} - \Pi^{XM}(1)))
\end{eqnarray*}
in the seventh equality, 
$T^{L_X X M \rightarrow L_X X \hM}_{\epsilon^{1/4}}$ is
an isometry in the ninth and tenth equalities. 
The proof of the lemma is now complete.
\end{proof}

\begin{lemma}
\label{lem:hrhohhrho2}
\[
\Tr[\hrho^{L_Y Y \hM} \hhrho^{L_Y Y \hM}] \leq
\frac{4 \epsilon^{1/8}}{L^{3/2}} +
\frac{(1+3\sqrt{\epsilon})
      2^{D^\epsilon_\infty(\rho^{YM} \| \beta^Y \otimes \rho^M)}
     }{(L |F'_\beta|) |F'_\rho|}.
\]
\end{lemma}
\begin{proof}
Similar to proof of Lemma~\ref{lem:hrhohhrho1} above.
\end{proof}

\begin{lemma}
\label{lem:hrhohhrho}
\[
\Tr[\hrho^{\hM} \hhrho^{\hM}] \leq
\frac{8 \epsilon^{1/8}}{L^{1/2}} +
\frac{1+\sqrt{\epsilon}}{|F'_\rho|}.
\]
\end{lemma}
\begin{proof}
Using Equation~\ref{eq:convsplit7} and Lemma~\ref{lem:smoothing}, we get
\begin{eqnarray*}
\lefteqn{
\Tr[\hrho^{\hM} \hhrho^{\hM}] 
} \\
& = &
\Tr[\rho^{\hM} \hhrho^{\hM}] +
\Tr[(\hrho^{\hM} - \rho^{\hM}) \hhrho^{\hM}] 
\;\leq\;
\frac{1 + \sqrt{\epsilon}}{|F'_\rho|}
\Tr[\one^{F'_\rho} \hhrho^{\hM}] +
\|\hrho^{\hM} - \rho^{\hM}\|_\infty \cdot \|\hhrho^{\hM}\|_1 \\
& \leq &
\frac{1 + \sqrt{\epsilon}}{|F'_\rho|}
\Tr[\one^{\hM} \hhrho^{\hM}] +
8 \epsilon^{1/8} L^{-1/2} 
\;\leq\;
\frac{1 + \sqrt{\epsilon}}{|F'_\rho|} +
8 \epsilon^{1/8} L^{-1/2}.
\end{eqnarray*}
This completes the proof of the lemma.
\end{proof}

By Lemma~\ref{lem:tiltclose} and Equation~\ref{eq:convsplit9}, we have
\begin{equation}
\label{eq:convsplittilt}
\|\sigma^{(L_X X)^A (L_Y Y)^B \hM} - \hsigma^{(L_X X)^A (L_Y Y)^B \hM}\|_1 
 \leq 
\|\rho^{L_X X L_Y Y \hM} - \hrho^{L_X X L_Y Y \hM}\|_1 
=
2\sqrt{2} \cdot \epsilon^{1/8}.
\end{equation}
Hence by triangle inequality and Equation~\ref{eq:convsplit10}, in 
order to prove 
Lemma~\ref{lem:convexsplitflattening}, it suffices to show that
\begin{equation}
\label{eq:convsplit12}
\|\hsigma^{(L_X X)^A (L_Y Y)^B \hM} - \tau^{(L_X X)^A (L_Y Y)^B \hM}\|_1 
\leq 14 \epsilon^{1/32}.
\end{equation}

By Corollary~\ref{cor:triangleineq} and Equation~\ref{eq:convsplit9}, 
we have
\begin{eqnarray*}
\|\hsigma^{(L_X X)^A (L_Y Y)^B \hM} - 
  \hhsigma^{(L_X X)^A (L_Y Y)^B \hM}\|_1 
& \leq &
\|\hrho^{L_X X L_Y Y \hM} - \hhrho^{L_X X L_Y Y \hM}\|_1 
\;<\;
25 \cdot \epsilon^{1/8}.
\end{eqnarray*}
By Lemma~\ref{lem:asyml2}, there is a subnormalised state
$\hsigma^{'(L_X X)^A (L_Y Y)^B \hM}$ such that
\begin{equation}
\label{eq:convsplitasym}
\begin{array}{c}
\hsigma^{'(L_X X)^A (L_Y Y)^B \hM} \leq 
\hsigma^{(L_X X)^A (L_Y Y)^B \hM}, \\
\|\hsigma^{'(L_X X)^A (L_Y Y)^B \hM} - 
  \hsigma^{(L_X X)^A (L_Y Y)^B \hM}\|_1 < 5 \epsilon^{1/16}, \\
\|\hsigma^{'(L_X X)^A (L_Y Y)^B \hM}\|_2^2 \leq
(1 + 10\epsilon^{1/16})
\Tr[\hsigma^{(L_X X)^A (L_Y Y)^B \hM} 
    \hhsigma^{(L_X X)^A (L_Y Y)^B \hM}].
\end{array}
\end{equation}
Hence by triangle inequality and Equation~\ref{eq:convsplit12}, in 
order to prove 
Lemma~\ref{lem:convexsplitflattening}, it suffices to show that
\begin{equation}
\label{eq:convsplit13}
\|\hsigma^{'(L_X X)^A (L_Y Y)^B \hM} - \tau^{(L_X X)^A (L_Y Y)^B \hM}\|_1 
\leq 9 \epsilon^{1/32}.
\end{equation}
Note that
$\hsigma^{'(L_X X)^A (L_Y Y)^B \hM}$ is subnormalised.

Define $\Pi^{(L_X X)^A (L_Y Y)^B \hM}$ to be the orthogonal projector
from the ambient space $(L_X X)^A (L_Y Y)^B \hM$ onto
\[
\supp(\tau^{(L_X X)^A (L_Y Y)^B \hM}) =
(L_X \otimes F'_\alpha)^A \otimes (L_Y \otimes F'_\beta)^B \otimes 
F'_\rho \geq
\supp(\sigma^{(L_X X)^A (L_Y Y)^B \hM}).
\]
By Equations~\ref{eq:convsplitasym}, \ref{eq:convsplittilt}, we get
\begin{equation}
\label{eq:convsplit14}
\begin{array}{rcl}
\lefteqn{
\Tr[\Pi^{(L_X X)^A (L_Y Y)^B \hM} \tau^{(L_X X)^A (L_Y Y)^B \hM}]
} \\
& = &
\Tr[\tau^{(L_X X)^A (L_Y Y)^B \hM}]
\;=\;
1, \\
\lefteqn{
\Tr[\Pi^{(L_X X)^A (L_Y Y)^B \hM} \hsigma^{'(L_X X)^A (L_Y Y)^B \hM}]
} \\
& \geq &
\Tr[\Pi^{(L_X X)^A (L_Y Y)^B \hM} \sigma^{(L_X X)^A (L_Y Y)^B \hM}] -
\|\sigma^{(L_X X)^A (L_Y Y)^B \hM} - 
  \hsigma^{(L_X X)^A (L_Y Y)^B \hM}\|_1 \\
&  &
~~~~~~
{} -
\|\hsigma^{(L_X X)^A (L_Y Y)^B \hM} - 
  \hsigma^{'(L_X X)^A (L_Y Y)^B \hM}\|_1 \\
&   =  &
\Tr[\sigma^{(L_X X)^A (L_Y Y)^B \hM}] -
\|\sigma^{(L_X X)^A (L_Y Y)^B \hM} - 
  \hsigma^{(L_X X)^A (L_Y Y)^B \hM}\|_1 \\
&  &
~~~~~~
{} -
\|\hsigma^{(L_X X)^A (L_Y Y)^B \hM} - 
  \hsigma^{'(L_X X)^A (L_Y Y)^B \hM}\|_1 \\
& \geq &
1 - 2\sqrt{2} \cdot \epsilon^{1/8} - 5 \epsilon^{1/16}
\;\geq\;
1 - 8 \epsilon^{1/16}.
\end{array}
\end{equation}
By Proposition~\ref{prop:shavedCauchySchwarz} and
Equations~\ref{eq:convsplit13}, \ref{eq:convsplit14}, in order to prove
Lemma~\ref{lem:convexsplitflattening}, it suffices to show that
\[
\sqrt{\Tr[\Pi^{(L_X X)^A (L_Y Y)^B \hM}]} \cdot
\|\hsigma^{'(L_X X)^A (L_Y Y)^B \hM} - \tau^{(L_X X)^A (L_Y Y)^B \hM}\|_2 
< 6 \epsilon^{1/32}.
\]
Observe that
$
\Tr[\Pi^{(L_X X)^A (L_Y Y)^B \hM}] = 
L^{A+B} |F'_\alpha|^A |F'_\beta|^B |F'_\rho|.
$
Hence it suffices to show that
\begin{equation}
\label{eq:convsplit15}
\|\hsigma^{'(L_X X)^A (L_Y Y)^B \hM} - 
  \tau^{(L_X X)^A (L_Y Y)^B \hM}\|_2^2 < 
\frac{35 \epsilon^{1/16}}{L^{A+B} |F'_\alpha|^A |F'_\beta|^B |F'_\rho|}.
\end{equation}

The left hand side of the above inequality is
\begin{equation}
\label{eq:convsplit16}
\begin{array}{rcl}
\lefteqn{
\|\hsigma^{'(L_X X)^A (L_Y Y)^B \hM} -\tau^{(L_X X)^A (L_Y Y)^B \hM}\|_2^2 
} \\
& = &
\|\hsigma^{'(L_X X)^A (L_Y Y)^B \hM}\|_2^2 +
\|\tau^{(L_X X)^A (L_Y Y)^B \hM}\|_2^2 \\
& . &
~~~
{} -
2\Tr[\hsigma^{'(L_X X)^A (L_Y Y)^B \hM}
     \tau^{(L_X X)^A (L_Y Y)^B \hM}] \\
& \leq &
(1 + 10\epsilon^{1/16})
\Tr[\hsigma^{(L_X X)^A (L_Y Y)^B \hM} 
    \hhsigma^{(L_X X)^A (L_Y Y)^B \hM}] \\
& &
~~~~
{} + 
(1 + \sqrt{\epsilon})^2
\left\|
\frac{(\one^{L_X F'_\alpha})^{\otimes A} \otimes 
      (\one^{L_Y F'_\beta})^{\otimes B} \otimes 
      (\one^{F'_\rho})
     }{(L |F'_\alpha|)^A (L |F'_\beta|)^B |F'_\rho|}
\right\|_2^2 \\
& &
~~~~
{} - 
2 (1 - \sqrt{\epsilon})
\Tr\left[
\hsigma^{'(L_X X)^A (L_Y Y)^B \hM}
\left(
\frac{(\one^{L_X F'_\alpha})^{\otimes A} \otimes 
      (\one^{L_Y F'_\beta})^{\otimes B} \otimes 
      (\one^{F'_\rho})
     }{(L |F'_\alpha|)^A (L |F'_\beta|)^B |F'_\rho|}
\right)
\right] \\
&   =  &
(1 + 10\epsilon^{1/16})
\Tr[\hsigma^{(L_X X)^A (L_Y Y)^B \hM} 
    \hhsigma^{(L_X X)^A (L_Y Y)^B \hM}] +
\frac{(1 + \sqrt{\epsilon})^2}
     {(L |F'_\alpha|)^A (L |F'_\beta|)^B |F'_\rho|}  \\
& &
~~~~
{} - 
\frac{2 (1 - \sqrt{\epsilon})}
     {(L |F'_\alpha|)^A (L |F'_\beta|)^B |F'_\rho|}
\Tr[\hsigma^{'(L_X X)^A (L_Y Y)^B \hM} \Pi^{(L_X X)^A (L_Y Y)^B \hM}] \\
& \leq &
(1 + 10\epsilon^{1/16})
\Tr[\hsigma^{(L_X X)^A (L_Y Y)^B \hM} 
    \hhsigma^{(L_X X)^A (L_Y Y)^B \hM}] +
\frac{(1 + \sqrt{\epsilon})^2}
     {(L |F'_\alpha|)^A (L |F'_\beta|)^B |F'_\rho|}  \\
& &
~~~~
{} - 
\frac{2 (1 - \sqrt{\epsilon})(1 - 8 \epsilon^{1/16})}
     {(L |F'_\alpha|)^A (L |F'_\beta|)^B |F'_\rho|} \\
& \leq &
(1 + 10\epsilon^{1/16})
\Tr[\hsigma^{(L_X X)^A (L_Y Y)^B \hM} 
    \hhsigma^{(L_X X)^A (L_Y Y)^B \hM}] + 
\frac{12 \epsilon^{1/16} - 1}
     {(L |F'_\alpha|)^A (L |F'_\beta|)^B |F'_\rho|},
\end{array}
\end{equation}
where the first inequality follows from
Equations~\ref{eq:convsplitasym}, \ref{eq:convsplit11}, and the
second inequality follows from Equation~\ref{eq:convsplit14}.

From Equations~\ref{eq:convsplit15} and \ref{eq:convsplit16},
in order to prove
Lemma~\ref{lem:convexsplitflattening}, it suffices to show the
following lemma.
\begin{lemma}
\label{lem:asymconvexsplit}
Suppose 
\[
\sqrt{L} >
\max\left\{
\frac{4 |F'_\beta| |F'_\rho|}
     {\epsilon^{3/8} 
      2^{D^\epsilon_\infty(\rho^{YM} \| \beta^Y \otimes \rho^M)}},
\frac{4 |F'_\alpha| |F'_\rho|}
     {\epsilon^{3/8} 
      2^{D^\epsilon_\infty(\rho^{XM} \| \alpha^X \otimes \rho^M)}
     },
\frac{8 |F'_\rho|}{\epsilon^{3/8}} 
\right\},
\]
and
\begin{eqnarray*}
\log A 
& > &
D^\epsilon_\infty(\rho^{XM} \| \alpha^X \otimes \rho^M) +
\log \epsilon^{-1/16}, \\
\log B 
& > &
D^\epsilon_\infty(\rho^{YM} \| \beta^Y \otimes \rho^M) +
\log \epsilon^{-1/16}, \\
\log A + \log B 
& > &
D^\epsilon_\infty(\rho^{XYM} \| \alpha^X \otimes \beta^Y \otimes \rho^M) +
\log \epsilon^{-1/16}.
\end{eqnarray*}
Then,
\[
\Tr[\hsigma^{(L_X X)^A (L_Y Y)^B \hM} 
    \hhsigma^{(L_X X)^A (L_Y Y)^B \hM}] <
\frac{1 + 11 \epsilon^{1/16}}
     {(L |F'_\alpha|)^A (L |F'_\beta|)^B |F'_\rho|}.
\]
\end{lemma}
\begin{proof}
Given the lower bound on $\sqrt{L}$, we can write down the following
consequences of Lemmas~\ref{lem:hrhohhrho},
\ref{lem:hrhohhrho1}, \ref{lem:hrhohhrho2}, \ref{lem:hrhohhrho3}.
\begin{equation}
\label{eq:convsplit17}
\begin{array}{rcl}
\Tr[\hrho^{\hM} \hhrho^{\hM}] 
& \leq &
\frac{1+2\sqrt{\epsilon}}{|F'_\rho|}, \\
\Tr[\hrho^{L_X X \hM} \hhrho^{L_X X \hM}] 
& \leq &
\frac{(1+4\sqrt{\epsilon})
      2^{D^\epsilon_\infty(\rho^{XM} \| \alpha^X \otimes \rho^M)}
     }{(L |F'_\alpha|) |F'_\rho|}, \\
\Tr[\hrho^{L_Y Y \hM} \hhrho^{L_Y Y \hM}] 
& \leq &
\frac{(1+4\sqrt{\epsilon})
      2^{D^\epsilon_\infty(\rho^{YM} \| \beta^Y \otimes \rho^M)}
     }{(L F'_\beta) |F'_\rho|}, \\
\Tr[\hrho^{L_X X L_Y Y \hM} \hhrho^{L_X X L_Y Y \hM}] 
& \leq &
\frac{(1+3\sqrt{\epsilon})
      2^{D^\epsilon_\infty(\rho^{XYM} \| 
			   \alpha^X \otimes \beta^Y \otimes \rho^M)}
     }{(L |F'_\alpha|) (L F'_\beta) |F'_\rho|}, \\
\end{array}
\end{equation}

We have,
\begin{equation}
\label{eq:convsplit18}
\begin{array}{rcl}
\lefteqn{
\Tr[\hsigma^{(L_X X)^A (L_Y Y)^B \hM} 
    \hhsigma^{(L_X X)^A (L_Y Y)^B \hM}] 
} \\
& = &
(AB)^{-2} \cdot 
\Tr\left[
\sum_{a,\hat{a}=1}^A \sum_{b,\hat{b}=1}^B
(\hrho^{(L_X X)_a (L_Y Y)_b \hM} \otimes \alpha^{(L_X X)^{-a}} \otimes 
 \beta^{(L_Y Y)^{-b}}) 
\right. \\
& &
~~~~~~~~~~~~~~~~~~~~~~~~~~~~~~~~~~~~~~~
\left.
(\hhrho^{(L_X X)_{\hat{a}} (L_Y Y)_{\hat{b}} \hM} \otimes 
 \alpha^{(L_X X)^{-\hat{a}}} \otimes \beta^{(L_Y Y)^{-\hat{b}}})
\right] \\
& \leq &
\frac{(1+\sqrt{\epsilon})^2}
     {(AB)^2 (L |F'_\alpha|)^{2(A-1)} (L |F'_\beta|)^{2(B-1)}} \\
&  &
~~~~
{} \cdot
\sum_{a,\hat{a}=1}^A \sum_{b,\hat{b}=1}^B
\Tr[
(\hrho^{(L_X X)_a (L_Y Y)_b \hM} \otimes 
 \one^{(L_X F'_\alpha)^{-a}} \otimes \one^{(L_Y F'_\beta)^{-b}}) \\
& &
~~~~~~~~~~~~~~~~~~~~~~~~~~~~
(\hhrho^{(L_X X)_{\hat{a}} (L_Y Y)_{\hat{b}} \hM} \otimes 
 \one^{(L_X F'_\alpha)^{-\hat{a}}} \otimes\one^{(L_Y F'_\beta)^{-\hat{b}}})
] \\
\end{array}
\end{equation}

We analyse the above summation by considering several cases.
Consider the following term for a fixed choice of 
$a \neq \hat{a}$, $b \neq \hat{b}$.
\begin{equation}
\label{eq:convsplit19}
\begin{array}{rcl}
\lefteqn{
\Tr[
(\hrho^{(L_X X)_a (L_Y Y)_b \hM} \otimes 
 \one^{(L_X F'_\alpha)^{-a}} \otimes \one^{(L_Y F'_\beta)^{-b}})
(\hhrho^{(L_X X)_{\hat{a}} (L_Y Y)_{\hat{b}} \hM} \otimes 
 \one^{(L_X F'_\alpha)^{-\hat{a}}} \otimes\one^{(L_Y F'_\beta)^{-\hat{b}}})
]
} \\
& = &
\Tr[
(\hrho^{(L_X X)_a (L_Y Y)_b \hM} \otimes 
 \one^{(L_X F'_\alpha)_{\hat{a}}} \otimes 
 \one^{(L_Y F'_\beta)_{\hat{b}}} \otimes
 \one^{(L_X F'_\alpha)^{-a,\hat{a}}} \otimes 
 \one^{(L_Y F'_\beta)^{-b,\hat{b}}}) \\
&  &
~~~~~~~~~~~~~~~
(\hhrho^{(L_X X)_{\hat{a}} (L_Y Y)_{\hat{b}} \hM} \otimes 
 \one^{(L_X F'_\alpha)_a} \otimes 
 \one^{(L_Y F'_\beta)_b} \otimes
 \one^{(L_X F'_\alpha)^{-a,\hat{a}}} \otimes 
 \one^{(L_Y F'_\beta)^{-b,\hat{b}}})
] \\
& = &
\Tr[
((\hrho^{(L_X X)_a (L_Y Y)_b \hM} \otimes 
  \one^{(L_X F'_\alpha)_{\hat{a}}} \otimes 
  \one^{(L_Y F'_\beta)_{\hat{b}}}) \\
&  &
~~~~~~~~~~~~~~~
 (\hhrho^{(L_X X)_{\hat{a}} (L_Y Y)_{\hat{b}} \hM} \otimes 
  \one^{(L_X F'_\alpha)_a} \otimes 
  \one^{(L_Y F'_\beta)_b})) \otimes
\one^{(L_X F'_\alpha)^{-a,\hat{a}}} \otimes 
\one^{(L_Y F'_\beta)^{-b,\hat{b}}}
] \\
& = &
(L |F'_\alpha|)^{A - 2}
(L |F'_\beta|)^{B - 2} \\
&  &
~~~~
{} \cdot
\Tr[
(\hrho^{(L_X X)_a (L_Y Y)_b \hM} \otimes 
  \one^{(L_X F'_\alpha)_{\hat{a}}} \otimes 
  \one^{(L_Y F'_\beta)_{\hat{b}}}) \\
& &
~~~~~~~~~~~~~~~~~
(\hhrho^{(L_X X)_{\hat{a}} (L_Y Y)_{\hat{b}} \hM} \otimes 
  \one^{(L_X F'_\alpha)_a} \otimes 
  \one^{(L_Y F'_\beta)_b})
] \\
&   =  &
(L |F'_\alpha|)^{A - 2}
(L |F'_\beta|)^{B - 2} \\
&  &
~~~~
{} \cdot
\Tr[
(\hrho^{(L_X X)_a (L_Y Y)_b \hM} \otimes 
  \one^{(L_X X)_{\hat{a}}} \otimes 
  \one^{(L_Y Y)_{\hat{b}}}) \\
& &
~~~~~~~~~~~~~~~~~
(\hhrho^{(L_X X)_{\hat{a}} (L_Y Y)_{\hat{b}} \hM} \otimes 
  \one^{(L_X X)_a} \otimes 
  \one^{(L_Y Y)_b})
] \\
& = &
(L |F'_\alpha|)^{A - 2}
(L |F'_\beta|)^{B - 2}
\Tr[\hrho^{\hM} \hhrho^{\hM}] 
\;\leq\;
\frac{(1 + 2\sqrt{\epsilon}) (L |F'_\alpha|)^{A - 2} 
      (L |F'_\beta|)^{B - 2}}{|F'_\rho|},
\end{array}
\end{equation}
where we used Equation~\ref{eq:convsplit17} in the inequality above.
There are $A(A-1) B(B-1)$ such terms.

Consider the following term for a fixed choice of 
$a = \hat{a}$, $b \neq \hat{b}$.
\begin{equation}
\label{eq:convsplit20}
\begin{array}{rcl}
\lefteqn{
\Tr[
(\hrho^{(L_X X)_a (L_Y Y)_b \hM} \otimes 
 \one^{(L_X F'_\alpha)^{-a}} \otimes \one^{(L_Y F'_\beta)^{-b}})
(\hhrho^{(L_X X)_a (L_Y Y)_{\hat{b}} \hM} \otimes 
 \one^{(L_X F'_\alpha)^{-a}} \otimes\one^{(L_Y F'_\beta)^{-\hat{b}}})
]
} \\
& = &
\Tr[
(\hrho^{(L_X X)_a (L_Y Y)_b \hM} \otimes 
 \one^{(L_Y F'_\beta)_{\hat{b}}} \otimes
 \one^{(L_X F'_\alpha)^{-a}} \otimes 
 \one^{(L_Y F'_\beta)^{-b,\hat{b}}}) \\
&  &
~~~~~~~~~~~~~~~
(\hhrho^{(L_X X)_a (L_Y Y)_{\hat{b}} \hM} \otimes 
 \one^{(L_Y F'_\beta)_b} \otimes
 \one^{(L_X F'_\alpha)^{-a}} \otimes 
 \one^{(L_Y F'_\beta)^{-b,\hat{b}}})
] \\
& = &
\Tr[
((\hrho^{(L_X X)_a (L_Y Y)_b \hM} \otimes 
  \one^{(L_Y F'_\beta)_{\hat{b}}}) \\
&  &
~~~~~~~~~~~~~~~
 (\hhrho^{(L_X X)_a (L_Y Y)_{\hat{b}} \hM} \otimes 
  \one^{(L_Y F'_\beta)_b})) \otimes
\one^{(L_X F'_\alpha)^{-a}} \otimes 
\one^{(L_Y F'_\beta)^{-b,\hat{b}}}
] \\
& = &
(L |F'_\alpha|)^{A - 1}
(L |F'_\beta|)^{B - 2} 
\Tr[
(\hrho^{(L_X X)_a (L_Y Y)_b \hM} \otimes 
  \one^{(L_Y F'_\beta)_{\hat{b}}}) \\
& &
~~~~~~~~~~~~~~~~~~~~~~~~~~~~~~~~~~~~~~~~~
(\hhrho^{(L_X X)_a (L_Y Y)_{\hat{b}} \hM} \otimes 
  \one^{(L_Y F'_\beta)_b})
] \\
&   =  &
(L |F'_\alpha|)^{A - 1}
(L |F'_\beta|)^{B - 2} 
\Tr[
(\hrho^{(L_X X)_a (L_Y Y)_b \hM} \otimes 
  \one^{(L_Y Y)_{\hat{b}}}) 
(\hhrho^{(L_X X)_{\hat{a}} (L_Y Y)_{\hat{b}} \hM} \otimes 
  \one^{(L_Y Y)_b})
] \\
& = &
(L |F'_\alpha|)^{A - 1}
(L |F'_\beta|)^{B - 2}
\Tr[\hrho^{(L_X X) \hM} \hhrho^{(L_X X) \hM}] \\
& \leq &
\frac{(1 + 4\sqrt{\epsilon}) 
      (L |F'_\alpha|)^{A - 1} (L |F'_\beta|)^{B - 2}
      2^{D^\epsilon_\infty(\rho^{XM} \| \alpha^X \otimes \rho^M)}
     }{L |F'_\alpha| |F'_\rho|},
\end{array}
\end{equation}
where we used Equation~\ref{eq:convsplit17} in the inequality above.
There are $A B(B-1)$ such terms.

Consider the following term for a fixed choice of 
$a \neq \hat{a}$, $b = \hat{b}$. We have similarly,
\begin{equation}
\label{eq:convsplit21}
\begin{array}{rcl}
\lefteqn{
\Tr[
(\hrho^{(L_X X)_a (L_Y Y)_b \hM} \otimes 
 \one^{(L_X F'_\alpha)^{-a}} \otimes \one^{(L_Y F'_\beta)^{-b}}) 
} \\
&  &
~~~~~~~~
(\hhrho^{(L_X X)_{\hat{a}} (L_Y Y)_b \hM} \otimes 
 \one^{(L_X F'_\alpha)^{-\hat{a}}} \otimes\one^{(L_Y F'_\beta)^{-b}})
] \\
& \leq &
\frac{(1 + 4\sqrt{\epsilon}) 
      (L |F'_\alpha|)^{A - 2} (L |F'_\beta|)^{B - 1}
      2^{D^\epsilon_\infty(\rho^{YM} \| \beta^Y \otimes \rho^M)}
     }{L |F'_\beta| |F'_\rho|}.
~~~~~~~~~~~~~~~~~~~~~~~~~~~~~~~~~~
\end{array}
\end{equation}
There are $A(A-1) B$ such terms.

Finally, consider the following term for a fixed choice of 
$a = \hat{a}$, $b = \hat{b}$.
\begin{equation}
\label{eq:convsplit22}
\begin{array}{rcl}
\lefteqn{
\Tr[
(\hrho^{(L_X X)_a (L_Y Y)_b \hM} \otimes 
 \one^{(L_X F'_\alpha)^{-a}} \otimes \one^{(L_Y F'_\beta)^{-b}})
} \\
& &
~~~~~~
(\hhrho^{(L_X X)_a (L_Y Y)_b \hM} \otimes 
 \one^{(L_X F'_\alpha)^{-a}} \otimes\one^{(L_Y F'_\beta)^{-b}})
] \\
& = &
\Tr[
(\hrho^{(L_X X)_a (L_Y Y)_b \hM} \hhrho^{(L_X X)_a (L_Y Y)_b \hM}) \otimes 
\one^{(L_X F'_\alpha)^{-a}} \otimes \one^{(L_Y F'_\beta)^{-b}}
] \\
& = &
(L |F'_\alpha|)^{A - 1}
(L |F'_\beta|)^{B - 1} 
\Tr[\hrho^{(L_X X)_a (L_Y Y)_b \hM} \hhrho^{(L_X X)_a (L_Y Y)_b \hM}] \\
& \leq &
\frac{(1 + 3\sqrt{\epsilon}) 
      (L |F'_\alpha|)^{A - 1} (L |F'_\beta|)^{B - 1}
      2^{D^\epsilon_\infty(\rho^{XYM} \|
	 \alpha^X \otimes \beta^Y \otimes \rho^M)}
     }{(L |F'_\alpha|) (L |F'_\beta|) |F'_\rho|},
\end{array}
\end{equation}
where we used Equation~\ref{eq:convsplit17} in the inequality above.
There are $A B$ such terms.

From Equations~\ref{eq:convsplit19}, \ref{eq:convsplit20},
\ref{eq:convsplit21}, \ref{eq:convsplit22} we get
\begin{equation}
\label{eq:convsplit23}
\begin{array}{rcl}
\lefteqn{
\sum_{a,\hat{a}=1}^A \sum_{b,\hat{b}=1}^B
\Tr[
(\hrho^{(L_X X)_a (L_Y Y)_b \hM} \otimes 
 \one^{(L_X F'_\alpha)^{-a}} \otimes \one^{(L_Y F'_\beta)^{-b}})
} \\
& &
~~~~~~~~~~~~~~~~
(\hhrho^{(L_X X)_{\hat{a}} (L_Y Y)_{\hat{b}} \hM} \otimes 
 \one^{(L_X F'_\alpha)^{-\hat{a}}} \otimes\one^{(L_Y F'_\beta)^{-\hat{b}}})
] \\
& = &
\frac{A(A-1) B(B-1) (1 + 2\sqrt{\epsilon}) (L |F'_\alpha|)^{A - 2} 
      (L |F'_\beta|)^{B - 2}}{|F'_\rho|} \\
& &
~~~~
{} +
\frac{A B(B-1) (1 + 4\sqrt{\epsilon}) 
      (L |F'_\alpha|)^{A - 1} (L |F'_\beta|)^{B - 2}
      2^{D^\epsilon_\infty(\rho^{XM} \| \alpha^X \otimes \rho^M)}
     }{L |F'_\alpha| |F'_\rho|} \\
& &
~~~~
{} +
\frac{A(A-1) B (1 + 4\sqrt{\epsilon}) 
      (L |F'_\alpha|)^{A - 2} (L |F'_\beta|)^{B - 1}
      2^{D^\epsilon_\infty(\rho^{YM} \| \beta^Y \otimes \rho^M)}
     }{L |F'_\beta| |F'_\rho|} \\
& &
~~~~
{} +
\frac{A B (1 + 3\sqrt{\epsilon}) 
      (L |F'_\alpha|)^{A - 1} (L |F'_\beta|)^{B - 1}
      2^{D^\epsilon_\infty(\rho^{XYM} \| 
	                   \alpha^X \otimes \beta^Y \otimes \rho^M)}
     }{(L |F'_\alpha|) (L |F'_\beta|)  |F'_\rho|}.
\end{array}
\end{equation}
From Equations~\ref{eq:convsplit18}, \ref{eq:convsplit23} we get,
\begin{eqnarray*}
\lefteqn{
\Tr[\hsigma^{(L_X X)^A (L_Y Y)^B \hM} 
    \hhsigma^{(L_X X)^A (L_Y Y)^B \hM}] 
} \\
& \leq &
\frac{(1+\sqrt{\epsilon})^2}
     {(AB)^2 (L |F'_\alpha|)^{2(A-1)} (L |F'_\beta|)^{2(B-1)}} \\
& &
~~~~
{} \cdot 
\left[
\frac{A(A-1) B(B-1) (1 + 2\sqrt{\epsilon}) (L |F'_\alpha|)^{A - 2} 
      (L |F'_\beta|)^{B - 2}}{|F'_\rho|} 
\right. \\
& &
~~~~~~~~~~~~
{} +
\frac{A B(B-1) (1 + 4\sqrt{\epsilon}) 
      (L |F'_\alpha|)^{A - 1} (L |F'_\beta|)^{B - 2}
      2^{D^\epsilon_\infty(\rho^{XM} \| \alpha^X \otimes \rho^M)}
     }{L |F'_\alpha| |F'_\rho|} \\
& &
~~~~~~~~~~~~
{} +
\frac{A(A-1) B (1 + 4\sqrt{\epsilon}) 
      (L |F'_\alpha|)^{A - 2} (L |F'_\beta|)^{B - 1}
      2^{D^\epsilon_\infty(\rho^{YM} \| \beta^Y \otimes \rho^M)}
     }{L |F'_\beta| |F'_\rho|} \\
& &
~~~~~~~~~~~~
\left.
{} +
\frac{A B (1 + 3\sqrt{\epsilon}) 
      (L |F'_\alpha|)^{A - 1} (L |F'_\beta|)^{B - 1}
      2^{D^\epsilon_\infty(\rho^{XYM} \| 
	                   \alpha^X \otimes \beta^Y \otimes \rho^M)}
     }{(L |F'_\alpha|) (L |F'_\beta|)  |F'_\rho|}
\right] \\
& \leq &
(1+ 7 \sqrt{\epsilon})
\left[
\frac{1}{(L |F'_\alpha|)^A (L |F'_\beta|)^B |F'_\rho|} +
A^{-1} \cdot
\frac{2^{D^\epsilon_\infty(\rho^{XM} \| \alpha^X \otimes \rho^M)}
     }{(L |F'_\alpha|)^A (L |F'_\beta|)^B |F'_\rho|} 
\right. \\
& &
~~~~~~~~~~~~~~~~~~~
\left.
{} +
B^{-1} \cdot 
\frac{2^{D^\epsilon_\infty(\rho^{YM} \| \beta^Y \otimes \rho^M)}
     }{(L |F'_\alpha|)^A (L |F'_\beta|)^B |F'_\rho|} +
(A B)^{-1} \cdot
\frac{2^{D^\epsilon_\infty(\rho^{XYM} \| 
	                   \alpha^X \otimes \beta^Y \otimes \rho^M)}
     }{(L |F'_\alpha|)^A (L |F'_\beta|)^B |F'_\rho|}
\right] \\
& \leq &
\frac{(1+ 7 \sqrt{\epsilon}) (1 + 3 \epsilon^{1/16})} 
     {(L |F'_\alpha|)^A (L |F'_\beta|)^B |F'_\rho|} 
\;\leq\;
\frac{1+ 11 \epsilon^{1/16}} 
     {(L |F'_\alpha|)^A (L |F'_\beta|)^B |F'_\rho|}.
\end{eqnarray*}

This finishes the proof of the lemma.
\end{proof}

Since the dimension $L$ of the augmenting Hilbert spaces $L_X$, $L_Y$
is a free parameter in our proof, it can be chosen large enough as 
in the requirement of Lemma~\ref{lem:asymconvexsplit}. Thus
Lemma~\ref{lem:asymconvexsplit} holds, which completes the proof of
Lemma~\ref{lem:convexsplitflattening}.

\section{Smooth soft covering without pairwise independence}
\label{sec:coveringnonpairwise}
We will now prove a smooth multipartite soft covering classical quantum 
version of
Lemma~\ref{lem:convexsplitflattening} where pairwise independence amongst
the classical random variables can be relaxed slightly. For simplicity,
we will prove the case with two parties only. However, analogous results
hold for greater than two parties also, as the machinery will make amply
clear below.

Let $\cX$, $\cY$ be classical alphabets. 
Let $p^{XY}$ be a normalised probability distribution on $X Y$, whose
marginals are denoted by $p^X$, $p^Y$.
Let $q^{X}$, $q^{Y}$ be normalised probability distributions.
For each $(x, y) \in \cX \times \cY$, 
let $\rho^M_{x y}$ be a subnormalised density matrix on $M$. 
We define the following classical quantum {\em control state}:
\[
\rho^{XYM} :=
\sum_{(x, y) \in \cX \times \cY}
p^{XY}(x, y) \ketbra{x, y}^{XY} \otimes \rho^M_{x y}.
\]
Define the density matrices
\[
\rho^M_x := \sum_y p^{Y|X=x}(y) \rho^M_{xy}, ~
\rho^M_y := \sum_x p^{X|Y=y}(x) \rho^M_{xy},
\]
and the classical quantum states
\[
\rho^{XM, Y|X} :=
\sum_{x, y}
p^{XY}(x, y) \ketbra{x, y}^{XY} \otimes \rho^M_{x}, ~
\rho^{YM, X|Y} :=
\sum_{x, y}
p^{XY}(x, y) \ketbra{x, y}^{XY} \otimes \rho^M_{y}.
\]
Intuitively, $\rho^M_x$ is the quantum state induced on $M$ when the
$X$ register of the control state $\rho^{XYM}$ is measured and found to
be equal to $x$. Similarly, the classical quantum state $\rho^{XM, Y|X}$
is what one gets if one `decouples' the register $M$ from $Y$, making
it depend only on the contents of $X$, without
affecting the probability distribution on $XY$. Analogous remarks can
be made for $\rho^M_y$ and $\rho^{YM, X|Y}$.

Suppose $\supp(p^{X}) \leq \supp(q^{X})$ and
$\supp(p^{Y}) \leq \supp(q^{Y})$. Let $A$, $B$ be positive integers.
Let $x^{(A)} := (x(1), \ldots, x(A)) \in \cX^A$ denote a $|A|$-tuple
of elements from  $\cX$; $y^{(B)}$ has a similar meaning.
For any pair of tuples $x^{(A)}$, $y^{(B)}$,
define the {\em sample average covering state}
\begin{equation}
\label{eq:sampleaverage}
\sigma^M_{x^{(A)}, y^{(B)}} :=
(A B)^{-1} \sum_{a = 1}^{A} \sum_{b = 1}^{B}
\frac{p^{XY}(x(a), y(b))}{q^{X}(x(a)) q^{Y}(y(b))}
\rho^M_{x(a), y(b)},
\end{equation}
where the fraction term above represents the `change of measure' from
the product probability distribution $q^{X} \times q^Y$ to the
joint probability distribution $p^{XY}$.

Let $\bar{q}^{X^A}$, $\bar{q}^{Y^B}$ be two probability distributions
on $\cX^A$, $\cY^B$ with the following properties:
\begin{enumerate}
\item
The marginal on any coordinate $\bar{q}^{X_a}$ for any $a \in [A]$ equals
$q^X$;

\item
The marginal on any two coordinates $\bar{q}^{X_{a} X_{a'}}$ for any
$a, a' \in [A]$, $a \neq a'$ is independent of $a$, $a'$; hence it
can be modelled as a distribution $\bar{q}^{X X'}$;

\item
Similar properties as the two above hold for $\bar{q}^{Y^B}$.
\end{enumerate}
For any $(x,y) \in \cX \times \cY$, define the positive semidefinite
matrices
\[
\begin{array}{c}
\rho(X)^M_{x y} :=
\sum_{y'} \frac{\bar{q}^{Y'|Y=y}(y') p^{Y|X=x}(y')}{q^Y(y')} 
\rho^M_{x y'}, ~
\rho(Y)^M_{x y} :=
\sum_{x'} \frac{\bar{q}^{X'|X=x}(x') p^{X|Y=y}(x')}{q^X(x')} 
\rho^M_{x' y}, \\
\rho()^M_{x y} :=
\sum_{x' y'} \frac{\bar{q}^{Y'|Y=y}(y') q^{X'|X=x}(x') p(x',y')}
	          {q^X(x') q^Y(y')} \rho^M_{x' y'}.
\end{array}
\]
\begin{definition}
\label{def:nonpairwise}
Let $0 < \epsilon < 1$.
For the two party classical quantum soft covering problem,
the consequences of loss of pairwise independence are said to hold up to 
scale factors of $(\epsilon, f(\epsilon), g(\epsilon))$ if there exist 
probability distributions
$p(X)^{XY}$, $p(Y)^{XY}$, $p()^{XY}$ satisfying
\[
\|p^{XY} - p(X)^{XY}\|_1 < \epsilon, ~~
\|p^{XY} - p(Y)^{XY}\|_1 < \epsilon, ~~
\|p^{XY} - p()^{XY}\|_1 < \epsilon,
\]
such that 
\[
\rho(X)^{XYM} \leq f(\epsilon) \rho^{XM, Y|X}, ~
\rho(Y)^{XYM} \leq f(\epsilon) \rho^{YM, X|Y}, ~
\rho()^{XYM} \leq  (1 + g(\epsilon)) (p^{XY} \otimes \rho^M),
\]
where
\[
\begin{array}{c}
\rho(X)^{XYM} :=
\sum_{(x, y) \in \cX \times \cY}
	p(X)^{XY}(x,y) \ketbra{x, y}^{XY} \otimes \rho(X)^M_{x y}, \\
\rho(Y)^{XYM} :=
\sum_{(x, y) \in \cX \times \cY}
	p(Y)^{XY}(x,y)  \ketbra{x, y}^{XY} \otimes \rho(Y)^M_{x y}, \\
\rho()^{XYM} :=
\sum_{(x, y) \in \cX \times \cY}
p()^{XY}(x,y) \ketbra{x, y}^{XY} \otimes \rho()^M_{x y}.
\end{array}
\]
In this paper, we shall take 
$f(\epsilon) \leq 2^{\mathrm{polylog}(\epsilon^{-1})}$ and
$g(\epsilon) \leq \epsilon^{\Theta(1)}$.
\end{definition}

We are now ready to give the statement of our smooth two party classical
quantum covering lemma under slight relaxation of pairwise independence.
\begin{proposition}
\label{prop:coveringnonpairwise}
Consider the setting of the two party classical quantum soft covering
problem.
Suppose,
\begin{eqnarray*}
\log A 
& > &
D^\epsilon_\infty(\rho^{XM} \| q^X \otimes \rho^M) +
\log f(e) + 3 + \log \epsilon^{-1/2}, \\
\log B 
& > &
D^\epsilon_\infty(\rho^{YM} \| q^Y \otimes \rho^M) +
\log f(e) + 3 + \log \epsilon^{-1/2}, \\
\log A + \log B 
& > &
D^\epsilon_\infty(\rho^{XYM} \| q^X \otimes q^Y \otimes \rho^M) +
1 + \log \epsilon^{-1/2}.
\end{eqnarray*}
Then,
\[
\E_{x^{(A)}, y^{(B)}}[
\|\sigma^M_{x^{(A)}, y^{(B)}} - \rho^M\|_1
]  < 
(2 \sqrt{15 \epsilon^{1/64} + \sqrt{40\epsilon^{1/32} + 8e + g(e)}} + 4e) 
(\Tr\rho).
\]
where the expectation is taken over independent choices of tuples
$x^{(A)}$, $y^{(B)}$ from the distributions $\bar{q}^{X^A}$, 
$\bar{q}^{Y^B}$ which satisfy the consequences of loss of 
pairwise independence
up to scale factors of $(e, f(e), g(e))$.
Note that for the usual setting of $e = \epsilon$,
$f(e) \leq 2^{\mathrm{polylog}(e^{-1})}$ and
$g(e) \leq e^{\Theta(1)}$, the right hand 
sides
of the inequalities in the above rate region will have additive terms
of $\mathrm{polylog}(\epsilon^{-1})$ and the covering error
will be bounded by $\epsilon^{\Theta(1)}$. 
\end{proposition}

\noindent
{\bf Remarks:}

\noindent
1.\ 
In the pairwise independent setting, $\bar{q}^{XX'} = q^X \times q^{X'}$,
$\bar{q}^{YY'} = q^X \times q^{Y'}$. Hence, 
$\rho(X)^{XYM} = \rho^{XM, Y|X}$,
$\rho(Y)^{XYM} = \rho^{YM, X|Y}$ and
$\rho()^{XYM} = p^{XY} \otimes \rho^{M}$. Thus, $\bar{q}^{XX'}$,
$\bar{q}^{YY'}$ satisfy the consequences of loss of pairwise independence
up to scale factors of $(0, 1, 0)$ i.e. perfectly. The rate region and
covering error promised by Proposition~\ref{prop:coveringnonpairwise}
is essentially the same as that given by the smooth soft covering
lemma of \cite{Sen:telescoping}. The main difference is that
the rate region of Proposition~\ref{prop:coveringnonpairwise} is
given in terms of
$D^\epsilon_\infty(\cdot\|\cdot)$ versus $D^\epsilon_2(\cdot\|\cdot)$
in \cite{Sen:telescoping}.

\noindent
2.\ 
Proposition~\ref{prop:coveringnonpairwise} cannot be proved by the 
telescoping technique. Consider a simplified setting where 
$p^{XY} = p^X \otimes p^Y$, $q^X = p^X$, $q^Y = p^Y$, 
$p^{XY} = p(X)^{XY} = p(Y)^{XY} = p()^{XY}$, but 
$q^{XX'} \neq q^X \otimes q^{X'}$, $q^{YY'} \neq q^Y \otimes q^{Y'}$. 
In other words, $p^{XY}$ is a tensor
product distribution, the sampling distributions $q^X$, $q^Y$ are the 
same as the marginals $p^X$, $p^Y$ but there is a loss of pairwise
independence in sampling $(x(a))_a$ as well as in sampling $(y(b))_b$;
nevertheless $p^{XY}$ does not require perturbation in order to
satisfy the operator inequalities involving $f(e)$, $g(e)$.
Telescoping starts off by using the mapping 
$
(x,y) \mapsto 
\rho^{'M}_{xy} := (\rho^M_{xy} - \rho^M_x - \rho^M_y + \rho^M)
$
instead of using the original mapping 
$(x,y) \mapsto \rho^M_{xy}$.
The mean zero requirement of telescoping is the fact that
\[
\forall (x,y): ~~
\E_{(x,x') \sim p^X \times p^{X'}}[
\rho^{'M}_{xy} \rho^{'M}_{x'y}] =
\E_{(y,y') \sim p^Y \times p^{Y'}}[
\rho^{'M}_{xy} \rho^{'M}_{xy'}] = 0.
\]
However under the `slight breaking of pairwise independence', we have to
evaluate the expectation under $(x,x') \sim q^{XX'}$, and similarly for
$(y,y')$. Since $q^{XX'} \neq p^X \times p^{X'}$, the mean zero requirement
fails. Replacing $\rho^{'M}_{xy}$ with other candidates like 
$
(\rho^M_{xy} - \rho(X)^M_{xy} - \rho(Y)^M_{xy} + \rho()^M_{xy})
$
does not help either.

\medskip

To prove Proposition~\ref{prop:coveringnonpairwise}, we generally
follow the lines of the proof of Lemma~\ref{lem:convexsplitflattening}.
Several steps are simpler because $X$, $Y$ are classical. In particular,
we do not need to flatten $q^X$, $q^Y$. Only the flattening superoperator
$\cF_{\rho}^{M \rightarrow L_M M}$ is needed. Again, without loss of
generality, we assume that $\rho^{XYM}$ is a normalised density matrix.

From the notion of intersection of classical probability distributions
defined in Equation~\ref{eq:classicalintersection}, we get
\[
\|(p \cap p(X) \cap p(Y) \cap p())^{XY} - p^{XY}\|_1 < 3e.
\]
Arguing along the lines of the initial part of the proof of 
Proposition\ref{prop:projsmoothing} and defining a set 
\[
\mathrm{Bad} := 
\{(x,y): p^{XY}(x,y) > 4 (p \cap p(X) \cap p(Y) \cap p())^{XY}(x,y)\},
\]
we see that there 
is a subnormalised probability distribution $\hat{p}^{XY}$ such that
\begin{equation}
\label{eq:softcoveringminus1}
\|\hat{p}^{XY} - p^{XY}\|_1 < 4 e,
\hat{p}^{XY} \leq p^{XY} < 4 \hat{p}^{XY}, 
\hat{p}^{XY} \leq p(X)^{XY}, 
\hat{p}^{XY} \leq p(Y)^{XY},
\hat{p}^{XY} \leq p()^{XY},
\end{equation}
where the inequality $p^{XY} < 4 \hat{p}^{XY}$ above holds only 
over $\supp(\hat{p}^{XY})$.
We can define the sample average covering state 
$\hsigma^M_{x^{(A)}, y^{(B)}}$ using $\hat{p}^{XY}$ instead of
$p^{XY}$ as in Equation~\ref{eq:sampleaverage} in the natural fashion.
Similarly, we can define the new classical quantum control state
$\hrho^{XYM}$.
Then for any $x^{(A)}$, $y^{(B)}$, using Equations~\ref{eq:sampleaverage},
\ref{eq:softcoveringminus1}, we can easily see that
$
\|\hsigma^M_{x^{(A)}, y^{(B)}} - \sigma^M_{x^{(A)}, y^{(B)}}\|_1 < 
4 e.
$
So in order to prove Proposition~\ref{prop:coveringnonpairwise}, it
suffices to prove
\begin{equation}
\label{eq:softcoveringminus2}
\E_{x^{(A)}, y^{(B)}}[
\|\hsigma^M_{x^{(A)}, y^{(B)}} - \rho^M\|_1
] 
< 2 \sqrt{15 \epsilon^{1/64} + \sqrt{40\epsilon^{1/32} + 8e + g(e)}}.
\end{equation}

From the definitions of smooth $D^\epsilon_\infty(\cdot \| \cdot)$,
we have subnormalised classical quantum density matrices 
$\rho^{XYM}(3)$, $\rho^{XM}(1)$ and $\rho^{YM}(2)$ satisfying
\[
\begin{array}{c}
\|\rho^{XYM} - \rho^{XYM}(3)\|_1 \leq \epsilon, ~~
\|\rho^{XM} - \rho^{XM}(1)\|_1 \leq \epsilon, ~~
\|\rho^{YM} - \rho^{YM}(2)\|_1 \leq \epsilon,
\end{array}
\]
such that
\begin{equation}
\label{eq:softcovering2}
\begin{array}{rcl}
\rho^{XYM}(3)
& \leq &
2^{D^{\epsilon}_\infty(\rho^{XYM} \| q^X \otimes q^Y \otimes \rho^M)} 
(q^X \otimes q^Y \otimes \rho^M), \\
\rho^{XM}(1)
& \leq &
2^{D^{\epsilon}_\infty(\rho^{XM} \| q^{X} \otimes \rho^M)}
(q^{X} \otimes \rho^M), \\
\rho^{YM}(2)
& \leq &
2^{D^{\epsilon}_\infty(\rho(Y)^{YM} \| q^{Y} \otimes \rho^M)}.
\end{array}
\end{equation}
Note that in general $\rho^{XM}(1)$, $\rho^{YM}(2)$ have nothing to
do with the marginals of $\rho^{XYM}(3)$. 
From Equation~\ref{eq:softcoveringminus1} and the triangle inequality,
we can write
\begin{equation}
\label{eq:softcoveringminus3}
\begin{array}{c}
\|\hrho^{XYM} - \rho^{XYM}(3)\|_1 \leq 4e + \epsilon, 
\|\hrho^{XM} - \rho^{XM}(1)\|_1 \leq 4e + \epsilon, 
\|\hrho^{YM} - \rho^{YM}(2)\|_1 \leq 4e + \epsilon.
\end{array}
\end{equation}
We can extend the classical quantum states $\rho^{XM}(1)$, $\rho^{YM}(2)$
to classical quantum states $\rho^{XM, Y|X}(1)$, $\rho^{YM, X|Y}(2)$ in
the natural fashion. Henceforth we will be working with the control
state $\hrho^{XYM}$ and probability distribution $\hat{p}^{XY}$ only. 
So to ease the notation below, we rename $\hat{p}^{XY}$ to $p^{XY}$,
$\hrho^{XYM}$ to $\rho^{XYM}$, $\hsigma^M_{x^{(A)}, y^{(B)}}$ to
$\sigma^M_{x^{(A)}, y^{(B)}}$. Under this simplified notation, 
in order to prove Proposition~\ref{prop:coveringnonpairwise},
Equation~\ref{eq:softcoveringminus2} rewritten says that
it suffices to show
\begin{equation}
\label{eq:softcovering0}
\E_{x^{(A)}, y^{(B)}}[
\|\sigma^M_{x^{(A)}, y^{(B)}} - \rho^M\|_1
]  
< 2 \sqrt{15 \epsilon^{1/64} + \sqrt{40\epsilon^{1/32} + 8e + g(e)}}.
\end{equation}
Also, Equation~\ref{eq:softcoveringminus3} is rewritten as
\begin{equation}
\label{eq:softcovering1}
\begin{array}{c}
\|\rho^{XYM} - \rho^{XYM}(3)\|_1 \leq 4e + \epsilon, 
\|\rho^{XM} - \rho^{XM}(1)\|_1 \leq 4e + \epsilon, 
\|\rho^{YM} - \rho^{YM}(2)\|_1 \leq 4e + \epsilon.
\end{array}
\end{equation}

Apply the flattening CPTP superoperator to get the normalised 
density matrix
\[
\rho^{XY L_M M} :=
(\I^{XY} \otimes \cF_{\rho}^{M \rightarrow L_M M})(\rho^{XYM}).
\]
Define $\rho^{XY L_M M}(3)$, $\rho^{X L_M M}(1)$,
$\rho^{Y L_M M}(2)$, $\rho(X)^{XY L_M M}$, $\rho(Y)^{XY L_M M}$,
$\rho()^{XY L_M M}$ via flattening in the natural fashion.
Extend $\rho^{X L_M M}(1)$, $\rho^{Y L_M M}(2)$
to classical quantum states $\rho^{X L_M M, Y|X}(1)$, 
$\rho^{Y L_M M, X|Y}(2)$ in the natural fashion. 
Then we use Equation~\ref{eq:softcovering1} and conclude,
\begin{equation}
\label{eq:softcovering3}
\begin{array}{c}
\|\rho^{X Y L_M M} - \rho^{X Y L_M M}(3)\|_1 
< 4e + \epsilon, \\
\|\rho^{X L_M M} - \rho^{X L_M M}(1)\|_1 
< 4e + \epsilon, ~~
\|\rho^{Y L_M M} - \rho^{Y L_M M}(2)\|_1 
< 4e + \epsilon, 
\end{array}
\end{equation}
since a CPTP superoperator cannot increase the trace distance.

By Lemma~\ref{lem:flatteningdenom} and Proposition~\ref{prop:flattening}, 
we have
\begin{equation}
\label{eq:softcovering4}
(1+\delta)^{-2} \frac{\one^{F'_{\rho}}}{|F'_{\rho}|} \leq
\cF_{\rho}^{M \rightarrow L_M M}(\rho^M) \leq
(1+\delta) \frac{\one^{F'_{\rho}}}{|F'_{\rho}|}.
\end{equation}
Because a CPTP superoperator respects the L\"{o}wner partial order on
Hermitian matrices, combining Equations~\ref{eq:softcovering2} and
\ref{eq:softcovering4} gives us
\begin{equation}
\label{eq:softcovering5}
\begin{array}{rcl}
\rho^{X Y L_M M}(3)
& \leq &
(1+\delta)
2^{D^\epsilon_\infty(\rho^{XYM} \| q^X \otimes q^Y \otimes \rho^M)} 
(q^X \otimes q^Y \otimes (\frac{\one^{F'_{\rho}}}{|F'_{\rho}|})),\\
\rho^{X L_M M}(1)
& \leq &
(1+\delta)
2^{D^\epsilon_\infty(\rho^{XM} \| q^X \otimes \rho^M)}
(q^X \otimes (\frac{\one^{F'_{\rho}}}{|F'_{\rho}|})),\\
\rho^{Y L_M M}(2)
& \leq &
(1+\delta)
2^{D^\epsilon_\infty(\rho^{YM} \| q^Y \otimes \rho^M)}
(q^Y \otimes (\frac{\one^{F'_{\rho}}}{|F'_{\rho}|})).
\end{array}
\end{equation}

Applying Proposition~\ref{prop:projsmoothing} to
Equations~\ref{eq:softcovering3}, \ref{eq:softcovering5} tells us that 
there exists classical quantum orthogonal projectors 
\begin{eqnarray*}
\Pi^{X Y L_M M}(3) 
& := &
\sum_{x y} \ketbra{x,y}^{XY} \otimes \Pi^{L_M M}(3)_{xy}, \\
\Pi^{X L_M M}(1) 
& := &
\sum_{x} \ketbra{x}^{X} \otimes \Pi^{L_M M}(1)_{x}, \\
\Pi^{Y L_M M}(2) 
& := &
\sum_{y} \ketbra{y}^{Y} \otimes \Pi^{L_M M}(2)_{y}, \\
\end{eqnarray*}
such that
\begin{equation}
\label{eq:softcovering6}
\begin{array}{rcl}
\Tr[\Pi^{X Y L_M M}(3) \rho^{X Y L_M M}] 
& \geq & 
1 - \sqrt{4e + \epsilon}, \\
\Tr[\Pi^{X L_M M}(1) \rho^{X L_M M}] 
& \geq & 
1 - \sqrt{4e + \epsilon}, \\
\Tr[\Pi^{Y L_M M}(2) \rho^{Y L_M M}] 
& \geq & 
1 - \sqrt{4e + \epsilon}, \\
\Pi^{X Y L_M M}(3) \rho^{X Y L_M M} 
& = &
\rho^{X Y L_M M} \Pi^{X Y L_M M}(3), \\
\Pi^{X L_M M}(1) \rho^{X L_M M}
& = &
\rho^{X L_M M} \Pi^{X L_M M}(1), \\
\Pi^{Y L_M M}(2) \rho^{Y L_M M}
& = &
\rho^{Y L_M M} \Pi^{Y L_M M}(2), \\
\forall (x,y): \frac{p(x,y)}{q(x) q(y)} 
\|\Pi^{L_M M}_{x y}(3) \, \rho^{L_M M}_{xy}\|_\infty
& \leq &
\frac{
(1+\delta) (1 + 2\sqrt{4e + \epsilon})
2^{
D^\epsilon_\infty(\rho^{XYM} \| q^X \otimes q^Y \otimes \rho^M)
}}{|F'_{\rho}|}, \\
\forall x: \frac{p(x)}{q(x)} 
\|\Pi^{L_M M}_{x}(1) \, \rho^{L_M M}_{x}\|_\infty
& \leq &
\frac{
(1+\delta) (1 + 2\sqrt{4e + \epsilon})
2^{
D^\epsilon_\infty(\rho^{XM} \| q^X \otimes \rho^M)
}}{|F'_{\rho}|}, \\
\forall y: \frac{p(y)}{q(y)} 
\|\Pi^{L_M M}_{y}(2) \, \rho^{L_M M}_{y}\|_\infty
& \leq &
\frac{
(1+\delta) (1 + 2\sqrt{4e + \epsilon})
2^{
D^\epsilon_\infty(\rho^{YM} \| q^Y \otimes \rho^M)
}}{|F'_{\rho}|}.
\end{array}
\end{equation}

Define the {\em flattened sample average covering state}
\[
\sigma^{L_M M}_{x^{(A)}, y^{(B)}} :=
(A B)^{-1} \sum_{a = 1}^{A} \sum_{b = 1}^{B}
\frac{p^{XY}(x(a), y(b))}{q^{X}(x(a)) q^{Y}(y(b))}
\rho^{L_M M}_{x(a), y(b)},
\]
Proposition \ref{prop:fidelitypreserving} and 
Fact~\ref{fact:fidelitytracedist} together with concavity of square
root imply that
\begin{equation}
\label{eq:flatteningfidelity}
\begin{array}{rcl}
\E_{x^{(A)}, y^{(B)}}[\|\sigma^M_{x^{(A)}, y^{(B)}} - \rho^M\|_1] 
& \leq &
2 \E_{x^{(A)}, y^{(B)}}[
\sqrt{\|\sigma^{L_M M}_{x^{(A)}, y^{(B)}} - \rho^{L_M M}\|_1}] \\
& \leq &
2 \sqrt{\E_{x^{(A)}, y^{(B)}}[
	\|\sigma^{L_M M}_{x^{(A)}, y^{(B)}} - \rho^{L_M M}\|_1]},
\end{array}
\end{equation}
where the expectation is taken over independent choices of tuples
$x^{(A)}$, $y^{(B)}$ from the distributions $\bar{q}^{X^A}$, 
$\bar{q}^{Y^B}$. 
We note that $\delta$ can be taken as small as we please
at the expense of increasing $|F'_{\rho}|$. Essentially, 
$\delta$ can be treated as a free parameter. We will set 
$\delta := \sqrt{\epsilon} / 3$.

To make the
notation lighter, henceforth we will always work with the flattened
states and the Hilbert space $L_M M$. Hence from
now on we redefine $M \equiv L_M M$. Under this redefinition, we restate 
Equations~\ref{eq:softcovering6}, \ref{eq:softcovering4}, and
the consequences of Equation~\ref{eq:softcoveringminus1}  as 
Equation~\ref{eq:softcovering7} below. That is, we can assume that there
are classical quantum orthogonal projectors 
$\Pi^{XYM}(3)$, $\Pi^{XM}(1)$ and $\Pi^{YM}(2)$ such that 
\begin{equation}
\label{eq:softcovering7}
\begin{array}{rcl}
\Tr[\Pi^{X Y M}(3) \rho^{X Y M}] 
& \geq & 
1 - \sqrt{\epsilon} - 2\sqrt{e}, \\
\Tr[\Pi^{X M}(1) \rho^{X M}] 
& \geq & 
1 - \sqrt{\epsilon} - 2\sqrt{e}, \\
\Tr[\Pi^{Y M}(2) \rho^{Y M}] 
& \geq & 
1 - \sqrt{\epsilon} - 2\sqrt{e}, \\
\Pi^{X Y M}(3) \rho^{X Y M} 
& = &
\rho^{X Y M} \Pi^{X Y M}(3), \\
\Pi^{X M}(1) \rho^{X M}
& = &
\rho^{X M} \Pi^{X M}(1), \\
\Pi^{Y M}(2) \rho^{Y M}
& = &
\rho^{Y M} \Pi^{Y M}(2), \\
\forall (x,y) \in \supp(p^{XY}): 
\rho(X)^{M}_{xy}
& \leq &
4 f(e) \rho^M_x, \\
\forall (x,y) \in \supp(p^{XY}): 
\rho(Y)^{M}_{xy}
& \leq &
4 f(e) \rho^M_y, \\
\forall (x,y) \in \supp(p^{XY}): 
\rho()^{M}_{xy}
& \leq &
(1 + g(e)) \rho^M, \\
\forall (x,y): \frac{p(x,y)}{q(x) q(y)} 
\|\Pi^{M}_{x y}(3) \, \rho^{M}_{xy}\|_\infty
& \leq &
\frac{
(1 + 4\sqrt{e} + 2\sqrt{\epsilon})
2^{
D^\epsilon_\infty(\rho^{XYM} \| q^X \otimes q^Y \otimes \rho^M)
}}{|F'_{\rho}|}, \\
\forall x: \frac{p(x)}{q(x)} 
\|\Pi^{M}_{x}(1) \, \rho^{M}_{x}\|_\infty
& \leq &
\frac{
(1 + 4\sqrt{e} + 2\sqrt{\epsilon})
2^{
D^\epsilon_\infty(\rho^{XM} \| q^X \otimes \rho^M)
}}{|F'_{\rho}|}, \\
\forall y: \frac{p(y)}{q(y)} 
\|\Pi^{M}_{y}(2) \, \rho^{M}_{y}\|_\infty
& \leq &
\frac{
(1 + 4\sqrt{e} + 2\sqrt{\epsilon})
2^{
D^\epsilon_\infty(\rho^{YM} \| q^Y \otimes \rho^M)
}}{|F'_{\rho}|}, \\
F'_{\rho}
& \leq &
M, \\
(1 - \sqrt{\epsilon}) \frac{\one^{F'_{\rho}}}{|F'_{\rho}|}
& \leq &
\rho^M
\;\leq\;
(1 + \sqrt{\epsilon}) \frac{\one^{F'_{\rho}}}{|F'_{\rho}|}.
\end{array}
\end{equation}
To finish the proof of Proposition~\ref{prop:coveringnonpairwise}, we 
only need to show 
\begin{equation}
\label{eq:softcovering8}
\E_{x^{(A)}, y^{(B)}}[\|\sigma^M_{x^{(A)}, y^{(B)}} - \rho^M\|_1]
\leq 15 \epsilon^{1/64} + \sqrt{40\epsilon^{1/32} + 8e + g(e)}.
\end{equation}
because of Equation~\ref{eq:flatteningfidelity}.

In order to prove the above inequality, we augment $X$ to $L_X X$,
$Y$ to $L_Y Y$ as in Section~\ref{subsec:tilting}. 
Recall that the new spaces $L_X$, $L_Y$ required for augmentation
have the same dimension $L$. The value of $L$ will be chosen later,
and it will be sufficiently large for our later purposes. Keeping
augmentation in mind, we define
\[
q^{L_X X} := \frac{\one^{L_X}}{L} \otimes q^X, ~
q^{L_Y Y} := \frac{\one^{L_Y}}{L} \otimes q^Y, ~~
\rho^{L_X X L_Y Y M} := 
\frac{\one^{L_X}}{L} \otimes \frac{\one^{L_Y}}{L} \otimes 
\rho^{XYM}.
\]

As in Section~\ref{subsec:tilting}, we enlarge the
Hilbert space $M$ into a larger space $\hM$. Because of this containment,
the state
$\rho^{L_X X L_Y Y \hM}$ is identical to $\rho^{L_X X L_Y Y M}$; which
notation we use depends on the application. 
Define tuples $(l_x x)^{(A)}$, $(l_y y)^{(B)}$ in the natural fashion. 
Define the normalised
state $\hrho^{L_X X L_Y Y \hM}$ and subnormalised state
$\hhrho^{L_X X L_Y Y \hM}$ as in  Equation~\ref{eq:hrhohhrho} via
the classical quantum projectors $\Pi^{XYM}(3)$, 
$\Pi^{XYM}(2) := \Pi^{YM}(2) \otimes \one^X$,
$\Pi^{XYM}(1) := \Pi^{XM}(1) \otimes \one^Y$. 
Observe that the states $\hrho^{L_X X L_Y Y \hM}$,
$\hhrho^{L_X X L_Y Y \hM}$, and
the approximate intersection projector
$\hPi^{L_X X L_Y Y \hM}$
of Equation~\ref{eq:approxintersection} are now classical quantum.
Thus we can write
\begin{equation}
\label{eq:nonpairwisecq}
\begin{array}{rcl}
\hrho^{L_X X L_Y Y \hM}
& := &
L^{-2} \sum_{l_x, x, l_y, y} p^{XY}(x,y)
\ketbra{l_x, x, l_y, y}^{L_X X L_Y Y} \otimes
\hrho^{\hM}_{l_x x l_y y}, \\
\hrho^{\hM}_{l_x x l_y y}
& = &
T^{M \rightarrow \hM}_{l_x,l_y,\epsilon^{1/4}}(\rho^M_{xy}), \\
\hhrho^{L_X X L_Y Y \hM}
& := &
L^{-2} \sum_{l_x, x, l_y, y} p^{XY}(x,y)
\ketbra{l_x, x, l_y, y}^{L_X X L_Y Y} \otimes
\hhrho^{\hM}_{l_x x l_y y}, \\
\forall i \in [3]: \hPi^{X Y M}(i)
& := &
\sum_{x, y} 
\ketbra{x, y}^{X Y} \otimes \Pi^{M}_{x y}(i), \\
\hPi^{L_X X L_Y Y \hM}
& := &
\sum_{l_x, x, l_y, y} 
\ketbra{l_x, x, l_y, y}^{L_X X L_Y Y} \otimes
\hPi^{\hM}_{l_x x l_y y}, \\
\hPi^{\hM}_{l_x x l_y y}
& = &
\one^{\hM} - 
\spanning\{
T^{M \rightarrow \hM}_{l_x,\epsilon^{1/4}}(
   (\one^{M} - \Pi^{M}_{x}(1))), 
T^{M \rightarrow \hM}_{l_y,\epsilon^{1/4}}(
   (\one^{M} - \Pi^{M}_{y}(2))), \\
&   &
~~~~~~~~~~~~~~~~~~~~~~~
T^{M \rightarrow \hM}_{l_x,l_y,\epsilon^{1/4}}(
   (\one^{M} - \Pi^{M}_{xy}(3)))
\}, \\
\hhrho^{\hM}_{l_x x l_y y}
& = &
\hPi^{\hM}_{l_x x l_y y} \circ \rho^{\hM}_{xy}.
\end{array}
\end{equation}
Define
\begin{equation}
\label{eq:softcovering9}
\begin{array}{rcl}
\sigma^{\hM}_{(l_x x)^{(A)}, (l_y y)^{(B)}} 
& := &
\sigma^M_{(l_x x)^{(A)}, (l_y y)^{(B)}} 
\;=\;
\sigma^M_{x^{(A)}, y^{(B)}} \\
& = &
(A B)^{-1} \sum_{a = 1}^{A} \sum_{b = 1}^{B}
\frac{p^{XY}(x(a), y(b))}{q^{X}(x(a)) q^{Y}(y(b))}
\rho^M_{x(a), y(b)}, \\
\hsigma^{\hM}_{(l_x x)^{(A)}, (l_y y)^{(B)}} 
& := &
(AB)^{-1} 
\sum_{a=1}^A \sum_{b=1}^B 
\frac{p^{XY}(x(a), y(b))}{q^{X}(x(a)) q^{Y}(y(b))}
\hrho^{\hM}_{(l_x x)(a), (l_y y)(b)}, \\
\hhsigma^{\hM}_{(l_x x)^{(A)}, (l_y y)^{(B)}} 
& := &
(AB)^{-1} 
\sum_{a=1}^A \sum_{b=1}^B 
\frac{p^{XY}(x(a), y(b))}{q^{X}(x(a)) q^{Y}(y(b))}
\hhrho^{\hM}_{(l_x x)(a), (l_y y)(b)}.
\end{array}
\end{equation}
Observe that 
\begin{equation}
\label{eq:softcoveringsupport}
\begin{array}{c}
\forall ((l_x x)^{(A)}, (l_y y)^{(B)}):
\supp(\sigma^{\hM}_{(l_x x)^{(A)}, (l_y y)^{(B)}}) \leq F'_\rho \leq M, \\
\|\rho^{XYM}\|_1 = \|\rho^M\|_1 = \|p^{XY}\|_1 \geq 1 - 4e.
\end{array}
\end{equation}

Define the probability distributions
\[
\bar{q}^{(L_X X)^A} := 
(\frac{\one^{L_X}}{L})^{\otimes A} \otimes \bar{q}^{X^A}, ~
\bar{q}^{(L_Y Y)^B} := 
(\frac{\one^{L_Y}}{L})^{\otimes B} \otimes \bar{q}^{Y^B}.
\]
Observe that 
\begin{eqnarray*}
\E_{(l_x x)^{(A)}, (l_y y)^{(B)}}[
\|\sigma^{\hM}_{(l_x x)^{(A)}, (l_y y)^{(B)}} - \rho^{\hM}\|_1]
& \equiv &
\E_{(l_x x)^{(A)}, (l_y y)^{(B)}}[
\|\sigma^{M}_{(l_x x)^{(A)}, (l_y y)^{(B)}} - \rho^{M}\|_1] \\
& = &
\E_{x^{(A)}, y^{(B)}}[
\|\sigma^{M}_{x^{(A)}, y^{(B)}} - \rho^{M}\|_1],
\end{eqnarray*}
where the expectation is taken over independent choices of tuples
$(l_x x)^{(A)}$, $(l_y y)^{(B)}$ from the distributions 
$\bar{q}^{(L_X X)^A}$, $\bar{q}^{(L_Y Y)^B}$.
So to prove Proposition~\ref{prop:coveringnonpairwise}, we just need to
show that
\begin{equation}
\label{eq:softcovering10}
\E_{(l_x x)^{(A)}, (l_y y)^{(B)}}[
\|\sigma^{\hM}_{(l_x x)^{(A)}, (l_y y)^{(B)}} - \rho^{\hM}\|_1]
\leq 15 \epsilon^{1/64} + \sqrt{40\epsilon^{1/32} + 8e + g(e)}.
\end{equation}
because of Equations~\ref{eq:softcovering8}, \ref{eq:softcovering9}.

We now prove four lemmas which are crucially required in the following
arguments, leading to the proof of
Proposition~\ref{prop:coveringnonpairwise}.
\begin{lemma}
\label{lem:nonhrhohhrho3}
\begin{eqnarray*}
\lefteqn{
L^{-2} \sum_{l_x, x, l_y, y} q^X(x) q^Y(y) 
\Tr\left[
\left(
\frac{p^{XY}(x, y)}{q^{X}(x) q^{Y}(y)}
\hrho^{\hM}_{l_x x l_y y}
\right)
\left(
\frac{p^{XY}(x, y)}{q^{X}(x) q^{Y}(y)}
\hhrho^{\hM}_{l_x x l_y y}
\right) 
\right]
} \\
& \leq &
\frac{(1 + 4\sqrt{e} + 2\sqrt{\epsilon})
      2^{D^\epsilon_\infty(\rho^{XYM}\| q^X\otimes q^Y\otimes \rho^M)}
     }{|F'_{\rho}|}.
~~~~~~~~~~~~~~~~~~~~~~~~~~~~~~~
\end{eqnarray*}
\end{lemma}
\begin{proof}
Using Equations~\ref{eq:nonpairwisecq}, \ref{eq:softcovering7} and
arguing as in the proof of Lemma~\ref{lem:hrhohhrho3}, we get
\begin{eqnarray*}
\lefteqn{
L^{-2} \sum_{l_x, x, l_y, y} q^X(x) q^Y(y) 
\Tr\left[
\left(
\frac{p^{XY}(x, y)}{q^{X}(x) q^{Y}(y)}
\hrho^{\hM}_{l_x x l_y y}
\right)
\left(
\frac{p^{XY}(x, y)}{q^{X}(x) q^{Y}(y)}
\hhrho^{\hM}_{l_x x l_y y}
\right) 
\right]
} \\
& = &
L^{-2} \sum_{l_x, x, l_y, y} 
\Tr\left[
\left(
\frac{p^{XY}(x, y)}{q^{X}(x) q^{Y}(y)}
\hrho^{\hM}_{l_x x l_y y}
\right)
(p^{XY}(x, y) \hPi^{\hM}_{l_x x l_y y} \circ \rho^{\hM}_{x y})
\right] \\
& = &
L^{-2} \sum_{l_x, x, l_y, y} 
\Tr\left[
\left(
\frac{p^{XY}(x, y)}{q^{X}(x) q^{Y}(y)}
 T^{M \rightarrow \hM}_{l_x,l_y,\epsilon^{1/4}}(\rho^{M}_{x y})
\right)
\right. \\
&  &
~~~~~~~~~~~~~~~~~~~~~~~
\left.
(p^{XY}(x, y)
((\one^{\hM} - 
  T^{M \rightarrow \hM}_{l_x,l_y,\epsilon^{1/4}}(
	(\one^{M} - \Pi^{M}_{xy}(3)))) \circ
 (\hPi^{\hM}_{l_x x l_y y} \circ \rho^{\hM}_{x y})
)  
\right] \\
& = &
L^{-2} \sum_{l_x, x, l_y, y} 
\Tr\left[
\left(
\frac{p^{XY}(x, y)}{q^{X}(x) q^{Y}(y)}
 (T^{M \rightarrow \hM}_{l_x,l_y,\epsilon^{1/4}}(\one^{M})) \circ
 (T^{M \rightarrow \hM}_{l_x,l_y,\epsilon^{1/4}}(\rho^{M}_{x y}))
\right) 
\right. \\
&  &
~~~~~~~~~~~~~~~~~~~~~~~
\left.
(
p^{XY}(x, y)
((\one^{\hM} - 
  T^{M \rightarrow \hM}_{l_x,l_y,\epsilon^{1/4}}(
	(\one^{M} - \Pi^{M}_{xy}(3)))) \circ \hhrho^{\hM}_{l_x x l_y y})
)  
\right] \\
& = &
L^{-2} \sum_{l_x, x, l_y, y} 
\Tr\left[
\left(
\frac{p^{XY}(x, y)}{q^{X}(x) q^{Y}(y)}
 T^{M \rightarrow \hM}_{l_x,l_y,\epsilon^{1/4}}(\rho^{M}_{x y})
\right) 
\right. \\
&  &
~~~~~~~~~~~~~~~~~~~~~~~
\left.
(
p^{XY}(x, y)
((T^{M \rightarrow \hM}_{l_x,l_y,\epsilon^{1/4}}(\one^{M}) -
  T^{M \rightarrow \hM}_{l_x,l_y,\epsilon^{1/4}}(
	(\one^{M} - \Pi^{M}_{xy}(3)))) \circ \hhrho^{\hM}_{l_x x l_y y})
) 
\right] \\
& = &
L^{-2} \sum_{l_x, x, l_y, y} 
\Tr\left[
\left(
\frac{p^{XY}(x, y)}{q^{X}(x) q^{Y}(y)}
 T^{M \rightarrow \hM}_{l_x,l_y,\epsilon^{1/4}}(\rho^{M}_{x y})
\right) 
(
p^{XY}(x, y)
((T^{M \rightarrow \hM}_{l_x,l_y,\epsilon^{1/4}}(
	\Pi^{M}_{xy}(3))) \circ \hhrho^{\hM}_{l_x x l_y y})
) 
\right] \\
& = &
L^{-2} \sum_{l_x, x, l_y, y} 
\Tr\left[
\left(
\frac{p^{XY}(x, y)}{q^{X}(x) q^{Y}(y)} 
 T^{M \rightarrow \hM}_{l_x,l_y,\epsilon^{1/4}}(
    \Pi^M_{xy}(3) \circ \rho^{M}_{x y})
\right) 
(p^{XY}(x, y) \hhrho^{\hM}_{l_x x l_y y}) 
\right] \\
& \leq &
L^{-2} \sum_{l_x, x, l_y, y} 
\frac{p^{XY}(x, y)}{q^{X}(x) q^{Y}(y)}
\|T^{M \rightarrow \hM}_{l_x,l_y,\epsilon^{1/4}}(
    \Pi^M_{xy}(3) \circ \rho^{M}_{x y})\|_\infty \cdot 
(p^{XY}(x, y) \|\hhrho^{\hM}_{l_x x l_y y}\|_1)  \\
& = &
L^{-2} \sum_{l_x, x, l_y, y} 
\frac{p^{XY}(x, y)}{q^{X}(x) q^{Y}(y)}
\|\Pi^M_{xy}(3) \circ \rho^{M}_{x y})\|_\infty \cdot 
(p^{XY}(x, y) \|\hhrho^{\hM}_{l_x x l_y y}\|_1)  \\
& \leq &
L^{-2} \sum_{l_x, x, l_y, y} 
\frac{(1+4\sqrt{e}+2\sqrt{\epsilon}) 
      2^{D^\epsilon_\infty(\rho^{XYM}\|q^X\otimes q^Y \otimes \rho^M)}
     }{|F'_{\rho}|} 
p^{XY}(x, y) \\
& \leq &
\frac{(1+4\sqrt{e}+2\sqrt{\epsilon}) 
      2^{D^\epsilon_\infty(\rho^{XYM}\|q^X\otimes q^Y \otimes \rho^M)}
     }{|F'_{\rho}|}.
\end{eqnarray*}
The proof of the lemma is now complete.
\end{proof}

\begin{lemma}
\label{lem:nonhrhohhrho1}
\begin{eqnarray*}
\lefteqn{
L^{-3} \sum_{l_x, x, l_y, y, l'_y, y'} q^X(x) \bar{q}^{YY'}(y, y')
\Tr\left[
\left(
\frac{p^{XY}(x, y')}{q^{X}(x) q^{Y}(y')}
\hrho^{\hM}_{l_x x l'_y y'}
\right)
\left(
\frac{p^{XY}(x, y)}{q^{X}(x) q^{Y}(y)}
\hhrho^{\hM}_{l_x x l_y y}
\right) 
\right]
} \\
& \leq &
4 \epsilon^{1/8} L^{-1/2} \sum_{x, y} 
p^{XY}(x, y) \frac{p^X(x)}{q^X(x)} \|\rho(X)^M_{xy}\|_1 +
\frac{4 (1+4\sqrt{e}+2\sqrt{\epsilon}) f(e) 
      2^{D^\epsilon_\infty(\rho^{XM}\|q^X \otimes \rho^M)}
     }{|F'_{\rho}|}.
\end{eqnarray*}
\end{lemma}
\begin{proof}
Using Equations~\ref{eq:nonpairwisecq}, \ref{eq:softcovering7} and
Lemma~\ref{lem:smoothing}, and
arguing as in the proof of Lemma~\ref{lem:hrhohhrho1}, we get
\begin{eqnarray*}
\lefteqn{
L^{-3} \sum_{l_x, x, l_y, y, l'_y, y'} q^X(x) \bar{q}^{YY'}(y, y')
\Tr\left[
\left(
\frac{p^{XY}(x, y')}{q^{X}(x) q^{Y}(y')}
\hrho^{\hM}_{l_x x l'_y y'}
\right)
\left(
\frac{p^{XY}(x, y)}{q^{X}(x) q^{Y}(y)}
\hhrho^{\hM}_{l_x x l_y y}
\right) 
\right]
} \\
& = &
L^{-3} \sum_{l_x, x, l_y, y, l'_y, y'} 
\frac{p^X(x)}{q^X(x)} \cdot 
\frac{p^{Y|X=x}(y') \bar{q}^{Y'|Y=y}(y')}{q^Y(y')} 
\Tr\left[
\hrho^{\hM}_{l_x x l'_y y'}
(p^{XY}(x, y) \hhrho^{\hM}_{l_x x l_y y})  
\right] \\
& = &
L^{-3} \sum_{l_x, x, l_y, y, l'_y, y'} 
\frac{p^X(x)}{q^X(x)} \cdot
\frac{p^{Y|X=x}(y') \bar{q}^{Y'|Y=y}(y')}{q^Y(y')} 
\Tr\left[
(T^{M \rightarrow \hM}_{l_x,l'_y,\epsilon^{1/4}}(\rho^{M}_{x y'})) 
(p^{XY}(x, y) \hhrho^{\hM}_{l_x x l_y y})  
\right] \\
& = &
L^{-2} \sum_{l_x, x, l_y, y} 
\frac{p^X(x)}{q^X(x)} 
\Tr\left[
\left(
L^{-1} \sum_{l'_y y'}
\frac{p^{Y|X=x}(y') \bar{q}^{Y'|Y=y}(y')}{q^Y(y')} 
T^{M \rightarrow \hM}_{l_x,l'_y,\epsilon^{1/4}}(\rho^{M}_{x y'}) 
\right)
(p^{XY}(x, y) \hhrho^{\hM}_{l_x x l_y y})  
\right] \\
& = &
L^{-2} \sum_{l_x, x, l_y, y} 
\frac{p^X(x)}{q^X(x)} 
\Tr\left[
\left(
\left(
L^{-1} \sum_{l'_y}
T^{M \rightarrow \hM}_{l_x,l'_y,\epsilon^{1/4}}\left(
   \sum_{y'} \frac{p^{Y|X=x}(y') \bar{q}^{Y'|Y=y}(y')}{q^Y(y')} 
	     \rho^{M}_{x y'}
\right) 
\right) 
\right.
\right. \\
& &
~~~~~~~~~~~~~~~~~~~~~~~~~~~~~~~~~~
\left.
\left.
{} -
T^{M \rightarrow \hM}_{l_x,\epsilon^{1/4}}\left(
   \sum_{y'} \frac{p^{Y|X=x}(y') \bar{q}^{Y'|Y=y}(y')}{q^Y(y')} 
	     \rho^{M}_{x y'}
\right) 
\right)
(p^{XY}(x, y) \hhrho^{\hM}_{l_x x l_y y})  
\right] \\
&  &
{} +
L^{-2} \sum_{l_x, x, l_y, y} 
\frac{p^X(x)}{q^X(x)} 
\Tr\left[
\left(
T^{M \rightarrow \hM}_{l_x,\epsilon^{1/4}}\left(
   \sum_{y'} \frac{p^{Y|X=x}(y') \bar{q}^{Y'|Y=y}(y')}{q^Y(y')} 
	     \rho^{M}_{x y'}
\right) 
\right)
(p^{XY}(x, y) \hhrho^{\hM}_{l_x x l_y y})  
\right] \\
& \leq &
L^{-2} \sum_{l_x, x, l_y, y} 
\frac{p^X(x)}{q^X(x)} 
\left\|
\left(
L^{-1} \sum_{l'_y}
T^{M \rightarrow \hM}_{l_x,l'_y,\epsilon^{1/4}}(\rho(X)^M_{xy}) 
\right) -
T^{M \rightarrow \hM}_{l_x,\epsilon^{1/4}}(\rho(X)^M_{xy}) 
\right\|_\infty \cdot
(p^{XY}(x, y) \|\hhrho^{\hM}_{l_x x l_y y}\|_1)  \\
&  &
{} +
L^{-2} \sum_{l_x, x, l_y, y} 
\frac{p^X(x)}{q^X(x)} 
\Tr\left[
(T^{M \rightarrow \hM}_{l_x,\epsilon^{1/4}}(\rho(X)^{M}_{x y}))
(p^{XY}(x, y) \hhrho^{\hM}_{l_x x l_y y})  
\right] \\
& \leq &
4 \epsilon^{1/8} L^{-5/2} \sum_{l_x, x, l_y, y} 
p^{XY}(x, y) \frac{p^X(x)}{q^X(x)} \|\rho(X)^M_{xy}\|_1 \\
&  &
~~~
{} + L^{-2} \sum_{l_x, x, l_y, y} 
\frac{p^X(x)}{q^X(x)} 
\Tr\left[
(T^{M \rightarrow \hM}_{l_x,\epsilon^{1/4}}(\rho(X)^{M}_{x y}))
(p^{XY}(x, y) (\hPi^{\hM}_{l_x x l_y y} \circ \rho^{\hM}_{x y}))
\right] \\
& = &
4 \epsilon^{1/8} L^{-1/2} \sum_{x, y} 
p^{XY}(x, y) \frac{p^X(x)}{q^X(x)} \|\rho(X)^M_{xy}\|_1 \\
&  &
~~~
{} + 
L^{-2} \sum_{l_x, x, l_y, y} 
p^{XY}(x, y) \cdot \frac{p^X(x)}{q^X(x)} \\
&  &
~~~~~~~~~~~~~~~~~~~~~~~
\Tr\left[
T^{M \rightarrow \hM}_{l_x,\epsilon^{1/4}}(\rho(X)^{M}_{x y}) 
((\one^{\hM} - 
  T^{M \rightarrow \hM}_{l_x,\epsilon^{1/4}}(
	(\one^{M} - \Pi^{M}_{x}(1)))) \circ
 (\hPi^{\hM}_{l_x x l_y y} \circ \rho^{\hM}_{x y}))
\right] \\
& = &
4 \epsilon^{1/8} L^{-1/2} \sum_{x, y} 
p^{XY}(x, y) \frac{p^X(x)}{q^X(x)} \|\rho(X)^M_{xy}\|_1 \\
&  &
~~~
{} + 
L^{-2} \sum_{l_x, x, l_y, y} 
p^{XY}(x, y) \cdot \frac{p^X(x)}{q^X(x)} \\
&  &
~~~~~~~~~~~~~~~~~~~~~~~
\Tr[
(
 (T^{M \rightarrow \hM}_{l_x,\epsilon^{1/4}}(\one^{M})) \circ
 (T^{M \rightarrow \hM}_{l_x,\epsilon^{1/4}}(\rho(X)^{M}_{x y}))
) \\
&  &
~~~~~~~~~~~~~~~~~~~~~~~~~~~~~~~
((\one^{\hM} - 
  T^{M \rightarrow \hM}_{l_x,\epsilon^{1/4}}(
	(\one^{M} - \Pi^{M}_{x}(1)))) \circ \hhrho^{\hM}_{l_x x l_y y}
)
] \\
& = &
4 \epsilon^{1/8} L^{-1/2} \sum_{x, y} 
p^{XY}(x, y) \frac{p^X(x)}{q^X(x)} \|\rho(X)^M_{xy}\|_1 \\
&  &
~~~
{} + 
L^{-2} \sum_{l_x, x, l_y, y} 
p^{XY}(x, y) \cdot \frac{p^X(x)}{q^X(x)} \\
&  &
~~~~~~~~~~~~~~~~~~~~~~~
\Tr[
(T^{M \rightarrow \hM}_{l_x,\epsilon^{1/4}}(\rho(X)^{M}_{x y}))
(
(T^{M \rightarrow \hM}_{l_x,\epsilon^{1/4}}(\one^{M}) -
  T^{M \rightarrow \hM}_{l_x,\epsilon^{1/4}}(
	(\one^{M} - \Pi^{M}_{x}(1)))) \circ \hhrho^{\hM}_{l_x x l_y y}
) 
] \\
& = &
4 \epsilon^{1/8} L^{-1/2} \sum_{x, y} 
p^{XY}(x, y) \frac{p^X(x)}{q^X(x)} \|\rho(X)^M_{xy}\|_1 \\
&  &
~~~
{} + 
L^{-2} \sum_{l_x, x, l_y, y} 
p^{XY}(x, y) \cdot \frac{p^X(x)}{q^X(x)} 
\Tr[
(T^{M \rightarrow \hM}_{l_x,\epsilon^{1/4}}(\rho(X)^{M}_{x y})) 
(
(T^{M \rightarrow \hM}_{l_x,\epsilon^{1/4}}(
	\Pi^{M}_{x}(1))) \circ \hhrho^{\hM}_{l_x x l_y y}
) 
] \\
& = &
4 \epsilon^{1/8} L^{-1/2} \sum_{x, y} 
p^{XY}(x, y) \frac{p^X(x)}{q^X(x)} \|\rho(X)^M_{xy}\|_1 \\
&  &
~~~
{} + 
L^{-2} \sum_{l_x, x, l_y, y} 
p^{XY}(x, y) \cdot \frac{p^X(x)}{q^X(x)} 
\Tr[
(T^{M \rightarrow \hM}_{l_x,\epsilon^{1/4}}(
    \Pi^M_{x}(1) \circ \rho(X)^{M}_{x y})
) 
\hhrho^{\hM}_{l_x x l_y y}) 
] \\
& \leq &
4 \epsilon^{1/8} L^{-1/2} \sum_{x, y} 
p^{XY}(x, y) \frac{p^X(x)}{q^X(x)} \|\rho(X)^M_{xy}\|_1 \\
&  &
~~~
{} + 
4 f(e) L^{-2} \sum_{l_x, x, l_y, y} 
p^{XY}(x, y) \cdot \frac{p^X(x)}{q^X(x)} 
\Tr[
(T^{M \rightarrow \hM}_{l_x,\epsilon^{1/4}}(
    \Pi^M_{x}(1) \circ \rho^{M}_{x})
) 
\hhrho^{\hM}_{l_x x l_y y}) 
] \\
& \leq &
4 \epsilon^{1/8} L^{-1/2} \sum_{x, y} 
p^{XY}(x, y) \frac{p^X(x)}{q^X(x)} \|\rho(X)^M_{xy}\|_1 \\
&  &
~~~
{} + 
4 f(e) L^{-2} \sum_{l_x, x, l_y, y} 
p^{XY}(x, y) \cdot \frac{p^X(x)}{q^X(x)} 
\|T^{M \rightarrow \hM}_{l_x,\epsilon^{1/4}}(
    \Pi^M_{x}(1) \circ \rho^{M}_{x})\|_\infty \cdot 
\|\hhrho^{\hM}_{l_x x l_y y}\|_1  \\
& = &
4 \epsilon^{1/8} L^{-1/2} \sum_{x, y} 
p^{XY}(x, y) \frac{p^X(x)}{q^X(x)} \|\rho(X)^M_{xy}\|_1 \\
&  &
~~~
{} + 
4 f(e) L^{-2} \sum_{l_x, x, l_y, y} 
p^{XY}(x, y) \cdot \frac{p^X(x)}{q^X(x)} 
\|\Pi^M_{x}(1) \circ \rho^{M}_{x })\|_\infty \cdot 
\|\hhrho^{\hM}_{l_x x l_y y}\|_1  \\
& \leq &
4 \epsilon^{1/8} L^{-1/2} \sum_{x, y} 
p^{XY}(x, y) \frac{p^X(x)}{q^X(x)} \|\rho(X)^M_{xy}\|_1 \\
&  &
~~~
{} + 
4 f(e) L^{-2} \sum_{l_x, x, l_y, y} 
p^{XY}(x, y) \cdot 
\frac{(1+4\sqrt{e}+2\sqrt{\epsilon}) 
      2^{D^\epsilon_\infty(\rho^{XM}\|q^X \otimes \rho^M)}
     }{|F'_{\rho}|}  \\
& \leq &
4 \epsilon^{1/8} L^{-1/2} \sum_{x, y} 
p^{XY}(x, y) \frac{p^X(x)}{q^X(x)} \|\rho(X)^M_{xy}\|_1 +
\frac{4 (1+4\sqrt{e}+2\sqrt{\epsilon}) f(e)
      2^{D^\epsilon_\infty(\rho^{XM}\|q^X \otimes \rho^M)}
     }{|F'_{\rho}|}.
\end{eqnarray*}
The proof of the lemma is now complete.
\end{proof}

\begin{lemma}
\label{lem:nonhrhohhrho2}
\begin{eqnarray*}
\lefteqn{
L^{-3} \sum_{l_x, x, l_y, y, l'_x, x'} q^Y(y) \bar{q}^{XX'}(x, x')
\Tr\left[
\left(
\frac{p^{XY}(x', y)}{q^{X}(x') q^{Y}(y)}
\hrho^{\hM}_{l'_x x' l_y y}
\right)
\left(
\frac{p^{XY}(x, y)}{q^{X}(x) q^{Y}(y)}
\hhrho^{\hM}_{l_x x l_y y}
\right) 
\right]
} \\
& \leq &
4 \epsilon^{1/8} L^{-1/2} \sum_{x, y} 
p^{XY}(x, y) \frac{p^Y(y)}{q^Y(y)} \|\rho(Y)^M_{xy}\|_1 +
\frac{4 (1+4\sqrt{e}+2\sqrt{\epsilon}) f(e) 
      2^{D^\epsilon_\infty(\rho^{YM}\|q^Y \otimes \rho^M)}
     }{|F'_{\rho}|}.
\end{eqnarray*}
\end{lemma}
\begin{proof}
Similar to proof of Lemma~\ref{lem:nonhrhohhrho1} above.
\end{proof}

\begin{lemma}
\label{lem:nonhrhohhrho}
\begin{eqnarray*}
\lefteqn{
L^{-4} \sum_{l_x, x, l_y, y, l'_x, x', l'_y, y'} 
\bar{q}^{X X'}(x, x') \bar{q}^{YY'}(y, y')
\Tr\left[
\left(
\frac{p^{XY}(x', y')}{q^{X}(x') q^{Y}(y')}
\hrho^{\hM}_{l'_x x' l'_y y'}
\right)
\left(
\frac{p^{XY}(x, y)}{q^{X}(x) q^{Y}(y)}
\hhrho^{\hM}_{l_x x l_y y}
\right) 
\right]
} \\
& \leq &
8 \epsilon^{1/8} L^{-1/2} \sum_{x, y} p^{XY}(x,y) \|\rho()^M_{xy}\|_1 +
\frac{(1+g(e) + 2 \sqrt{\epsilon}) 
     }{|F'_{\rho}|}.
~~~~~~~~~~~~~~~
\end{eqnarray*}
\end{lemma}
\begin{proof}
Using Equations~\ref{eq:nonpairwisecq}, \ref{eq:softcovering7} and
Lemma~\ref{lem:smoothing}, and
arguing as in the proof of Lemma~\ref{lem:hrhohhrho}, we get
\begin{eqnarray*}
\lefteqn{
L^{-4} \sum_{l_x, x, l_y, y, l'_x, x', l'_y, y'} 
\bar{q}^{XX'}(x,x') \bar{q}^{YY'}(y, y')
\Tr\left[
\left(
\frac{p^{XY}(x', y')}{q^{X}(x') q^{Y}(y')}
\hrho^{\hM}_{l'_x x' l'_y y'}
\right)
\left(
\frac{p^{XY}(x, y)}{q^{X}(x) q^{Y}(y)}
\hhrho^{\hM}_{l_x x l_y y}
\right) 
\right]
} \\
& = &
L^{-4} \sum_{l_x, x, l_y, y, l'_x, x', l'_y, y'} 
\frac{p^{XY}(x',y') \bar{q}^{X'|X=x}(x') \bar{q}^{Y'|Y=y}(y')}
     {q^X(x') q^Y(y')} 
\Tr\left[
\hrho^{\hM}_{l'_x x' l'_y y'}
(p^{XY}(x, y) \hhrho^{\hM}_{l_x x l_y y})  
\right] \\
& = &
L^{-4} \sum_{l_x, x, l_y, y, l'_x, x', l'_y, y'} 
\frac{p^{XY}(x',y') \bar{q}^{X'|X=x}(x') \bar{q}^{Y'|Y=y}(y')}
     {q^X(x') q^Y(y')} 
\Tr\left[
(T^{M \rightarrow \hM}_{l'_x,l'_y,\epsilon^{1/4}}(\rho^{M}_{x' y'})) 
(p^{XY}(x, y) \hhrho^{\hM}_{l_x x l_y y})  
\right] \\
& = &
L^{-2} \sum_{l_x, x, l_y, y} 
\Tr\left[
\left(
L^{-2} \sum_{l'_x x' l'_y y'}
\frac{p^{XY}(x',y') \bar{q}^{X'|X=x}(x') \bar{q}^{Y'|Y=y}(y')}
     {q^X(x') q^Y(y')} 
T^{M \rightarrow \hM}_{l'_x,l'_y,\epsilon^{1/4}}(\rho^{M}_{x' y'}) 
\right) 
\right. \\
&  &
\left.
~~~~~~~~~~~~~~~~~~~~~~~~~~~~~~~
(p^{XY}(x, y) \hhrho^{\hM}_{l_x x l_y y})  
\right] \\
& = &
L^{-2} \sum_{l_x, x, l_y, y} 
\Tr\left[
\left(
\left(
L^{-2} \sum_{l'_x l'_y}
T^{M \rightarrow \hM}_{l'_x,l'_y,\epsilon^{1/4}}\left(
   \sum_{x' y'} 
   \frac{p^{XY}(x',y') \bar{q}^{X'|X=x}(x') \bar{q}^{Y'|Y=y}(y')}
        {q^X(x') q^Y(y')} \rho^{M}_{x' y'}
\right) 
\right) 
\right.
\right. \\
& &
~~~~~~~~~~~~~~~~~~~~~~~~~~~~~
\left.
\left.
{} -
\left(
\sum_{x' y'} 
\frac{p^{XY}(x',y') \bar{q}^{X'|X=x}(x') \bar{q}^{Y'|Y=y}(y')}
     {q^X(x') q^Y(y')} \rho^{M}_{x' y'}
\right) 
\right)
(p^{XY}(x, y) \hhrho^{\hM}_{l_x x l_y y})  
\right] \\
&  &
{} +
L^{-2} \sum_{l_x, x, l_y, y} 
\Tr\left[
\left(
\sum_{x' y'} 
\frac{p^{XY}(x',y') \bar{q}^{X'|X=x}(x') \bar{q}^{Y'|Y=y}(y')}
     {q^X(x') q^Y(y')} \rho^{M}_{x' y'}
\right) 
(p^{XY}(x, y) \hhrho^{\hM}_{l_x x l_y y})  
\right] \\
& \leq &
L^{-2} \sum_{l_x, x, l_y, y} 
\left\|
\left(
L^{-2} \sum_{l'_x l'_y}
T^{M \rightarrow \hM}_{l'_x,l'_y,\epsilon^{1/4}}(\rho()^M_{xy}) 
\right) -
\rho()^M_{xy} 
\right\|_\infty \cdot
(p^{XY}(x, y) \|\hhrho^{\hM}_{l_x x l_y y}\|_1)  \\
&  &
{} +
L^{-2} \sum_{l_x, x, l_y, y} 
p^{XY}(x, y)
\Tr[\rho()^{M}_{x y} \, \hhrho^{\hM}_{l_x x l_y y}] \\
& \leq &
8 \epsilon^{1/8} L^{-5/2} 
\sum_{l_x, x, l_y, y} p^{XY}(x,y) \|\rho()^M_{xy}\|_1 + 
(1 + g(e)) L^{-2} \sum_{l_x, x, l_y, y} 
p^{XY}(x, y)
\Tr[\rho^{M} \hhrho^{\hM}_{l_x x l_y y}] \\
& \leq &
8 \epsilon^{1/8} L^{-1/2} 
\sum_{x, y} p^{XY}(x,y) \|\rho()^M_{xy}\|_1 + 
(1 + g(e))(1+\sqrt{\epsilon}) L^{-2} \sum_{l_x, x, l_y, y} 
p^{XY}(x, y)
\frac{\Tr[\one^M \hhrho^{\hM}_{l_x x l_y y}]}{|F'_\rho|} \\
& \leq &
8 \epsilon^{1/8} L^{-1/2} 
\sum_{x, y} p^{XY}(x,y) \|\rho()^M_{xy}\|_1 + 
(1 + g(e)+2\sqrt{\epsilon}) L^{-2} \sum_{l_x, x, l_y, y} 
p^{XY}(x, y)
\frac{\|\hhrho^{\hM}_{l_x x l_y y}\|_1}{|F'_\rho|} \\
& \leq &
8 \epsilon^{1/8} L^{-1/2} \sum_{x, y} 
p^{XY}(x, y) \|\rho()^M_{xy}\|_1 +
\frac{(1+g(e) + 2\sqrt{\epsilon})}{|F'_{\rho}|}.
\end{eqnarray*}
The proof of the lemma is now complete.
\end{proof}

The reader may have noticed some differences between the statements of
Lemma~\ref{lem:nonhrhohhrho} and Lemma~\ref{lem:hrhohhrho},
Lemma~\ref{lem:nonhrhohhrho1} and Lemma~\ref{lem:hrhohhrho1},
Lemma~\ref{lem:nonhrhohhrho2} and Lemma~\ref{lem:hrhohhrho2}. 
The notable looking differences arise from corrections due to
augmentation smoothing viz. from the application of
Lemma~\ref{lem:smoothing}. Part of the reason for these differences
is the presence of terms like $\rho(X)_{xy}$, $\rho(Y)_{xy}$ etc.
in the non-pairwise independent case, which do not exist in the
pairwise independent case. The other part of the reason for the
differences is cosmetic: the proof of the fully quantum smooth convex 
split lemma Lemma~\ref{lem:convexsplitflattening}
uses flattening for $\alpha^X$ and $\beta^Y$, whereas in the non pairwise
independent classical quantum covering lemma 
Propostion~\ref{prop:coveringnonpairwise}, we 
do not need to flatten their classical analogues $q^X$ and $q^Y$. 
Nevertheless, an effect similar to flattening $\alpha^X$, $\beta^Y$
persists in the statements of Lemmas~\ref{lem:nonhrhohhrho1},
\ref{lem:nonhrhohhrho2}. Terms like $\frac{1}{q^X(x)}$, 
$\frac{1}{q^Y(y)}$ in the statements of 
Lemmas~\ref{lem:nonhrhohhrho1}, \ref{lem:nonhrhohhrho2} are
related to flattening terms like $\frac{1}{F'_\alpha}$,
$\frac{1}{F'_\beta}$ in the statements of Lemmas~\ref{lem:hrhohhrho1},
\ref{lem:hrhohhrho2}.

By Lemma~\ref{lem:tiltclose} and Equation~\ref{eq:softcovering9}, we have
\begin{equation}
\label{eq:softcoveringtilt}
\begin{array}{rcl}
\lefteqn{
\E_{(l_x x)^{(A)}, (l_y y)^{(B)}}\left[
\left\|
   \sigma^{\hM}_{(l_x x)^{(A)}, (l_y y)^{(B)}}] -
   \hsigma^{\hM}_{(l_x x)^{(A)}, (l_y y)^{(B)}}
\right\|_1 
\right]
} \\
&  =   &
(AB)^{-1} 
\E_{(l_x x)^{(A)}, (l_y y)^{(B)}}\left[
\left\|
\sum_{a=1}^A \sum_{b=1}^B 
\frac{p^{XY}(x(a), y(b))}{q^{X}(x(a)) q^{Y}(y(b))}
\rho^{\hM}_{(l_x x)(a), (l_y y)(b)} - {} 
\right.
\right. \\
&  &
~~~~~~~~~~~~~~~~~~~~~~~~~~~~~~~~~~~~~~~~~~~
\left.
\left.
\sum_{a=1}^A \sum_{b=1}^B 
\frac{p^{XY}(x(a), y(b))}{q^{X}(x(a)) q^{Y}(y(b))}
\hrho^{\hM}_{(l_x x)(a), (l_y y)(b)} 
\right]
\right\|_1 \\
& \leq &
(AB)^{-1} 
\sum_{a=1}^A \sum_{b=1}^B \\
&  &
~~~~~~~~~~~~~~~
\E_{(l_x x)^{(A)}, (l_y y)^{(B)}}\left[
\frac{p^{XY}(x(a), y(b))}{q^{X}(x(a)) q^{Y}(y(b))}
\left\|
\rho^{\hM}_{(l_x x)(a), (l_y y)(b)} - 
\hrho^{\hM}_{(l_x x)(a), (l_y y)(b)} 
\right\|_1 
\right] \\
&   =  &
L^{-2} \sum_{l_x, x, l_y, y}
q^X(x) q^Y(y)
\frac{p^{XY}(x, y)}{q^{X}(x) q^{Y}(y)}
\left\|
\rho^{\hM}_{l_x, x, l_y, y} - 
\hrho^{\hM}_{l_x, x, l_y, y} 
\right\|_1 \\
&   =  &
L^{-2} \sum_{l_x, x, l_y, y} p^{XY}(x, y)
\left\|
\rho^{\hM}_{x, y} - 
\hrho^{\hM}_{l_x, x, l_y, y} 
\right\|_1 \\
& = &
\|\rho^{L_X X L_Y Y \hM} - \hrho^{L_X X L_Y Y \hM}\|_1 
\;\leq \;
2\sqrt{2} \cdot \epsilon^{1/8}.
\end{array}
\end{equation}
Hence by triangle inequality and Equation~\ref{eq:softcovering10}, in 
order to prove 
Proposition~\ref{prop:coveringnonpairwise}, it suffices to show that
\begin{equation}
\label{eq:softcovering12}
\E_{(l_x x)^{(A)}, (l_y y)^{(B)}}[
\|\hsigma^{\hM}_{(l_x x)^{(A)}, (l_y y)^{(B)}} - \rho^{\hM}\|_1
]
\leq 12 \epsilon^{1/64} + \sqrt{40\epsilon^{1/32} + 8e + g(e)}.
\end{equation}

By Corollary~\ref{cor:triangleineq} and Equation~\ref{eq:softcovering9}, 
we have
\begin{eqnarray*}
\lefteqn{
\E_{(l_x x)^{(A)}, (l_y y)^{(B)}}\left[
\left\|
\hsigma^{\hM}_{(l_x x)^{(A)}, (l_y y)^{(B)}} -
\hhsigma^{\hM}_{(l_x x)^{(A)}, (l_y y)^{(B)}}
\right\|_1 
\right]
} \\
&  =   &
(AB)^{-1} 
\E_{(l_x x)^{(A)}, (l_y y)^{(B)}}\left[
\left\|
\sum_{a=1}^A \sum_{b=1}^B 
\frac{p^{XY}(x(a), y(b))}{q^{X}(x(a)) q^{Y}(y(b))}
\hrho^{\hM}_{(l_x x)(a), (l_y y)(b)} - {} 
\right.
\right. \\
&  &
~~~~~~~~~~~~~~~~~~~~~~~~~~~~~~~~~~~~~~~~~~~
\left.
\left.
\sum_{a=1}^A \sum_{b=1}^B 
\frac{p^{XY}(x(a), y(b))}{q^{X}(x(a)) q^{Y}(y(b))}
\hhrho^{\hM}_{(l_x x)(a), (l_y y)(b)} 
\right\|_1 
\right] \\
& \leq &
(AB)^{-1} 
\sum_{a=1}^A \sum_{b=1}^B 
\E_{(l_x x)^{(A)}, (l_y y)^{(B)}}\left[
\frac{p^{XY}(x(a), y(b))}{q^{X}(x(a)) q^{Y}(y(b))}
\left\|
\hrho^{\hM}_{(l_x x)(a), (l_y y)(b)} - 
\hhrho^{\hM}_{(l_x x)(a), (l_y y)(b)} 
\right\|_1 
\right] \\
&   =  &
L^{-2} \sum_{l_x, x, l_y, y}
q^X(x) q^Y(y)
\frac{p^{XY}(x, y)}{q^{X}(x) q^{Y}(y)}
\left\|
\hrho^{\hM}_{l_x, x, l_y, y} - 
\hhrho^{\hM}_{l_x, x, l_y, y} 
\right\|_1 \\
&   =  &
L^{-2} \sum_{l_x, x, l_y, y} p^{XY}(x, y)
\left\|
\hrho^{\hM}_{l_x, x, l_y, y} - 
\hhrho^{\hM}_{l_x, x, l_y, y} 
\right\|_1 
\;=\;
\|\hrho^{L_X X L_Y Y \hM} - \hhrho^{L_X X L_Y Y \hM}\|_1 
\;<\;
25 \cdot \epsilon^{1/8}.
\end{eqnarray*}
Define 
\[
\mathrm{Bad} :=
\{
((l_x x)^{(A)}, (l_y y)^{(B)}): 
\|\hsigma^{\hM}_{(l_x x)^{(A)}, (l_y y)^{(B)}} -
  \hhsigma^{\hM}_{(l_x x)^{(A)}, (l_y y)^{(B)}}\|_1 > 25 \epsilon^{1/16}
\}.
\]
By Markov's inequality,
$
\prob_{(l_x x)^{(A)}, (l_y y)^{(B)}}[\mathrm{Bad}] < \epsilon^{1/16}.
$
For $((l_x x)^{(A)}, (l_y y)^{(B)}) \in \mathrm{Bad}$, we
define 
$\hsigma^{'\hM}_{(l_x x)^{(A)}, (l_y y)^{(B)}} := \zero$.
For $((l_x x)^{(A)}, (l_y y)^{(B)}) \not \in \mathrm{Bad}$, we
apply Lemma~\ref{lem:asyml2} to the states 
$\hsigma^{\hM}_{(l_x x)^{(A)}, (l_y y)^{(B)}}$,
$\hhsigma^{\hM}_{(l_x x)^{(A)}, (l_y y)^{(B)}}$ to get a state
$\hsigma^{'\hM}_{(l_x x)^{(A)}, (l_y y)^{(B)}}$ satisfying the 
following properties:
\begin{equation}
\label{eq:softcoveringasym}
\begin{array}{c}
\forall ((l_x x)^{(A)}, (l_y y)^{(B)}):
\hsigma^{'\hM}_{(l_x x)^{(A)}, (l_y y)^{(B)}} \leq
\hsigma^{\hM}_{(l_x x)^{(A)}, (l_y y)^{(B)}}, \\
\forall ((l_x x)^{(A)}, (l_y y)^{(B)}) \not \in \mathrm{Bad}:
\|\hsigma^{'\hM}_{(l_x x)^{(A)}, (l_y y)^{(B)}} -
  \hsigma^{\hM}_{(l_x x)^{(A)}, (l_y y)^{(B)}}\|_1
< 5 \epsilon^{1/32}, \\
\implies
\E_{(l_x x)^{(A)}, (l_y y)^{(B)}}\left[
\|\hsigma^{'\hM}_{(l_x x)^{(A)}, (l_y y)^{(B)}} -
  \hsigma^{\hM}_{(l_x x)^{(A)}, (l_y y)^{(B)}}\|_1
\right] < 5 \epsilon^{1/32} + \epsilon^{1/16} < 6 \epsilon^{1/32}, \\
\forall ((l_x x)^{(A)}, (l_y y)^{(B)}):
\|\hsigma^{'\hM}_{(l_x x)^{(A)}, (l_y y)^{(B)}}\|_2^2 \leq
(1 + 10\epsilon^{1/32})
\Tr[
\hsigma^{\hM}_{(l_x x)^{(A)}, (l_y y)^{(B)}}
\hhsigma^{\hM}_{(l_x x)^{(A)}, (l_y y)^{(B)}}
].
\end{array}
\end{equation}
Hence by triangle inequality and Equations~\ref{eq:softcovering12}, 
\ref{eq:softcoveringasym}, in order to prove 
Proposition~\ref{prop:coveringnonpairwise}, it suffices to show that
\begin{equation}
\label{eq:softcovering13}
\E_{(l_x x)^{(A)}, (l_y y)^{(B)}}[
\|\hsigma^{'\hM}_{(l_x x)^{(A)}, (l_y y)^{(B)}} - \rho^{\hM}\|_1
]
\leq 6 \epsilon^{1/64} + \sqrt{40\epsilon^{1/32} + 8e + g(e)}.
\end{equation}

Define $\one^{F'_\rho}$ to be the orthogonal projector
from the ambient space $\hM$ onto $F'_\rho$.
Using Equations~\ref{eq:softcoveringsupport}, \ref{eq:softcoveringtilt},
\ref{eq:softcoveringasym} (some of these equations were stated 
earlier for $\Tr\rho = 1$ but they continue to hold in the natural
fashion for general $\Tr\rho$ too), we get
\begin{equation}
\label{eq:softcovering14}
\begin{array}{rcl}
\lefteqn{~~~\Tr[\one^{F'_\rho} \rho^{\hM}]} \\
& = &
\Tr[\rho^{\hM}],\\
\lefteqn{
\E_{(l_x x)^{(A)}, (l_y y)^{(B)}}[
\Tr[\one^{F'_\rho} \, \hsigma^{'\hM}_{(l_x x)^{(A)}, (l_y y)^{(B)}}]
]
} \\
& \geq &
\E_{(l_x x)^{(A)}, (l_y y)^{(B)}}\left[
\Tr[\one^{F'_\rho} \, \sigma^{\hM}_{(l_x x)^{(A)}, (l_y y)^{(B)}}] -
\|\sigma^{\hM}_{(l_x x)^{(A)}, (l_y y)^{(B)}} -
  \hsigma^{\hM}_{(l_x x)^{(A)}, (l_y y)^{(B)}}\|_1 
\right. \\
&  &
~~~~~~~~~~~~~~~~~~~~~~~~
\left.
{} -
\|\hsigma^{\hM}_{(l_x x)^{(A)}, (l_y y)^{(B)}} -
  \hsigma^{'\hM}_{(l_x x)^{(A)}, (l_y y)^{(B)}}\|_1
\right] \\
&   =  &
\E_{(l_x x)^{(A)}, (l_y y)^{(B)}}[
\Tr[\sigma^{\hM}_{(l_x x)^{(A)}, (l_y y)^{(B)}}]
] \\
&  &
~~~~~~
{} -
\E_{(l_x x)^{(A)}, (l_y y)^{(B)}}[
\|\sigma^{(L_X X)^A (L_Y Y)^B \hM} - 
  \hsigma^{(L_X X)^A (L_Y Y)^B \hM}\|_1 
] \\
&  &
~~~~~~
{} -
\E_{(l_x x)^{(A)}, (l_y y)^{(B)}}[
\|\hsigma^{(L_X X)^A (L_Y Y)^B \hM} - 
  \hsigma^{'(L_X X)^A (L_Y Y)^B \hM}\|_1
] \\
&   =  &
\Tr \rho^{\hM} -
\E_{(l_x x)^{(A)}, (l_y y)^{(B)}}[
\|\sigma^{(L_X X)^A (L_Y Y)^B \hM} - 
  \hsigma^{(L_X X)^A (L_Y Y)^B \hM}\|_1 
] \\
&  &
~~~~~~
{} -
\E_{(l_x x)^{(A)}, (l_y y)^{(B)}}[
\|\hsigma^{(L_X X)^A (L_Y Y)^B \hM} - 
  \hsigma^{'(L_X X)^A (L_Y Y)^B \hM}\|_1
] \\
& \geq &
(\Tr\rho^{\hM})(1  - 2\sqrt{2} \cdot \epsilon^{1/8} - 6 \epsilon^{1/32})
\;\geq\;
(\Tr\rho^{\hM})(1 - 9 \epsilon^{1/32}).
\end{array}
\end{equation}

Using Proposition~\ref{prop:shavedCauchySchwarz} and concavity of
square root, we see that
\begin{eqnarray*}
\lefteqn{
\E_{(l_x x)^{(A)}, (l_y y)^{(B)}}[
\|\hsigma^{'\hM}_{(l_x x)^{(A)}, (l_y y)^{(B)}} - \rho^{\hM}\|_1
]
} \\
& \leq &
\E_{(l_x x)^{(A)}, (l_y y)^{(B)}}\left[
2\sqrt{\Tr[\rho^{\hM}]} 
 \sqrt{\Tr[\rho^{\hM}] - \Tr[\one^{F'_\rho} \rho^{\hM}]} 
\right. \\
& &
~~~~~~~~~~~~~~~~~~~~~~~~~
{} +
2\sqrt{\Tr[\rho^{\hM}]} 
\sqrt{\Tr[\rho^{\hM}] - 
\Tr[\one^{F'_\rho} \, \hsigma^{'\hM}_{(l_x x)^{(A)}, (l_y y)^{(B)}}]
} \\
&  &
~~~~~~~~~~~~~~~~~~~~~~~~
\left.
{} +
\sqrt{|F'_\rho|} \cdot 
\|\hsigma^{'\hM}_{(l_x x)^{(A)}, (l_y y)^{(B)}} - \rho^{\hM}\|_2
\right] \\
& \leq &
2 \sqrt{\Tr[\rho^{\hM}] - 
\E_{(l_x x)^{(A)}, (l_y y)^{(B)}}[
\Tr[\one^{F'_\rho} \, \hsigma^{'\hM}_{(l_x x)^{(A)}, (l_y y)^{(B)}}]
]
} \\
&  &
~~~~~~~~~~~~
{} +
\sqrt{|F'_\rho|} \cdot 
\E_{(l_x x)^{(A)}, (l_y y)^{(B)}}[
\|\hsigma^{'\hM}_{(l_x x)^{(A)}, (l_y y)^{(B)}} - \rho^{\hM}\|_2 
] \\
& \leq &
6 \epsilon^{1/64} +
\sqrt{|F'_\rho|} \cdot 
\E_{(l_x x)^{(A)}, (l_y y)^{(B)}}[
\|\hsigma^{'\hM}_{(l_x x)^{(A)}, (l_y y)^{(B)}} - \rho^{\hM}\|_2 
] \\
\end{eqnarray*}
From Equation~\ref{eq:softcovering13} and convexity of the square
function, in order to prove
Proposition~\ref{prop:coveringnonpairwise}, it suffices to show that
\begin{equation}
\label{eq:softcovering15}
\E_{(l_x x)^{(A)}, (l_y y)^{(B)}}[
\|\hsigma^{'\hM}_{(l_x x)^{(A)}, (l_y y)^{(B)}} - \rho^{\hM}\|_2^2
] <
\frac{40 \epsilon^{1/32} + 8e + g(e)}{|F'_\rho|}.
\end{equation}

The left hand side of the above inequality is
\begin{equation}
\label{eq:softcovering16}
\begin{array}{rcl}
\lefteqn{
\E_{(l_x x)^{(A)}, (l_y y)^{(B)}}[
\|\hsigma^{'\hM}_{(l_x x)^{(A)}, (l_y y)^{(B)}} - \rho^{\hM}\|_2^2
] 
} \\
& = &
\E_{(l_x x)^{(A)}, (l_y y)^{(B)}}[
\|\hsigma^{'\hM}_{(l_x x)^{(A)}, (l_y y)^{(B)}}\|_2^2
] + \|\rho^{\hM}\|_2 \\
&   &
~~~
{} -
2 \E_{(l_x x)^{(A)}, (l_y y)^{(B)}}[
\Tr[\hsigma^{'\hM}_{(l_x x)^{(A)}, (l_y y)^{(B)}} \, \rho^{\hM}]
] \\
& \leq &
(1 + 10\epsilon^{1/32})
\E_{(l_x x)^{(A)}, (l_y y)^{(B)}}\left[
\Tr[\hsigma^{\hM}_{(l_x x)^{(A)}, (l_y y)^{(B)}}
    \hhsigma^{\hM}_{(l_x x)^{(A)}, (l_y y)^{(B)}}
]
\right] \\
& &
~~~~
{} + 
(1 + \sqrt{\epsilon})^2
\left\|
\frac{\one^{F'_\rho}}{|F'_\rho|}
\right\|_2^2 \\
& &
~~~~
{} - 
2 (1 - \sqrt{\epsilon})
\E_{(l_x x)^{(A)}, (l_y y)^{(B)}}\left[
\Tr\left[
\hsigma^{'\hM}_{(l_x x)^{(A)}, (l_y y)^{(B)}}
\left(
\frac{\one^{F'_\rho}}{|F'_\rho|}
\right)
\right] 
\right] \\
&   =  &
(1 + 10\epsilon^{1/32})
\E_{(l_x x)^{(A)}, (l_y y)^{(B)}}\left[
\Tr[\hsigma^{\hM}_{(l_x x)^{(A)}, (l_y y)^{(B)}}
    \hhsigma^{\hM}_{(l_x x)^{(A)}, (l_y y)^{(B)}}
]
\right] +
\frac{(1 + \sqrt{\epsilon})^2}
     {|F'_\rho|}  \\
& &
~~~~
{} - 
\frac{2 (1 - \sqrt{\epsilon})}
     {|F'_\rho|}
\E_{(l_x x)^{(A)}, (l_y y)^{(B)}}\left[
\Tr[\hsigma^{'\hM}_{(l_x x)^{(A)}, (l_y y)^{(B)}} \, \one^{F'_\rho}] 
\right] \\
& \leq &
(1 + 10\epsilon^{1/32})
\E_{(l_x x)^{(A)}, (l_y y)^{(B)}}\left[
\Tr[\hsigma^{\hM}_{(l_x x)^{(A)}, (l_y y)^{(B)}}
    \hhsigma^{\hM}_{(l_x x)^{(A)}, (l_y y)^{(B)}}
]
\right] +
\frac{(1 + \sqrt{\epsilon})^2}
     {|F'_\rho|}  \\
& &
~~~~
{} - 
\frac{2 (1 - \sqrt{\epsilon})(1 - 4e - 9 \epsilon^{1/32})}
     {|F'_\rho|} \\
& \leq &
(1 + 10\epsilon^{1/32})
\E_{(l_x x)^{(A)}, (l_y y)^{(B)}}\left[
\Tr[\hsigma^{\hM}_{(l_x x)^{(A)}, (l_y y)^{(B)}}
    \hhsigma^{\hM}_{(l_x x)^{(A)}, (l_y y)^{(B)}}
]
\right] \\
&  &
~~~~~~~~~
{} +
\frac{8e + 23 \epsilon^{1/32} - 1}
     {|F'_\rho|},
\end{array}
\end{equation}
where the first inequality follows from
Equations~\ref{eq:softcoveringasym}, \ref{eq:softcovering7}, and the
second inequality follows from Equations~\ref{eq:softcovering14},
\ref{eq:softcoveringsupport}.

From Equations~\ref{eq:softcovering15} and \ref{eq:softcovering16},
in order to prove
Proposition~\ref{prop:coveringnonpairwise}, it suffices to show the
following lemma.
\begin{lemma}
\label{lem:asymnonpairwise}
Suppose 
\begin{eqnarray*}
\sqrt{L} 
& > &
\frac{8 |F'_\rho|}{\epsilon^{3/8}} \cdot
\max\left\{
\sum_{x,y} p^{XY}(x,y) \frac{p^X(x)}{q^X(x)} \|\rho(X)^M_{xy}\|_1, 
\right. \\
&  &
~~~~~~~~~~~~~~~~~~~
\left.
\sum_{x,y} p^{XY}(x,y) \frac{p^Y(y)}{q^Y(y)} \|\rho(Y)^M_{xy}\|_1,
\sum_{x,y} p^{XY}(x,y) \|\rho()^M_{xy}\|_1
\right\},
\end{eqnarray*}
and
\begin{eqnarray*}
\log A 
& > &
D^\epsilon_\infty(\rho^{XM} \| q^X \otimes \rho^M) +
\log f(e) + 3 + \log \epsilon^{-1/2}, \\
\log B 
& > &
D^\epsilon_\infty(\rho^{YM} \| q^Y \otimes \rho^M) +
\log f(e) + 3 + \log \epsilon^{-1/2}, \\
\log A + \log B 
& > &
D^\epsilon_\infty(\rho^{XYM} \| q^X \otimes q^Y \otimes \rho^M) +
1 + \log \epsilon^{-1/2}.
\end{eqnarray*}
Then,
\[
\E_{(l_x x)^{(A)}, (l_y y)^{(B)}}\left[
\Tr[\hsigma^{\hM}_{(l_x x)^{(A)}, (l_y y)^{(B)}}
    \hhsigma^{\hM}_{(l_x x)^{(A)}, (l_y y)^{(B)}}
]
\right] <
\frac{1 + 6 \epsilon^{1/2} + g(e)}
     {|F'_\rho|}.
\]
\end{lemma}
\begin{proof}
Given the lower bound on $\sqrt{L}$, we can write down the following
consequences of Lemmas~\ref{lem:nonhrhohhrho},
\ref{lem:nonhrhohhrho1}, \ref{lem:nonhrhohhrho2}, \ref{lem:nonhrhohhrho3}.
\begin{equation}
\label{eq:softcovering17}
\begin{array}{rcl}
\lefteqn{
L^{-2} \sum_{l_x, x, l_y, y} q^X(x) q^Y(y) 
\Tr\left[
\left(
\frac{p^{XY}(x, y)}{q^{X}(x) q^{Y}(y)}
\hrho^{\hM}_{l_x x l_y y}
\right)
\left(
\frac{p^{XY}(x, y)}{q^{X}(x) q^{Y}(y)}
\hhrho^{\hM}_{l_x x l_y y}
\right) 
\right]
} \\
& \leq &
\frac{(1 + 4\sqrt{e} + 2\sqrt{\epsilon})
      2^{D^\epsilon_\infty(\rho^{XYM}\| q^X\otimes q^Y\otimes \rho^M)}
     }{|F'_{\rho}|},
~~~~~~~~~~~~~~~~~~~~~~~~~~~~~~~\\
\lefteqn{
L^{-3} \sum_{l_x, x, l_y, y, l'_y, y'} q^X(x) \bar{q}^{YY'}(y, y')
\Tr\left[
\left(
\frac{p^{XY}(x, y')}{q^{X}(x) q^{Y}(y')}
\hrho^{\hM}_{l_x x l'_y y'}
\right)
\left(
\frac{p^{XY}(x, y)}{q^{X}(x) q^{Y}(y)}
\hhrho^{\hM}_{l_x x l_y y}
\right) 
\right]
} \\
& \leq &
\frac{4 (1+4\sqrt{e}+3\sqrt{\epsilon}) f(e) 
      2^{D^\epsilon_\infty(\rho^{XM}\|q^X \otimes \rho^M)}
     }{|F'_{\rho}|}, 
~~~~~~~~~~~~~~~~~~~~~~~~~~~~~~~~~~~~~~~~~~~~~~~~~~~~~~~~~\\
\lefteqn{
L^{-3} \sum_{l_x, x, l_y, y, l'_x, x'} q^Y(y) \bar{q}^{XX'}(x, x')
\Tr\left[
\left(
\frac{p^{XY}(x', y)}{q^{X}(x') q^{Y}(y)}
\hrho^{\hM}_{l'_x x' l_y y}
\right)
\left(
\frac{p^{XY}(x, y)}{q^{X}(x) q^{Y}(y)}
\hhrho^{\hM}_{l_x x l_y y}
\right) 
\right]
} \\
& \leq &
\frac{4 (1+4\sqrt{e}+3\sqrt{\epsilon}) f(e) 
      2^{D^\epsilon_\infty(\rho^{YM}\|q^Y \otimes \rho^M)}
     }{|F'_{\rho}|}, 
~~~~~~~~~~~~~~~~~~~~~~~~~~~~~~~~~~~~~~~~~~~~~~~~~~~~~~~~~\\
\lefteqn{
L^{-4} \sum_{l_x, x, l_y, y, l'_x, x', l'_y, y'} 
\bar{q}^{X X'}(x, x') \bar{q}^{YY'}(y, y')
\Tr\left[
\left(
\frac{p^{XY}(x', y')}{q^{X}(x') q^{Y}(y')}
\hrho^{\hM}_{l'_x x' l'_y y'}
\right)
\left(
\frac{p^{XY}(x, y)}{q^{X}(x) q^{Y}(y)}
\hhrho^{\hM}_{l_x x l_y y}
\right) 
\right]
} \\
& \leq &
\frac{(1+g(e) + 3 \sqrt{\epsilon}) 
     }{|F'_{\rho}|}.
~~~~~~~~~~~~~~~
\end{array}
\end{equation}

We have,
\begin{equation}
\label{eq:softcovering18}
\begin{array}{rcl}
\lefteqn{
\E_{(l_x x)^{(A)}, (l_y y)^{(B)}}\left[
\Tr[\hsigma^{\hM}_{(l_x x)^{(A)}, (l_y y)^{(B)}}
    \hhsigma^{\hM}_{(l_x x)^{(A)}, (l_y y)^{(B)}}
]
\right]
} \\
& = &
(AB)^{-2} \cdot 
\sum_{a,\hat{a}=1}^A \sum_{b,\hat{b}=1}^B \\
&  &
~~~~~~~~~~~~~~~~~
\E_{(l_x x)^{(A)}, (l_y y)^{(B)}}\left[
\frac{p^{XY}(x(\hat{a}), y(\hat{b}))}{q^X(x(\hat{a})) q^Y(y(\hat{b}))}
\cdot
\frac{p^{XY}(x(a), y(b))}{q^X(x(a)) q^Y(y(b))}
\right. \\
&  &
~~~~~~~~~~~~~~~~~~~~~~~~~~~~~~~~~~~~~~~~~~~~~~
\left.
\Tr[\hrho^{\hM}_{(l_x x)(\hat{a}), (l_y y)(\hat{b})}
    \hhrho^{\hM}_{(l_x x)(a), (l_y y)(b)}
   ]
\right] \\
& = &
(AB)^{-2} \cdot 
\sum_{a,\hat{a}=1}^A \sum_{b,\hat{b}=1}^B \\
&  &
~~~~~~~
\displaystyle{
\E_{((l_x x)(\hat{a}), (l_x x)(a)), ((l_y y)(\hat{b}), (l_y y)(b))}
}
\left[ 
\frac{p^{XY}(x(\hat{a}), y(\hat{b}))}{q^X(x(\hat{a})) q^Y(y(\hat{b}))}
\cdot
\frac{p^{XY}(x(a), y(b))}{q^X(x(a)) q^Y(y(b))}
\right. \\
&  &
~~~~~~~~~~~~~~~~~~~~~~~~~~~~~~~~~~~~~~~~~~~~~~~~~~~
\left.
\Tr[\hrho^{\hM}_{(l_x x)(\hat{a}), (l_y y)(\hat{b})}
    \hhrho^{\hM}_{(l_x x)(a), (l_y y)(b)}
   ]
\right].
\end{array}
\end{equation}

We analyse the above summation by considering several cases.
Consider the following term for a fixed choice of 
$a \neq \hat{a}$, $b \neq \hat{b}$.
\begin{equation}
\label{eq:softcovering19}
\begin{array}{rcl}
\lefteqn{
\E_{((l_x x)(\hat{a}), (l_x x)(a)), ((l_y y)(\hat{b}), (l_y y)(b))}
\left[ 
\frac{p^{XY}(x(\hat{a}), y(\hat{b}))}{q^X(x(\hat{a})) q^Y(y(\hat{b}))}
\cdot
\frac{p^{XY}(x(a), y(b))}{q^X(x(a)) q^Y(y(b))}
\right. 
} \\
&  &
~~~~~~~~~~~~~~~~~~~~~~~~~~~~~~~~~~~~~~
\left.
\Tr[\hrho^{\hM}_{(l_x x)(\hat{a}), (l_y y)(\hat{b})}
    \hhrho^{\hM}_{(l_x x)(a), (l_y y)(b)}
   ]
\right] \\
& = &
L^{-4} \sum_{l_x, x, l_y, y, l'_x, x', l'_y, y'} 
\bar{q}^{X X'}(x, x') \bar{q}^{YY'}(y, y') \\
&  &
~~~~~~~~~~~~~~~~~~~~~~~~~~~~~~~~~
\Tr\left[
\left(
\frac{p^{XY}(x', y')}{q^{X}(x') q^{Y}(y')}
\hrho^{\hM}_{l'_x x' l'_y y'}
\right)
\left(
\frac{p^{XY}(x, y)}{q^{X}(x) q^{Y}(y)}
\hhrho^{\hM}_{l_x x l_y y}
\right) 
\right] \\
& \leq &
\frac{(1+g(e) + 3 \sqrt{\epsilon}) 
     }{|F'_{\rho}|},
\end{array}
\end{equation}
where we used Equation~\ref{eq:softcovering17} in the inequality above.
There are $A(A-1) B(B-1)$ such terms.

Consider the following term for a fixed choice of 
$a = \hat{a}$, $b \neq \hat{b}$.
\begin{equation}
\label{eq:softcovering20}
\begin{array}{rcl}
\lefteqn{
\E_{(l_x x)(a), ((l_y y)(\hat{b}), (l_y y)(b))}
\left[ 
\frac{p^{XY}(x(a), y(\hat{b}))}{q^X(x(a)) q^Y(y(\hat{b}))}
\cdot
\frac{p^{XY}(x(a), y(b))}{q^X(x(a)) q^Y(y(b))}
\right. 
} \\
&  &
~~~~~~~~~~~~~~~~~~~~~~~~~~~
\left.
\Tr[\hrho^{\hM}_{(l_x x)(a), (l_y y)(\hat{b})}
    \hhrho^{\hM}_{(l_x x)(a), (l_y y)(b)}
   ]
\right] \\
& = &
L^{-3} \sum_{l_x, x, l_y, y, l'_y, y'} q^X(x) \bar{q}^{YY'}(y, y')
\Tr\left[
\left(
\frac{p^{XY}(x, y')}{q^{X}(x) q^{Y}(y')}
\hrho^{\hM}_{l_x x l'_y y'}
\right)
\left(
\frac{p^{XY}(x, y)}{q^{X}(x) q^{Y}(y)}
\hhrho^{\hM}_{l_x x l_y y}
\right) 
\right] \\
& \leq &
\frac{4 (1+4\sqrt{e}+3\sqrt{\epsilon}) f(e) 
      2^{D^\epsilon_\infty(\rho^{XM}\|q^X \otimes \rho^M)}
     }{|F'_{\rho}|}, 
\end{array}
\end{equation}
where we used Equation~\ref{eq:softcovering17} in the inequality above.
There are $A B(B-1)$ such terms.

Consider the following term for a fixed choice of 
$a \neq \hat{a}$, $b = \hat{b}$. We have similarly,
\begin{equation}
\label{eq:softcovering21}
\begin{array}{rcl}
\lefteqn{
\E_{((l_x x)(\hat{a}), (l_x x)(a)), (l_y y)(b)}
\left[ 
\frac{p^{XY}(x(\hat{a}), y(b))}{q^X(x(\hat{a})) q^Y(y(b))}
\cdot
\frac{p^{XY}(x(a), y(b))}{q^X(x(a)) q^Y(y(b))}
\right. 
} \\
&  &
~~~~~~~~~~~~~~~~~~~~~~~~~~~
\left.
\Tr[\hrho^{\hM}_{(l_x x)(\hat{a}), (l_y y)(b)}
    \hhrho^{\hM}_{(l_x x)(a), (l_y y)(b)}
   ]
\right] \\
& \leq &
\frac{4 (1+4\sqrt{e}+3\sqrt{\epsilon}) f(e) 
      2^{D^\epsilon_\infty(\rho^{YM}\|q^Y \otimes \rho^M)}
     }{|F'_{\rho}|}. 
~~~~~~~~~~~~~~~~~~~~~~~~~~~~~~~~~~~~~~~~~~
\end{array}
\end{equation}
There are $A(A-1) B$ such terms.

Finally, consider the following term for a fixed choice of 
$a = \hat{a}$, $b = \hat{b}$.
\begin{equation}
\label{eq:softcovering22}
\begin{array}{rcl}
\lefteqn{
\E_{(l_x x)(a), (l_y y)(b)}
\left[ 
\frac{p^{XY}(x(a), y(b))}{q^X(x(a)) q^Y(y(b))}
\cdot
\frac{p^{XY}(x(a), y(b))}{q^X(x(a)) q^Y(y(b))}
\right. 
} \\
&  &
~~~~~~~~~~~~~~~~~~~~
\left.
\Tr[\hrho^{\hM}_{(l_x x)(a), (l_y y)(b)}
    \hhrho^{\hM}_{(l_x x)(a), (l_y y)(b)}
   ]
\right] \\
& = &
L^{-2} \sum_{l_x, x, l_y, y} q^X(x) q^Y(y) 
\Tr\left[
\left(
\frac{p^{XY}(x, y)}{q^{X}(x) q^{Y}(y)}
\hrho^{\hM}_{l_x x l_y y}
\right)
\left(
\frac{p^{XY}(x, y)}{q^{X}(x) q^{Y}(y)}
\hhrho^{\hM}_{l_x x l_y y}
\right) 
\right] \\
& \leq &
\frac{(1 + 4\sqrt{e} + 2\sqrt{\epsilon})
      2^{D^\epsilon_\infty(\rho^{XYM}\| q^X\otimes q^Y\otimes \rho^M)}
     }{|F'_{\rho}|},
\end{array}
\end{equation}
where we used Equation~\ref{eq:softcovering17} in the inequality above.
There are $A B$ such terms.

From Equations~\ref{eq:softcovering18}, \ref{eq:softcovering19}, 
\ref{eq:softcovering20},
\ref{eq:softcovering21}, \ref{eq:softcovering22} we get
\begin{equation}
\label{eq:softcovering23}
\begin{array}{rcl}
\lefteqn{
\E_{(l_x x)^{(A)}, (l_y y)^{(B)}}\left[
\Tr[\hsigma^{\hM}_{(l_x x)^{(A)}, (l_y y)^{(B)}}
    \hhsigma^{\hM}_{(l_x x)^{(A)}, (l_y y)^{(B)}}
]
\right]
} \\
& = &
(AB)^{-2} \cdot 
\sum_{a,\hat{a}=1}^A \sum_{b,\hat{b}=1}^B \\
&  &
~~~~~~~
\displaystyle{
\E_{((l_x x)(\hat{a}), (l_x x)(a)), ((l_y y)(\hat{b}), (l_y y)(b))}
}
\left[ 
\frac{p^{XY}(x(\hat{a}), y(\hat{b}))}{q^X(x(\hat{a})) q^Y(y(\hat{b}))}
\cdot
\frac{p^{XY}(x(a), y(b))}{q^X(x(a)) q^Y(y(b))}
\right. \\
&  &
~~~~~~~~~~~~~~~~~~~~~~~~~~~~~~~~~~~~~~~~~~~~~~~~~~~
\left.
\Tr[\hrho^{\hM}_{(l_x x)(\hat{a}), (l_y y)(\hat{b})}
    \hhrho^{\hM}_{(l_x x)(a), (l_y y)(b)}
   ]
\right] \\
& \leq &
(AB)^{-2} 
\left(
\frac{(A(A-1)B(B-1)(1+g(e) + 3 \sqrt{\epsilon})}{|F'_{\rho}|} +
\frac{4 AB(B-1) (1+4\sqrt{e}+3\sqrt{\epsilon}) f(e) 
      2^{D^\epsilon_\infty(\rho^{XM}\|q^X \otimes \rho^M)}
     }{|F'_{\rho}|} 
\right. \\
&  &
~~~~~~~~~~~~~~~~~~
{} +
\frac{4 A(A-1)B (1+4\sqrt{e}+3\sqrt{\epsilon}) f(e) 
      2^{D^\epsilon_\infty(\rho^{YM}\|q^Y \otimes \rho^M)}
     }{|F'_{\rho}|} \\
&  &
~~~~~~~~~~~~~~~~~~
\left.
{} + 
\frac{AB (1 + 4\sqrt{e} + 2\sqrt{\epsilon})
      2^{D^\epsilon_\infty(\rho^{XYM}\| q^X\otimes q^Y\otimes \rho^M)}
     }{|F'_{\rho}|}
\right) \\
& \leq &
\frac{(1+g(e) + 3 \sqrt{\epsilon})}{|F'_{\rho}|} +
\frac{4 (1+4\sqrt{e}+3\sqrt{\epsilon}) f(e) 
      2^{D^\epsilon_\infty(\rho^{XM}\|q^X \otimes \rho^M)}
     }{A |F'_{\rho}|}  \\
&  &
{} +
\frac{4 (1+4\sqrt{e}+3\sqrt{\epsilon}) f(e) 
      2^{D^\epsilon_\infty(\rho^{YM}\|q^Y \otimes \rho^M)}
     }{B |F'_{\rho}|} + 
\frac{(1 + 4\sqrt{e} + 2\sqrt{\epsilon})
      2^{D^\epsilon_\infty(\rho^{XYM}\| q^X\otimes q^Y\otimes \rho^M)}
     }{AB |F'_{\rho}|}.
\end{array}
\end{equation}

From the constraints on $\log A$, $\log B$, we get
\[
\E_{(l_x x)^{(A)}, (l_y y)^{(B)}}\left[
\Tr[\hsigma^{\hM}_{(l_x x)^{(A)}, (l_y y)^{(B)}}
    \hhsigma^{\hM}_{(l_x x)^{(A)}, (l_y y)^{(B)}}
]
\right]
 \leq 
\frac{(1+g(e) + 6 \sqrt{\epsilon})}{|F'_{\rho}|}.
\]
This completes the proof of the lemma.
\end{proof}

Note that the entropic quantities in the statement of
Lemma~\ref{lem:asymnonpairwise} are
in terms of the original normalised probability distribution 
$p^{XY}$, and not in terms of $\hat{p}^{XY}$ which was later renamed
to $p^{XY}$ for notational convenience.
Since the dimension $L$ of the augmenting Hilbert spaces $L_X$, $L_Y$
is a free parameter in our proof, it can be chosen large enough as 
in the requirement of Lemma~\ref{lem:asymnonpairwise}. Thus
Lemma~\ref{lem:asymnonpairwise} holds, which completes the proof of
Proposition~\ref{prop:coveringnonpairwise}.

Arguing along similar lines, we can now prove a smooth multipartite
soft covering lemma without pairwise independence. This is the main 
theorem of this paper. 
\begin{theorem}
\label{thm:main}
Let $k$ be a positive integer. Let $X_1, \ldots, X_k$ be $k$ classical
alphabets. For any subset
$S \subseteq [k]$, let $X_S := (X_s)_{s \in S}$.
Let $p^{X_{[k]}}$ be a normalised probability distribution on $X_{[k]}$.
The notation $p^{X_S}$ denotes the marginal distribution on $X_S$. 
Let $q^{X_1}, \ldots, q^{X_k}$ be normalised
probability distributions on the respective alphabets.
For each $(x_1, \ldots, x_k) \in X_{[k]}$, 
let $\rho^M_{x_1, \ldots, x_k}$
be a subnormalised density matrix on $M$. 
The classical quantum {\em control state} is now defined as
\[
\rho^{X_{[k]} M} :=
\sum_{(x_1, \ldots, x_k) \in X_{[k]}}
p^{X_{[k]}}(x_1, \ldots, x_k) \ketbra{x_1, \ldots, x_k}^{X_{[k]}}
\otimes \rho^M_{x_1, \ldots, x_k}.
\]
Suppose $\supp(p^{X_i}) \leq \supp(q^{X_i})$.
For any subset $S \subseteq [k]$, let 
$q^{X_S} := \times_{s \in S} q^{X_s}$.
Let $A_1, \ldots, A_k$ be positive integers.
For each $i \in [k]$, let
$x_i^{(A_i)} := (x_i(1), \ldots, x_i(A_i))$ denote a $|A_i|$-tuple
of elements from  $X_i$.
Denote the $A_i$-fold product alphabet 
$X_i^{A_i} := X_i^{\times A_i}$, and the product probability distribution
$
q^{X_i^{A_i}} :=
(q^{X_i})^{\times A_i}.
$
For any collection of tuples
$x_i^{(A_i)} \in X_i^{A_i}$, $i \in [k]$, we
define the {\em sample average covering state}
\[
\sigma^M_{x_1^{(A_1}, \ldots, x_k^{(A_k)}} :=
(A_1 \cdots A_k)^{-1}
\sum_{a_1 = 1}^{A_1} \cdots \sum_{a_k = 1}^{A_k}
\frac{p^{X_{[k]}}(x_1(a_1), \ldots, x_k(a_k))}
     {q^{X_1}(x_1(a_1)) \cdots q^{X_k}(x_k(a_k))}
\rho^M_{x_1(a_1), \ldots, x_k(a_k)},
\]
where the fraction term above represents the `change of measure' from
the product probability distribution $q^{X_{[k]}}$ to the
joint probability distribution $p^{X_{[k]}}$.

Let the covering be done using independent choices of tuples
$x_1^{(A_1)}, \ldots, x_k^{(A_k)}$ from the distributions 
$\bar{q}^{X_1^{A_1}}, \ldots, \bar{q}^{X_k^{A_k}}$, which satisfy
the consequences of loss of pairwise independence, 
akin to the situation in Proposition~\ref{prop:coveringnonpairwise}, up 
to scale factors of
$(e, f(e), g(e))$. Suppose for each non-empty subset 
$\{\} \neq S \subseteq [k]$,
\[
\sum_{s \in S} \log A_s > 
D^\epsilon_\infty(\rho^{X_S M} \| q^{X_S} \otimes \rho^M) +
\log f(e) + 3 + \log \epsilon^{-1/2}.
\]
Then,
\[
\E_{x_1^{(A_1)}, \ldots, x_k^{(A_k)}}
\left[
\left\|
\sigma^M_{x_1^{(A_1)}, \ldots, x_k^{(A_k)}} - \rho^M
\right\|_1
\right] <
(2\sqrt{10 \cdot 2^{k} \epsilon^{1/64} +
	\sqrt{30 \cdot 2^{2k} \epsilon^{1/32} + 2^{k+1} e + g(e)}
       } + 2^k e
) (\Tr\rho).
\]
\end{theorem}

\paragraph{Remarks:} \ \\

\noindent
1.\ The scale factor of $f(e)$ in the consequences of loss of pairwise
independence in the statement of Theorem~\ref{thm:main} means that
for any sequence of coordinates $(a_1, \ldots, a_k)$, 
$a_i \in [A_i]$, for any sequence of symbols $(x_{a_1}, \ldots, x_{a_k})$,
for any non-empty subset $\{\} \neq S \subset [k]$, for any sequence
of coordinates $(a'_s)_{s \in S}$ with $a_s \neq a'_s$ for all $s \in S$,
\[
\E_{(x'_{a'_s})_{s \in S}|(x_{a_s})_{s \in S}} 
[
\rho^M_{(x_{a_t})_{t \in \bar{S}}, \, (x'_{a'_s})_{s \in S}}
]
=: \rho(\bar{S})^M_{x_{a_1}, \ldots, x_{a_k}} \leq
f(e) \rho^M_{(x_{a_t})_{t \in \bar{S}}},
\]
where the expectation is taken over independent choices from the
distributions $\bar{q}^{X'_s|X_s = x_{a_s}}$, $s \in S$, and
$\rho^M_{(x_{a_t})_{t \in \bar{S}}}$ is the quantum state of register
$M$ when the symbols $(x_{a_t})_{t \in \bar{S}}$ are fixed and
the symbols in $(a_s)_{s \in S}$ are averaged over according to independent
choices from the distribution $q^{X_S}$. 

\noindent
2.\ A factor of $k$ comes in the error bound on the right hand side
of Theorem~\ref{thm:main}
due to using Lemma~\ref{lem:tiltclose} for $k$ orthogonal tilting
directions that arise from the $k$ covering parties. However it is
subsumed by factors of $2^k$ that are arise from the $k$ party 
analogue of Lemma~\ref{lem:matrixtiltedspan}. The factor of
$2^{k+1} e$ in the error bound arises from the $k$ party analogue
of Equation~\ref{eq:softcoveringminus1}, where we now have to do
$2^k - 1$ intersections of classical probability distributions.

Similarly, we can also show the following smooth version of CMG
covering without pairwise independence.
\begin{lemma}
\label{lem:CMGcoveringnonpairwise}
Consider the CMG covering problem of Section~\ref{sec:CMGcovering}, 
where there is a loss of pairwise
independence in Alice's obfuscation strategy when she chooses
$L$ symbols $x_{l',1}, \ldots, x_{l', L} \in \cX$ conditioned on a sample
$x'_{l'} \in \cX'$, and a similar loss of pairwise independence in
Bob's obfuscation strategy. Suppose these are the only two places of
loss of pairwise independence, with the consequences of the loss holding
up to scale factors of $(e, f(e), 0)$. Suppose
\begin{eqnarray*}
\log L'
& > &
I_\infty^\epsilon(X':E)_\sigma + 3  + \log \epsilon^{-1/2}, \\
\log M'
& > &
I_\infty^\epsilon(Y':E)_\sigma + 3  + \log \epsilon^{-1/2}, \\
\log L' + \log L
& > &
I_\infty^\epsilon(X'X:E)_\sigma + 3  + \log \epsilon^{-1/2}, \\
\log M' + \log M
& > &
I_\infty^\epsilon(Y'Y:E)_\sigma + 3  + \log \epsilon^{-1/2}, \\
\log L' + \log M'
& > &
I_\infty^\epsilon(X'Y':E)_\sigma + 3 + 
\log f(e)  + \log \epsilon^{-1/2}, \\
\log L' + \log M' + \log M
& > &
I_\infty^\epsilon(X'Y'Y:E)_\sigma + 3 +
\log f(e) + \log \epsilon^{-1/2}, \\
\log L' + \log M' + \log L
& > &
I_\infty^\epsilon(X'XY':E)_\sigma + 3 +
\log f(e) + \log \epsilon^{-1/2}, \\
\log L' + \log M' + \log L + \log M
& > &
I_\infty^\epsilon(X'XY'Y:E)_\sigma + 3  + \log \epsilon^{-1/2}.
\end{eqnarray*}
Then, 
$
\CMG < 
2\sqrt{160 \cdot \epsilon^{1/64} +
	\sqrt{750 \epsilon^{1/32} + 8 e}
      } + 4 e
$
\end{lemma}

\section{Inner bound for the wiretap interference channel}
\label{sec:wiretapinterference}
In a wiretap quantum interference channel, there are two senders Alice 
and Bob and their corresponding intended receivers Charlie and Damru. 
There is a third receiver Eve, corresponding to an eavesdropper. 
Alice would like to
send a classical message $m_1 \in [2^{R_1}]$ to Charlie. Similarly,
Bob would like to send $m_2 \in [2^{R_2}]$ to Damru. 
Charlie and Damru should be able to decode correctly with large
probability at their 
respective channel outputs and recover $m_1$, $m_2$.
It is okay
if Damru gets some information about Alice's message $m_1$ in addition to
the message $m_2$ intended for him; similarly for Charlie. However,
Eve should have at most negligible information about any message 
tuple $(m_1, m_2)$.

The parties have at their disposal
a quantum channel $\cN: A B \rightarrow C D E$ with input 
Hilbert spaces $\cA$, $\cB$ and output Hilbert spaces $\cC$, $\cD$, $\cE$. 
Though $\cN$ has quantum inputs and quantum
outputs, we will assume without loss of generality that the inputs are
classical, because in the inner bound below, there is a standard 
optimisation
step over a choice of all ensembles of input states to Alice and Bob.
In other words, we fix classical to quantum encodings
$x \mapsto \rho^A_x$, $y \mapsto \rho^B_y$.
Thus, we will henceforth think of $\cN$ as a classical quantum (cq)
channel with its two classical input alphabets being denoted by
$\cX$, $\cY$, and the quantum output alphabets by $\cC$, $\cD$, $\cE$. 

We now describe the strategy of Chong, Motani, Garg, El Gamal 
\cite{CMGElGamal}, adapted to the one shot quantum case by
\cite{sen:simultaneous}, to achieve a rate inner bound for
sending $(m_1, m_2)$ over $\cN$ when there is no eavesdropper Eve. 
As shown below, we will suitably modify that strategy to ensure privacy
against Eve. For brevity, we will call the overall coding strategy
as the CMG wiretap strategy. For
this, we need to define additional classical alphabets $\cX'$, $\cY'$,
another `timesharing' classical alphabet $\cQ$,
a `control probobablity distribution' 
\begin{equation}
\label{eq:CMGcontroldist}
p(q,x',x,y',y) = p(q) p(x',x) p(y',y),
\end{equation}
and a `control classical quantum state'
\begin{equation}
\label{eq:CMGcontrolstate}
\sigma^{Q X' X Y' Y CDE} := 
\sum_{q,x',x,y',y} p(q) p(x',x)p(y',y) \ketbra{q,x',x,y',y}^{QX'XY'Y}
\otimes \sigma^{CDE}_{x'xy'y},
\end{equation}
where $\sigma^{CDE}_{x'xy'y} := \sigma^{CDE}_{xy}$ similar to the
definitions in Section~\ref{sec:CMGcovering}.

The CMG wiretap strategy uses {\em rate splitting} for both Alice and 
Bob. By this we mean that Alice `splits' her message $m_1 \in [2^{R_1}]$
into a pair $(m'_1, m''_1) \in [2^{R'_1}] \times [2^{R_1 - R'_1}]$.
Similarly Bob `splits' his message $m_2 \in [2^{R_2}]$
into a pair $(m'_1, m''_1) \in [2^{R'_2}] \times [2^{R_2 - R'_2}]$.
The intention is that $m'_1$, $m'_2$ should be decodable by both 
Charlie and Damru. The message $m''_1$, which should of course be
decodable by Charlie, need not satisfy any requirement with respect
to Damru; similarly for message $m''_2$. Of course, the state of
Eve for any quadruple $(m'_1, m''_1, m'_2, m''_2)$ should be almost
independent of the quadruple for the privacy requirement to be met.
The intuition, though not the actual decoding strategy, is that 
the $m'_1$, $m'_2$ parts of the encoded messages
are `easy' to decode at both Charlie and Damru. So intuitively
Charlie can decode $m'_2$ and use it as side information in order to
decode $m''_1$; similarly for Damru. It is known that this strategy
of rate splitting enhances the achievable inner bound in many important
examples of interference channels.
The quantities $R'_1$, $R'_2$ can be optimised in a later step in order 
to get the largest inner bound possible.

We now describe the codebook, Alice and Bob's encoding and obfuscation,
and Charlie and Damru's decoding strategies.

\medskip

\noindent
{\bf Codebook:}

\noindent
Let $0 \leq \epsilon \leq 1$. 
Let $L'$, $L$, $M'$, $M$ be positive integers.
As will become clear very soon, these four integers play the obfuscation
role as described for the CMG covering problem in 
Section~\ref{sec:CMGcovering}.
First generate a sample $q$ from the distribution $p^Q$. 
For each message
$m'_1 \in [2^{R'_1}]$ independently generate samples
$x'(m'_1,1), \ldots, x'(m'_1, L')$ from the distribution $p^{X'|Q=q}$.
Similarly, for each message
$m'_2 \in [2^{R'_2}]$ independently generate samples
$y'(m'_2, 1), \ldots y'(m'_2, M')$ from the distribution 
$p^{Y' | Q=q}$.
Now for each codeword $x'(m'_1, l')$, $l' \in [L']$  
independently generate samples
$x(m'_1, l', m''_1, l)$ from the distribution 
$p^{X | X'= x'(m'_1, l'), Q=q}$ for all
$(m''_1, l) \in [2^{R_1 - R'_1}] \times [L]$.
Similarly for each codeword $y'(m'_2, m')$, $m' \in [M']$ , 
independently generate samples
$y(m'_2, m', m''_2, m)$ from the probability distribution 
$p^{Y | Y'=y'(m'_2, m') Q=q}$ for all
$(m''_2, m) \in [2^{R_2 - R'_2}] \times [M]$. These samples together
consititute the random codebook. We augment the random codebook
as in \cite{sen:simultaneous} to get a random augmented codebook.
This random augmented codebook is revealed to Alice, Bob, Charlie,
Damru and Eve. 

\medskip

\noindent
{\bf Encoding and obfuscation:}

\noindent
To send message $m_1 = (m'_1, m''_1)$, Alice picks up a  symbol
$x(m'_1, l', m''_1, l)$ 
for a uniformly random pair $(l',l) \in [L'] \times [L]$
from the codebook and inputs the 
state $\rho_{x(m'_1, l', m''_1, l)}^{A}$ into the input $\cA$ of
the channel $\cN$.
Similarly, to send message $m_2 = (m'_2, m''_2)$, Bob picks up 
the symbol $y(m'_2, m', m''_2, m)$  
for a uniformly random pair $(m',m) \in [M'] \times [M]$
from the codebook and inputs the 
state $\rho_{y(m'_2, m', m''_2, m)}^{B}$ into the input $\cB$ of
$\cN$.

\medskip

\noindent
{\bf Decoding:}

\noindent
Charlie decodes the tuple $m_1 = (m'_1, m''_1)$ using
{\em simultaneous non-unique decoding} as explained in 
\cite{sen:simultaneous}, obtaining his guess $\hat{m}_1$ for $m_1$. 
Damru decodes the tuple $m_2 = (m'_2, m''_2)$ using
{\em simultaneous non-unique decoding}, 
obtaining his guess $\hat{m}_2$ for $m_2$. By simultaneous decoding,
we mean that Charlie's decoded pair $(\hat{m'}_1, \hat{m''}_1)$ is obtained
by applying a single POVM element to the state at the channel output
$\cC$; this POVM element is not composed from `simpler' POVM elements
like e.g. the successive cancellation decoder which first decodes,
say $\hat{m'}_1$, and then uses it as side information to decode
$\hat{m''}_1$.  By non-unique
decoding, we mean that Charlie's decoding POVM element for a fixed
tuple $(\hat{m'}_1, \hat{m''}_1)$ can be thought of as a `union',
ranging over $(l',l,m'_2) \in [L'] \times [L] \times [2^{R'_2}]$, of
projectors for each possible tuple 
$(\hat{m'}_1, l', \hat{m''}_1, l, m'_2)$. The `actual or unique 
identity' of
$(l',l,m'_2)$ in the `union' is immaterial as long as the correct
$(\hat{m'}_1, \hat{m''}_1) = (m'_1, m''_1)$ is decoded.

\medskip

\noindent
{\bf Correctness:}

\noindent
We want that 
$
\prob[(\hat{m}_1, \hat{m}_2) \neq 
    (m_1, m_2)] \leq \epsilon,
$
where the probability is over the uniform choice of the message pair
$(m_1,m_2) \in [2^{R_1}] \times [2^{R_2}]$, and actions of the
encoder, channel and decoders. 

\medskip

\noindent
{\bf Privacy:}

\noindent
We want that 
$
\E_{(m_1,m_2)}[\|\sigma^E_{m_1,m_2} - \sigma^E\|_1] \leq \epsilon,
$
where the expectation is over the uniform choice of the message pair
$(m_1,m_2) \in [2^{R_1}] \times [2^{R_2}]$, 
$\sigma^E_{m_1,m_2}$ denotes the quantum state of Eve when message
pair $(m_1,m_2)$ is transmitted, and $\sigma^E$ is the marginal of
the control state on $E$ from Equation~\ref{eq:CMGcontrolstate}. 
Note that $\sigma^E_{m_1,m_2}$ can be defined from the control state
$\sigma^{Q X'X Y'Y E}$ given the codebook.

If there exist such coding
schemes for a particular channel $\cN$, we say that there exists
an {\em $(R_1, R_2, \epsilon)$-wiretap quantum interference channel code}
for sending private classical information through $\cN$.

We are interested in the variant of the codebook construction where
pairwise independence is lost in the following step of constructing
Alice's part of the codebook.
\begin{quote}
Now for each codeword $x'(m'_1, l')$, $l' \in [L']$  
independently generate samples
$x(m'_1, l', m''_1, l)$ from the distribution 
$p^{X | X'= x'(m'_1, l'), Q=q}$ for all
$(m''_1, l) \in [2^{R_1 - R'_1}] \times [L]$.
\end{quote}
A similar loss of pairwise independence occurs in the corresponding
step of constructing Bob's part of the codebook. There is no further
loss of pairwise independence. The consequences of
the loss of pairwise independence hold up to scale factors of
$(e, f(e), 0)$. Loss of pairwise independence in these parts of 
codebook construction makes computational sense as these parts
require the largest amount of random bits because they contain
most of the codeword samples.

It turns out that the simultaneous non-unique decoding strategy 
for Charlie and Damru without any privacy considerations,
already described in \cite{sen:simultaneous},
continues to work with a natural slight degradation of parameters 
when pairwise
independence is lost up to scale factors of $(e, f(e), 0)$. This 
is because the strategy is based on the same machinery of
tilting and augmentation smoothing. The requirement of privacy against
Eve comes from the CMG covering problem. We have already shown a 
smooth inner bound for this problem under loss of pairwise independence
in Lemma~\ref{lem:CMGcoveringnonpairwise}. Combining both simultaneous
non-unique decoding for Charlie, Damru together with the CMG covering
obfuscation against Eve, we get the following inner bound for sending
private classical information over a quantum interference channel.
\begin{theorem}
\label{thm:wiretapinterference}
Consider the quantum interference channel with classical quantum
control state $\sigma^{Q X'X Y'Y CDE}$ defined in 
Equation~\ref{eq:CMGcontrolstate}.
Let $0 < \epsilon < 1$.
Consider a rate octuple
$(R'_1, L', R_1 - R'_1, L, R'_2, M', R_2 - R'_2, M)$ satisfying
the following inequalities:
\begin{eqnarray*}
R_1 - R'_1 + L 
& < &
I_H^{\epsilon}(X : C | X' Y' Q)_\sigma - \log f(e) - 
2 -  \log \epsilon^{-1} \\
R_1 + L' + L
& < &
I_H^{\epsilon}(X : C | Y' Q)_\sigma - 
2 -  \log \epsilon^{-1} \\
R_1 - R'_1 + L + R'_2 + M'
& < &
I_H^{\epsilon}(X Y' : C | X' Q)_\sigma - \log f(e) - 
2 -  \log \epsilon^{-1} \\
R_1 + L' + L + R'_2 + M'
& < &
I_H^{\epsilon}(X Y' : C | Q)_\sigma - 2 -  \log \epsilon^{-1} \\
R_2 - R'_2 + M 
& < &
I_H^{\epsilon}(Y : D | X' Y' Q)_\sigma - \log f(e) - 
2 -  \log \epsilon^{-1} \\
R_2 + M' + M
& < &
I_H^{\epsilon}(Y : D | X' Q)_\sigma - 
2 -  \log \epsilon^{-1} \\
R_2 - R'_2 + M + R'_1 + L'
& < &
I_H^{\epsilon}(Y X' : D | Y' Q)_\sigma - \log f(e) - 
2 -  \log \epsilon^{-1} \\
R_2 + M' + M + R'_1 + L'
& < &
I_H^{\epsilon}(Y X' : D | Q)_\sigma - 
2 -  \log \epsilon^{-1} \\
& & \\
L'
& > &
I_\infty^\epsilon(X':E)_\sigma + 3  + \log \epsilon^{-1}, \\
L' + L
& > &
I_\infty^\epsilon(X'X:E)_\sigma + 3  + \log \epsilon^{-1}, \\
M'
& > &
I_\infty^\epsilon(Y':E)_\sigma + 3  + \log \epsilon^{-1}, \\
M' + M
& > &
I_\infty^\epsilon(Y'Y:E)_\sigma + 3  + \log \epsilon^{-1}, \\
L' + M'
& > &
I_\infty^\epsilon(X'Y':E)_\sigma + 3 + \log f(e)  + \log \epsilon^{-1}, \\
L' + M' + M
& > &
I_\infty^\epsilon(X'Y'Y:E)_\sigma + 3 + \log f(e) + \log \epsilon^{-1}, \\
L' + M' + L
& > &
I_\infty^\epsilon(X'XY':E)_\sigma + 3 + \log f(e) + \log \epsilon^{-1}, \\
L' + M' + L + M
& > &
I_\infty^\epsilon(X'XY'Y:E)_\sigma + 3  + \log \epsilon^{-1}.
\end{eqnarray*}
Then there exists
an {\em $(R_1, R_2, \alpha)$-wiretap quantum interference channel code}
for sending private classical information through $\cN$ where
\[
\alpha < 
2^{2^{14}} \epsilon^{1/6} +
2\sqrt{160 \cdot \epsilon^{1/64} +
	\sqrt{750 \epsilon^{1/32} + 8 e}
      } + 4 e.
\]
Standard Fourier-Motzkin
elimination of $R'_1$, $R'_2$, $L'$, $L$, $M'$, $M$ can now be used
to obtain an inner bound in terms of only $R_1$, $R_2$.
\end{theorem}
The above one shot theorem automatically gives us the following
corollary in the asymptotic iid limit.
\begin{corollary}
\label{cor:wiretapinterference}
Consider the setting of Theorem~\ref{thm:wiretapinterference}.
Consider an octuple of rate per channel use
$(R'_1, L', R_1 - R'_1, L, R'_2, M', R_2 - R'_2, M)$ satisfying
the following inequalities:
\begin{eqnarray*}
R_1 - R'_1 + L 
& < &
I(X : C | X' Y' Q)_\sigma - \log f(e) \\
R_1 + L' + L
& < &
I(X : C | Y' Q)_\sigma \\
R_1 - R'_1 + L + R'_2 + M'
& < &
I(X Y' : C | X' Q)_\sigma - \log f(e) \\
R_1 + L' + L + R'_2 + M'
& < &
I(X Y' : C | Q)_\sigma \\
R_2 - R'_2 + M 
& < &
I(Y : D | X' Y' Q)_\sigma - \log f(e) \\
R_2 + M' + M
& < &
I(Y : D | X' Q)_\sigma \\
R_2 - R'_2 + M + R'_1 + L'
& < &
I(Y X' : D | Y' Q)_\sigma - \log f(e) \\
R_2 + M' + M + R'_1 + L'
& < &
I(Y X' : D | Q)_\sigma \\
& & \\
L'
& > &
I(X':E)_\sigma \\
L' + L
& > &
I(X'X:E)_\sigma \\
M'
& > &
I(Y':E)_\sigma \\
M' + M
& > &
I(Y'Y:E)_\sigma \\
L' + M'
& > &
I(X'Y':E)_\sigma + \log f(e) \\
L' + M' + M
& > &
I(X'Y'Y:E)_\sigma + \log f(e) \\
L' + M' + L
& > &
I(X'XY':E)_\sigma + \log f(e) \\
L' + M' + L + M
& > &
I(X'XY'Y:E)_\sigma.
\end{eqnarray*}
Then there exists 
a per channel use rate {\em $(R_1, R_2)$-wiretap quantum interference 
channel code}
for sending private classical information through $\cN$ with
asymptotically vanishing decoding and privacy error.
Standard Fourier-Motzkin
elimination of $R'_1$, $R'_2$, $L'$, $L$, $M'$, $M$ can now be used
to obtain an inner bound in terms of only $R_1$, $R_2$.
\end{corollary}
We note that the above corollary with loss of pairwise independence
was unknown earlier even in the classical asymptotic iid setting.

\section{Simpler proof of smooth multipartite convex split}
\label{sec:newconvexsplit}
In this section, we give a simpler proof of 
Lemma~\ref{lem:convexsplitflattening} which only tilts and augments
the registers $X$, $Y$ in the state $\sigma^{X^A Y^B M}$. The 
augmentation only doubles the dimensions of $X$, $Y$. The register $M$
is not affected. The advantage of the simpler style of proof is that
it can be extended to prove
our variant of the fully smooth decoupling theorem in the next section.
The simpler proof style is somewhat dual to the more involved proof in 
Section~\ref{sec:convexsplitflatten} which tilted and augmented register
$M$ but left registers $X$, $Y$ unchanged. The advantage of the 
more involved proof was that it could be extended to
prove the smooth multipartite soft covering
lemma without pairwise independence in 
Section~\ref{sec:coveringnonpairwise}. 

As in Section~\ref{sec:convexsplitflatten}, we assume without loss of
generality that the density matrix $\rho^{XYM}$ is normalised. We
flatten by applying the superoperators 
$\cF_\alpha^{X \rightarrow L_X X}$, $\cF_\beta^{Y \rightarrow L_Y Y}$, 
$\cF_\rho^{M \rightarrow L_M M}$ to $\rho^{XYM}$, followed by
renaming $X \equiv L_X X$, $Y \equiv L_Y Y$ and 
$M \equiv L_M M$ as in Section~\ref{sec:convexsplitflatten}. Thus,
we get orthogonal projectors
$\Pi^{XYM}(3)$, $\Pi^{XM}(1)$ and $\Pi^{YM}(2)$ satisfying
Equation~\ref{eq:convsplit7} as in Section~\ref{sec:convexsplitflatten}.
To finish the proof of Lemma~\ref{lem:convexsplitflattening}, we 
only need to show the validity of Equation~\ref{eq:convsplit8}
of Section~\ref{sec:convexsplitflatten}.

We will now need to augment and tilt the Hilbert spaces $X$, $Y$. We
will not touch $M$. This is the place from where the present proof 
starts differing from the proof in Section~\ref{sec:convexsplitflatten}.
Define new Hilbert spaces $\bar{X}$, $\bar{Y}$ isomorphic to $X$, $Y$.
Define the natural isometric bijections 
$\ket{x}^X \mapsto \ket{x}^{\bar{X}}$,
$\ket{y}^Y \mapsto \ket{y}^{\bar{Y}}$.
Define the augmented Hilbert spaces $\hX := X \oplus \bar{X}$,
$\hY := Y \oplus \bar{Y}$. Let $0 < \epsilon < 1$.
Define the tilting map
\[
T_{\epsilon^{1/4}}^{X \rightarrow \hX}:
\ket{x}^X \mapsto
\sqrt{1 - \epsilon^{1/4}} \ket{x}^X + \epsilon^{1/8} \ket{x}^{\bar{X}}.
\]
Define $T_{\epsilon^{1/4}}^{Y \rightarrow \hY}$ similarly.
Define 
\begin{equation}
\label{eq:simpleconvsplit1}
\begin{array}{rcl}
\rho^{\hX \hY M} 
& \equiv & 
\rho^{XYM}, \\ 
\hrho^{\hX \hY M} 
& := &
(T_{\epsilon^{1/4}}^{X \rightarrow \hX} \otimes
 T_{\epsilon^{1/4}}^{Y \rightarrow \hY} \otimes \I^M
)(\rho^{XYM}), \\
\hPi^{\hX \hY M} 
& := &
\one^{\hX \hY M} \\
&  &
{} -
\spanning\{
((T_{\epsilon^{1/4}}^{X \rightarrow \hX} \otimes \I^M)
 (\one^{XM} - \Pi(1)^{XM})) \otimes \one^{\hY}, \\
&  &
~~~~~~~~~~~~~~
((T_{\epsilon^{1/4}}^{Y \rightarrow \hY} \otimes \I^M)
 (\one^{YM} - \Pi(2)^{YM})) \otimes \one^{\hX}, \\
&  &
~~~~~~~~~~~~~~
(T_{\epsilon^{1/4}}^{X \rightarrow \hX} \otimes 
  T_{\epsilon^{1/4}}^{Y \rightarrow \hY} \otimes \I^M
)(\one^{XYM} - \Pi(3)^{XYM})
\}, \\
\hhrho^{\hX \hY M} 
& := &
\hPi^{\hX \hY M} \circ \rho^{\hX \hY M}, \\
\alpha^{\hX}
& := &
\alpha^X,
~~~
\beta^{\hY}
  :=  
\beta^Y, \\
\tau^{\hX^A \hY^B M}
& \equiv &
\tau^{X^A Y^B M}
:=
\alpha^{X^A} \otimes \beta^{Y^B} \otimes \rho^M, \\
\sigma^{\hX^A \hY^B M}
& \equiv &
\sigma^{X^A Y^B M}
:=
\frac{1}{AB} \sum_{a=1}^A \sum_{b=1}^B
\rho^{X_a Y_b M} \otimes \alpha^{X^{-a}} \otimes \beta^{Y^{-b}} \\
& \equiv &
\frac{1}{AB} \sum_{a=1}^A \sum_{b=1}^B
\rho^{\hX_a \hY_b M} \otimes \alpha^{\hX^{-a}} \otimes \beta^{\hY^{-b}}, \\
\hsigma^{\hX^A \hY^B M}
& := &
\frac{1}{AB} \sum_{a=1}^A \sum_{b=1}^B
\hrho^{\hX_a \hY_b M} \otimes \alpha^{\hX^{-a}} \otimes \beta^{\hY^{-b}},\\
\hhsigma^{\hX^A \hY^B M}
& := &
\frac{1}{AB} \sum_{a=1}^A \sum_{b=1}^B
\hhrho^{\hX_a \hY_b M} \otimes \alpha^{\hX^{-a}} \otimes \beta^{\hY^{-b}}.
\end{array}
\end{equation}

The following properties can be proved easily as in 
Section~\ref{sec:convexsplitflatten}.
\begin{equation}
\label{eq:simpleconvsplit2}
\begin{array}{rcl}
\hrho^{\hX M}
& = &
(T_{\epsilon^{1/4}}^{X \rightarrow \hX} \otimes \I^M)(\rho^{XM}),
~~
\hrho^{\hY M}
=
(T_{\epsilon^{1/4}}^{Y \rightarrow \hY} \otimes \I^M)(\rho^{YM}), \\
\|\hrho^{\hX \hY M} - \rho^{\hX \hY M}\|_1
& \leq &
4\sqrt{2} \epsilon^{1/8}, \\
\Tr[\hhrho^{\hX \hY M}] 
& \geq & 
1 - 121 \epsilon^{1/4}, \\
\implies
\|\hhrho^{\hX \hY M} - \rho^{\hX \hY M}\|_1  
& \leq &
22 \epsilon^{1/8}, 
~~
\|\hhrho^{\hX \hY M} - \hrho^{\hX \hY M}\|_1  
\leq
28 \epsilon^{1/8}, \\
\implies
\|\hsigma^{\hX^A \hY^B M} - \sigma^{\hX^A \hY^B M}\|_1  
& \leq &
4\sqrt{2} \epsilon^{1/8},
~~
\|\hsigma^{\hX^A \hY^B M} - \hhsigma^{\hX^A \hY^B M}\|_1  
\leq
28 \epsilon^{1/8}.
\end{array}
\end{equation}
By Lemma~\ref{lem:asyml2}, there exists a subnormalised density
matrix $\sigma^{'\hX^A \hY^B M}$ such that
\begin{equation}
\label{eq:simpleconvsplit3}
\begin{array}{rcl}
\|\sigma^{'\hX^A \hY^B M} - \hsigma^{\hX^A \hY^B M}\|_1
& \leq &
\sqrt{28} \epsilon^{1/16}, \\
\sigma^{'\hX^A \hY^B M}
& \leq &
\hsigma^{\hX^A \hY^B M}, \\
\|\sigma^{'\hX^A \hY^B M}\|_2^2
& \leq &
(1 + 11 \epsilon^{1/16})
\Tr [\hsigma^{\hX^A \hY^B M} \hhsigma^{\hX^A \hY^B M}], \\
\implies
\|\sigma^{'\hX^A \hY^B M} - \sigma^{\hX^A \hY^B M}\|_1
& < &
9 \epsilon^{1/16}.
\end{array}
\end{equation}
From Equation~\ref{eq:convsplit8} and the above calculation, it 
suffices to show
\[
\|\sigma^{'\hX^A \hY^B M} - \tau^{\hX^A \hY^B M}\|_1 <
11 \epsilon^{1/32}.
\]
Define $\Pi^{\hX^A \hY^B M}$ to be the projector from
$\hX^A \hY^B M$ onto 
$
\supp(\tau^{\hX^A \hY^B M}) = 
F_\alpha^{' \otimes A} \otimes
F_\beta^{' \otimes B} \otimes
F'_\rho.
$
Then,
\begin{equation}
\label{eq:simpleconvsplit4}
\Tr[\sigma^{'\hX^A \hY^B M} \Pi^{\hX^A \hY^B M}] \geq
\Tr[\sigma^{\hX^A \hY^B M} \Pi^{\hX^A \hY^B M}] -
\frac{1}{2} \|\sigma^{'\hX^A \hY^B M} - \sigma^{\hX^A \hY^B M}\|_1 \geq
1 - 5\epsilon^{1/16},
\end{equation}
as $\Tr[\sigma^{\hX^A \hY^B M} \Pi^{\hX^A \hY^B M}] = 1$ because
$
\supp(\sigma^{\hX^A \hY^B M}) \leq \Pi^{\hX^A \hY^B M}.
$
By Proposition~\ref{prop:shavedCauchySchwarz}, it thus suffices to show
\[
\sqrt{\Tr[\Pi^{\hX^A \hY^B M}]} \cdot
\|\sigma^{'\hX^A \hY^B M} - \tau^{\hX^A \hY^B M}\|_2 < 
11 \epsilon^{1/32} - 5 \epsilon^{1/32} = 6\epsilon^{1/32}, 
\]
which implies that it suffices to show,
\begin{equation}
\label{eq:simpleconvsplit5}
\|\sigma^{'\hX^A \hY^B M} - \tau^{\hX^A \hY^B M}\|_2^2 < 
\frac{36 \epsilon^{1/16}}{|F'_\alpha|^A |F'_\beta|^B |F'_\rho|}.
\end{equation}
Using Equations~\ref{eq:simpleconvsplit3}, \ref{eq:simpleconvsplit4}, 
\ref{eq:convsplit7}, the left 
hand side of the above equation can be simplified as
\begin{eqnarray*}
\lefteqn{
\|\sigma^{'\hX^A \hY^B M} - \tau^{\hX^A \hY^B M}\|_2^2 
} \\
& = &
\|\sigma^{'\hX^A \hY^B M}\|_2^2 -
2 \Tr[\sigma^{'\hX^A \hY^B M} \tau^{\hX^A \hY^B M}] +
\|\tau^{\hX^A \hY^B M}\|_2^2 \\
& \leq &
(1 + 11\epsilon^{1/16})
\Tr[\hsigma^{\hX^A \hY^B M} \hhsigma^{\hX^A \hY^B M}] -
\frac{2 (1-\sqrt{\epsilon}) 
      \Tr[\sigma^{'\hX^A \hY^B M} \Pi^{\hX^A \hY^B M}]
     }{|F'_\alpha|^A |F'_\beta|^B |F'_\rho|} +
\frac{1+\sqrt{\epsilon}}{|F'_\alpha|^A |F'_\beta|^B |F'_\rho|} \\
& <    &
(1 + 11\epsilon^{1/16})
\Tr[\hsigma^{\hX^A \hY^B M} \hhsigma^{\hX^A \hY^B M}] -
\frac{2 (1-\sqrt{\epsilon}) (1-5\epsilon^{1/16})
     }{|F'_\alpha|^A |F'_\beta|^B |F'_\rho|} +
\frac{1+\sqrt{\epsilon}}{|F'_\alpha|^A |F'_\beta|^B |F'_\rho|} \\
& <    &
(1 + 11\epsilon^{1/16})
\Tr[\hsigma^{\hX^A \hY^B M} \hhsigma^{\hX^A \hY^B M}] -
\frac{1-13\epsilon^{1/16}
     }{|F'_\alpha|^A |F'_\beta|^B |F'_\rho|}.
\end{eqnarray*}
From the above and Equation~\ref{eq:simpleconvsplit5}, it thus
suffices to show
\begin{equation}
\label{eq:simpleconvsplit6}
\Tr[\hsigma^{\hX^A \hY^B M} \hhsigma^{\hX^A \hY^B M}] <
\frac{1 + \epsilon^{1/16}}{|F'_\alpha|^A |F'_\beta|^B |F'_\rho|}.
\end{equation}

We now prove four lemmas which are crucially required to prove
Equation~\ref{eq:simpleconvsplit6}. They are the analogues of
Lemmas~\ref{lem:hrhohhrho3}, \ref{lem:hrhohhrho1}, \ref{lem:hrhohhrho2},
\ref{lem:hrhohhrho} in Section~\ref{sec:convexsplitflatten}.
\begin{lemma}
\label{lem:newhrhohhrho3}
\[
\Tr[\hrho^{\hX \hY M} \hhrho^{\hX \hY M}] \leq
\frac{(1+3\sqrt{\epsilon})
      2^{D^\epsilon_\infty(\rho^{XYM} \| 
			   \alpha^X \otimes \beta^Y \otimes \rho^M)}
     }{|F'_\alpha| |F'_\beta| |F'_\rho|}.
\]
\end{lemma}
\begin{proof}
Using Equations~\ref{eq:convsplit7}, \ref{eq:simpleconvsplit1}, 
\ref{eq:simpleconvsplit2}, we get
\begin{eqnarray*}
\lefteqn{
\Tr[\hrho^{\hX \hY M} \hhrho^{\hX \hY M}]
} \\
& = &
\Tr[((T^{X \rightarrow \hX}_{\epsilon^{1/4}} \otimes 
      T^{Y \rightarrow \hY}_{\epsilon^{1/4}} \otimes \I^M
     )(\rho^{X Y M})
    ) (\hPi^{\hX \hY M} \circ \rho^{\hX \hY M})
] \\
& = &
\Tr[((T^{X \rightarrow \hX}_{\epsilon^{1/4}} \otimes 
      T^{Y \rightarrow \hY}_{\epsilon^{1/4}} \otimes \I^M
     )(\rho^{X Y M})
    ) \\
&  &
~~~~~~~~
    ((\one^{\hX \hY M} - 
      ((T^{X \rightarrow \hX}_{\epsilon^{1/4}} \otimes 
        T^{Y \rightarrow \hY}_{\epsilon^{1/4}} \otimes \I^M 
       )(\one^{XYM} - \Pi(3)^{XYM})
      ) 
     ) \circ (\hPi^{\hX \hY M} \circ \rho^{\hX \hY M})
    )
] \\
& = &
\Tr[(((T^{X \rightarrow \hX}_{\epsilon^{1/4}} \otimes 
       T^{Y \rightarrow \hY}_{\epsilon^{1/4}} \otimes \I^M 
      )(\one^{X Y M})
     ) \circ
     ((T^{X \rightarrow \hX}_{\epsilon^{1/4}} \otimes 
       T^{Y \rightarrow \hY}_{\epsilon^{1/4}} \otimes \I^M 
      )(\rho^{X Y M})
     )
    ) \\ 
&  &
~~~~~~~~
    ((\one^{\hX \hY M} - 
      ((T^{X \rightarrow \hX}_{\epsilon^{1/4}} \otimes 
        T^{Y \rightarrow \hY}_{\epsilon^{1/4}} \otimes \I^M 
       )(\one^{XYM} - \Pi(3)^{XYM})
      ) 
     ) \circ \hhrho^{\hX \hY M}
    )
] \\
& = &
\Tr[((T^{X \rightarrow \hX}_{\epsilon^{1/4}} \otimes 
      T^{Y \rightarrow \hY}_{\epsilon^{1/4}} \otimes \I^M 
     )(\rho^{X Y M})
    ) \\
&  &
~~~~~~~~
    ((((T^{X \rightarrow \hX}_{\epsilon^{1/4}} \otimes 
        T^{Y \rightarrow \hY}_{\epsilon^{1/4}} \otimes \I^M 
       )(\one^{X Y M})
      ) \\
&   &
~~~~~~~~~~~~~~~
{} \circ
      (\one^{\hX \hY M} - 
       ((T^{X \rightarrow \hX}_{\epsilon^{1/4}} \otimes 
         T^{Y \rightarrow \hY}_{\epsilon^{1/4}} \otimes \I^M 
        )(\one^{XYM} - \Pi(3)^{XYM})
       )
      )
     ) \\
& &
~~~~~~~~~~~~~~~~~~~~~~
{} \circ \hhrho^{\hX \hY M}
    )
] \\
& = &
\Tr[((T^{X \rightarrow \hX}_{\epsilon^{1/4}} \otimes 
      T^{Y \rightarrow \hY}_{\epsilon^{1/4}} \otimes \I^M 
     )(\rho^{X Y M})
    ) \\
&  &
~~~~~~~~
    ((((T^{X \rightarrow \hX}_{\epsilon^{1/4}} \otimes 
        T^{Y \rightarrow \hY}_{\epsilon^{1/4}} \otimes \I^M 
       )(\one^{X Y M})
      ) \\
&  &
~~~~~~~~~~~~~~~
{} - ((T^{X \rightarrow \hX}_{\epsilon^{1/4}} \otimes 
        T^{Y \rightarrow \hY}_{\epsilon^{1/4}} \otimes \I^M 
       )(\one^{XYM} - \Pi(3)^{XYM})
      )
     ) \circ \hhrho^{\hX \hY M}
    )
] \\
& = &
\Tr[((T^{X \rightarrow \hX}_{\epsilon^{1/4}} \otimes 
      T^{Y \rightarrow \hY}_{\epsilon^{1/4}} \otimes \I^M 
     )(\rho^{X Y M})
    ) \\
&  &
~~~~~~~~
    (((T^{X \rightarrow \hX}_{\epsilon^{1/4}} \otimes 
       T^{Y \rightarrow \hY}_{\epsilon^{1/4}} \otimes \I^M 
      )(\Pi(3)^{XYM})
     ) \circ \hhrho^{\hX \hY M}
    )
] \\
& = &
\Tr[(((T^{X \rightarrow \hX}_{\epsilon^{1/4}} \otimes
       T^{Y \rightarrow \hY}_{\epsilon^{1/4}} \otimes \I^M 
      )(\Pi(3)^{XYM})
     ) \\
&  &
~~~~~~~~
{}  \circ
     ((T^{X \rightarrow \hX}_{\epsilon^{1/4}} \otimes 
       T^{Y \rightarrow \hY}_{\epsilon^{1/4}} \otimes \I^M 
      )(\rho^{X Y M})
     )
    ) \hhrho^{\hX \hY M}
] \\
& = &
\Tr[((T^{X \rightarrow \hX}_{\epsilon^{1/4}} \otimes 
      T^{Y \rightarrow \hY}_{\epsilon^{1/4}} \otimes \I^M 
     )(\Pi(3)^{XYM} \circ \rho^{X Y M})
    ) \hhrho^{\hX \hY M}
] \\
& \leq &
\|(T^{X \rightarrow \hX}_{\epsilon^{1/4}} \otimes 
   T^{Y \rightarrow \hY}_{\epsilon^{1/4}} \otimes \I^M 
  )(\Pi(3)^{XYM} \circ \rho^{X Y M})
\|_\infty \cdot \|\hhrho^{\hX \hY M}\|_1 \\
&   =  &
\|\Pi(3)^{XYM} \circ \rho^{X Y M}\|_\infty \cdot 
\Tr [\hhrho^{\hX \hY M}] 
\;\leq\;
\frac{(1+3\sqrt{\epsilon}) 
      2^{D^\epsilon_\infty(\rho^{XYM} \| \alpha^X \otimes 
                                         \beta^Y \otimes \rho^M)}
     }{|F'_\alpha| |F'_\beta| |F'_\rho|}.
\end{eqnarray*}
Above, we used the fact that
\[
\one^{\hX \hY M} - 
(T^{X \rightarrow \hX}_{\epsilon^{1/4}} \otimes 
 T^{Y \rightarrow \hY}_{\epsilon^{1/4}} \otimes \I^M 
)(\one^{XYM} - \Pi(3)^{XYM}) \geq
\hPi^{\hX \hY M}
\]
in the third equality,
\begin{eqnarray*}
\lefteqn{
((T^{X \rightarrow \hX}_{\epsilon^{1/4}} \otimes 
  T^{Y \rightarrow \hY}_{\epsilon^{1/4}} \otimes \I^M 
 )(\one^{X Y M})
) 
} \\
&  &
~~~~~~~~~~
{} \circ
(\one^{\hX \hY M} - 
 (T^{X \rightarrow \hX}_{\epsilon^{1/4}} \otimes 
  T^{Y \rightarrow \hY}_{\epsilon^{1/4}} \otimes \I^M 
 )(\one^{XYM} - \Pi(3)^{XYM})
) \\
& = &
(T^{X \rightarrow \hX}_{\epsilon^{1/4}} \otimes 
 T^{Y \rightarrow \hY}_{\epsilon^{1/4}} \otimes \I^M 
)(\one^{X Y M}) - 
(T^{X \rightarrow \hX}_{\epsilon^{1/4}} \otimes 
 T^{Y \rightarrow \hY}_{\epsilon^{1/4}} \otimes \I^M 
)(\one^{X Y M} - \Pi(3)^{X Y M})
\end{eqnarray*}
in the sixth equality, 
$
T^{X \rightarrow \hX}_{\epsilon^{1/4}} \otimes
T^{Y \rightarrow \hY}_{\epsilon^{1/4}} \otimes \I^M
$ 
is an isometry in the tenth equality. 
The proof of the lemma is now complete.
\end{proof}

\begin{lemma}
\label{lem:newhrhohhrho1}
\[
\Tr[\hrho^{\hX M} \hhrho^{\hX M}] \leq
\frac{(1+3\sqrt{\epsilon})
      2^{D^\epsilon_\infty(\rho^{XM} \| \alpha^X \otimes \rho^M)}
     }{|F'_\alpha| |F'_\rho|}.
\]
\end{lemma}
\begin{proof}
Using Equations~\ref{eq:convsplit7}, \ref{eq:simpleconvsplit1}, 
\ref{eq:simpleconvsplit2}, we get
\begin{eqnarray*}
\lefteqn{
\Tr[\hrho^{\hX M} \hhrho^{\hX M}]
} \\
& = &
\Tr[(\hrho^{\hX M} \otimes \one^{\hY}) \hhrho^{\hX \hY \hM}] 
\;=\;
\Tr[(((T^{X \rightarrow \hX}_{\epsilon^{1/4}} \otimes \I^M)(\rho^{X M}))
	\otimes \one^{\hY}) (\hPi^{\hX \hY M} \circ \rho^{\hX \hY M})
   ] \\
& = &
\Tr[(((T^{X \rightarrow \hX}_{\epsilon^{1/4}} \otimes \I^M)(\rho^{X M}))
     \otimes \one^{\hY}
    ) \\
&  &
~~~~~~~~
    ((\one^{\hX \hY M} - 
      ((T^{X \rightarrow \hX}_{\epsilon^{1/4}} \otimes \I^M)
       (\one^{XM} - \Pi(1)^{XM})
      ) \otimes \one^{\hY}
     ) \circ (\hPi^{\hX \hY M} \circ \rho^{\hX \hY M})
    )
] \\
& = &
\Tr[((((T^{X \rightarrow \hX}_{\epsilon^{1/4}} \otimes \I^M)(\one^{X M}))
      \circ
      ((T^{X \rightarrow \hX}_{\epsilon^{1/4}} \otimes \I^M)(\rho^{X M}))
     ) \otimes \one^{\hY}
    ) \\
&  &
~~~~~~~~
    (((\one^{\hX M} - 
       ((T^{X \rightarrow \hX}_{\epsilon^{1/4}} \otimes \I^M)
        (\one^{XM} - \Pi(1)^{XM})
       )
      ) \otimes \one^{\hY}
     ) \circ \hhrho^{\hX \hY M}
    )
] \\
& = &
\Tr[(((T^{X \rightarrow \hX}_{\epsilon^{1/4}} \otimes \I^M)(\rho^{X M}))
     \otimes \one^{\hY}
    ) \\
&  &
~~~~~~~~
    (((((T^{X \rightarrow \hX}_{\epsilon^{1/4}} \otimes \I^M)(\one^{X M}))
       \circ
       (\one^{\hX M} - 
        ((T^{X \rightarrow \hX}_{\epsilon^{1/4}} \otimes \I^M)
         (\one^{XM} - \Pi(1)^{XM})
        )
       )
      ) \otimes \one^{\hY}
     ) \\
& &
~~~~~~~~~~~~~~~
{} \circ \hhrho^{\hX \hY M}
    )
] \\
& = &
\Tr[(((T^{X \rightarrow \hX}_{\epsilon^{1/4}} \otimes \I^M)(\rho^{X M}))
     \otimes \one^{\hY}
    ) \\
&  &
~~~~~~~~
    (((((T^{X \rightarrow \hX}_{\epsilon^{1/4}} \otimes \I^M)(\one^{X M}))
       - ((T^{X \rightarrow \hX}_{\epsilon^{1/4}} \otimes \I^M)
          (\one^{XM} - \Pi(1)^{XM})
         )
      ) \otimes \one^{\hY}
     ) \circ \hhrho^{\hX \hY M}
    )
] \\
& = &
\Tr[(((T^{X \rightarrow \hX}_{\epsilon^{1/4}} \otimes \I^M)(\rho^{X M}))
     \otimes \one^{\hY}
    ) \\
&  &
~~~~~~~~
    ((((T^{X \rightarrow \hX}_{\epsilon^{1/4}} \otimes \I^M)(\Pi(1)^{XM})
      ) \otimes \one^{\hY}
     ) \circ \hhrho^{\hX \hY M}
    )
] \\
& = &
\Tr[((((T^{X \rightarrow \hX}_{\epsilon^{1/4}} \otimes \I^M)(\Pi(1)^{XM}))
      \circ
      ((T^{X \rightarrow \hX}_{\epsilon^{1/4}} \otimes \I^M)(\rho^{X M}))
     ) \otimes \one^{\hY} 
    ) \hhrho^{\hX \hY M}
] \\
& = &
\Tr[(((T^{X \rightarrow \hX}_{\epsilon^{1/4}} \otimes \I^M)
      (\Pi(1)^{XM} \circ \rho^{X M})
     ) \otimes \one^{\hY} 
    ) \hhrho^{\hX \hY M}
] \\
& \leq &
\|((T^{X \rightarrow \hX}_{\epsilon^{1/4}} \otimes \I^M)
      (\Pi(1)^{XM} \circ \rho^{X M})
  ) \otimes \one^{\hY} 
\|_\infty \cdot \|\hhrho^{\hX \hY M}\|_1 \\
&   =  &
\|\Pi(1)^{XM} \circ \rho^{X M}\|_\infty \cdot 
\Tr [\hhrho^{\hX \hY M}] 
\;\leq\;
\frac{(1+3\sqrt{\epsilon}) 
      2^{D^\epsilon_\infty(\rho^{XM} \| \alpha^X \otimes \rho^M)}
     }{|F'_\alpha| |F'_\rho|}.
\end{eqnarray*}
Above, we used the fact that
\[
\one^{\hX \hY M} - 
((T^{X \rightarrow \hX}_{\epsilon^{1/4}} \otimes \I^M)
       (\one^{XM} - \Pi(1)^{XM})
) \otimes \one^{\hY}
\geq
\hPi^{\hX \hY M}
\]
in the third equality,
\begin{eqnarray*}
\lefteqn{
((T^{X \rightarrow \hX}_{\epsilon^{1/4}} \otimes \I^M)(\one^{X M}))
\circ
(\one^{\hX M} - 
 ((T^{X \rightarrow \hX}_{\epsilon^{1/4}} \otimes \I^M)
  (\one^{XM} - \Pi(1)^{XM})
 )
)
} \\
& = &
(T^{X \rightarrow \hX}_{\epsilon^{1/4}} \otimes \I^M)(\one^{X M})
- (T^{X \rightarrow \hX}_{\epsilon^{1/4}} \otimes \I^M)
  (\one^{XM} - \Pi(1)^{XM})
\end{eqnarray*}
in the sixth equality, 
$T^{X \rightarrow \hX}_{\epsilon^{1/4}} \otimes \I^M$ is
an isometry in the tenth equality. 
The proof of the lemma is now complete.
\end{proof}

\begin{lemma}
\label{lem:newhrhohhrho2}
\[
\Tr[\hrho^{\hY M} \hhrho^{\hY M}] \leq
\frac{(1+3\sqrt{\epsilon})
      2^{D^\epsilon_\infty(\rho^{YM} \| \beta^Y \otimes \rho^M)}
     }{|F'_\beta| |F'_\rho|}.
\]
\end{lemma}
\begin{proof}
Similar to proof of Lemma~\ref{lem:newhrhohhrho1} above.
\end{proof}

\begin{lemma}
\label{lem:newhrhohhrho}
\[
\Tr[\hrho^{M} \hhrho^{M}] \leq
\frac{1+\sqrt{\epsilon}}{|F'_\rho|}.
\]
\end{lemma}
\begin{proof}
Using Equations~\ref{eq:simpleconvsplit1} and \ref{eq:convsplit7}, we get
\[
\Tr[\hrho^{M} \hhrho^{M}] = 
\Tr[\rho^{M} \hhrho^{M}] \leq
\frac{1 + \sqrt{\epsilon}}{|F'_\rho|} \Tr[\one^{F'_\rho} \hhrho^{M}] \leq 
\frac{1 + \sqrt{\epsilon}}{|F'_\rho|} \Tr[\hhrho^{M}] \leq
\frac{1 + \sqrt{\epsilon}}{|F'_\rho|}.
\]
This completes the proof of the lemma.
\end{proof}

We now prove the following lemma which completes the simpler proof
of Lemma~\ref{lem:convexsplitflattening} by finally
showing the validity of Equation~\ref{eq:simpleconvsplit6}.
\begin{lemma}
\label{lem:newasymconvexsplit}
Suppose 
\begin{eqnarray*}
\log A 
& > &
D^\epsilon_\infty(\rho^{XM} \| \alpha^X \otimes \rho^M) +
\log \epsilon^{-1/20}, \\
\log B 
& > &
D^\epsilon_\infty(\rho^{YM} \| \beta^Y \otimes \rho^M) +
\log \epsilon^{-1/20}, \\
\log A + \log B 
& > &
D^\epsilon_\infty(\rho^{XYM} \| \alpha^X \otimes \beta^Y \otimes \rho^M) +
\log \epsilon^{-1/20}.
\end{eqnarray*}
Then for sufficiently small $\epsilon$,
\[
\Tr[\hsigma^{\hX^A \hY^B M} \hhsigma^{\hX^A \hY^B M}] <
\frac{1 + \epsilon^{1/16}}{|F'_\alpha|^A |F'_\beta|^B |F'_\rho|}.
\]
\end{lemma}
\begin{proof}
We follow the proof method of Lemma~\ref{lem:asymconvexsplit} above,
but use 
Lemmas~\ref{lem:newhrhohhrho3}, \ref{lem:newhrhohhrho1},
\ref{lem:newhrhohhrho2}, \ref{lem:newhrhohhrho} 
instead of
Lemmas~\ref{lem:hrhohhrho3}, \ref{lem:hrhohhrho1},
\ref{lem:hrhohhrho2}, \ref{lem:hrhohhrho}. 
We get
\[
\Tr[\hsigma^{\hX^A \hY^B M} \hhsigma^{\hX^A \hY^B M}] <
\frac{1 + 40 \epsilon^{1/20}}{|F'_\alpha|^A |F'_\beta|^B |F'_\rho|} <
\frac{1 + \epsilon^{1/16}}{|F'_\alpha|^A |F'_\beta|^B |F'_\rho|}
\]
for sufficiently small $\epsilon$.
This finishes the proof of the lemma.
\end{proof}

The proof of
Lemma~\ref{lem:convexsplitflattening} is now complete.

\section{Fully smooth multipartite decoupling}
\label{sec:decoupling}
In this section, we prove a fully smooth multipartite decoupling
theorem where we have to enlarge the Hilbert space dimensions of the
senders by a factor of $2$. Our proof is inspired by the simpler
proof of the fully smooth multipartite convex split lemma in 
Section~\ref{sec:newconvexsplit}. 

The EPR state is defined as
\[
\Phi^{A'A} := 
|A|^{-1} \sum_{a_1,a_2 = 1}^{|A|} \ket{a_1, a_1}^{A'A}\bra{a_2, a_2},
\]
where $A'$ is another Hilbert space of the same dimension as $A$.
We remark that though Theorem~\ref{thm:decoupling} is stated in terms
of expectation over the Haar measure on unitaries as is the convention,
it continues to hold without change for expectation over a perfect
2-design of unitaries, and with minor additive terms for expectation
over an approximate 2-design of unitaries. 
\begin{theorem}
\label{thm:decoupling}
Let $\cT^{A_1 A_2 \rightarrow E}$ be a completely positive trace 
non increasing superoperator with Choi state given by 
$
\tau^{A'_1 A'_2 E} :=
(\cT^{A_1 A_2 \rightarrow E} \otimes \I^{A'_1 A'_2})
(\Phi^{A'_1 A_1} \otimes \Phi^{A'_2 A_2}).
$.
Let $\rho^{A_1 A_2 R}$ be a subnormalised density matrix. 
Define $\hA_1 := A_1 \otimes \C^2$. Similarly define $\hA'_1$,
$\hA_2$, $\hA'_2$.
Define 
$
\rho^{\hA_1 \hA_2 R} := 
\rho^{A_1 A_2 R} \otimes \ketbra{0}^{\C^2} \otimes \ketbra{0}^{\C^2},
$
$
\tau^{\hA'_1 \hA'_2 E} := 
\tau^{A'_1 A'_2 E} \otimes \ketbra{0}^{\C^2} \otimes \ketbra{0}^{\C^2}.
$
Let $\cT^{\hA_1 \hA_2 \rightarrow E}$ be the completely positive
superoperator corresponding to the inverse Choi image of
$\tau^{\hA'_1 \hA'_2 E}$. Then $\cT^{\hA_1 \hA_2 \rightarrow E}$ is
the completely positive superoperator corresponding to projecting
from $\hA_1 \otimes \hA_2$ onto 
$(A_1\otimes\ketbra{0}^{\C^2}) \otimes (A_2\otimes\ketbra{0}^{\C^2})$ 
followed by
applying $4 \cT^{A_1 A_2 \rightarrow E}$.

Let $U^{\hA_1}$, $U^{\hA_2}$ be unitary matrices on their respective
Hilbert spaces.  Let $0 < \epsilon < 1$. Then,
\[
\E_{U^{\hA_1} U^{\hA_2}}[
\|
(\cT^{\hA_1 \hA_2 \rightarrow E} \otimes \I^R)
((U^{\hA_1} \otimes U^{\hA_2} \otimes \one^R) \circ 
\rho^{\hA_1 \hA_2 R}
) - 
\tau^E \otimes \rho^R\|_1
] \leq 
18 \epsilon^{1/128} 
\]
if
\begin{eqnarray*}
\Hmin^\epsilon(A_1 | R)_\rho +
\Hmin^\epsilon(A'_1 | E)_\tau
& > &
\log \epsilon^{-1}, \\
\Hmin^\epsilon(A_2 | R)_\rho +
\Hmin^\epsilon(A'_2 | E)_\tau 
& > &
\log \epsilon^{-1}, \\
\Hmin^\epsilon(A_1 A_2 | R)_\rho +
\Hmin^\epsilon(A'_1 A'_2 | E)_\tau
& > &
\log \epsilon^{-1}, \\
|A_1|^2
& > &
\epsilon^{-1}, 
~~
|A_2|^2
 > 
\epsilon^{-1}, 
\end{eqnarray*}
where the expectation on the left is taken over independent choices
of unitaries $U^{\hA_1}$, $U^{\hA_2}$ from their respective Haar measures.
\end{theorem}

We now proceed with the constructions required to prove
Theorem~\ref{thm:decoupling}.
Observe without loss of generality that $\rho^{A_1 A_2 R}$ can be
assumed to be a normalised quantum state. Let $L_R$ and $L_E$ be the
ancilla Hilbert spaces required for flattening $\rho^R$ and $\tau^E$
respectively. By flattening using the CPTP superoperators 
$\cF_{\rho}^{R \rightarrow L_R R}$,
$\cF_{\tau}^{E \rightarrow L_E E}$, and by applying 
Proposition~\ref{prop:flattening} we observe that it suffices to show
\begin{eqnarray*}
\lefteqn{
\E_{U^{\hA_1} U^{\hA_2}}[
\|
((\cF_\tau^{E \rightarrow L_E E} \cdot \cT^{\hA_1 \hA_2 \rightarrow E}) 
 \otimes \cF_\rho^{R \rightarrow L_R R}
)((U^{\hA_1} \otimes U^{\hA_2} \otimes \one^R) \circ \rho^{\hA_1 \hA_2 R}
 ) 
} \\
&  &
~~~~~~~~~~~~~~
{} - 
\cF_\tau^{E \rightarrow L_E E}(\tau^E) \otimes 
\cF_\rho^{R \rightarrow L_R R}(\rho^R)\|_1
] \leq 
81 \epsilon^{1/64}.
\end{eqnarray*}
Henceforth, we shall redefine 
\begin{eqnarray*}
L_R R 
& \equiv & 
R, \\
L_E E 
& \equiv & 
E, \\
\cF_\tau^{E \rightarrow L_E E} \cdot \cT^{\hA_1 \hA_2 \rightarrow E}
& \equiv &
\cT^{\hA_1 \hA_2 \rightarrow E}, \\
(\I^{\hA_1} \otimes \I^{\hA_2} \otimes \cF_\rho^{R \rightarrow L_R R}) 
(\rho^{\hA_1 \hA_2 R})
& \equiv &
\rho^{\hA_1 \hA_2 R}, \\
((\cF_\tau^{E \rightarrow L_E E} \cdot \cT^{\hA_1 \hA_2 \rightarrow E})
 \otimes \I^{\hA'_1 \hA'_2}
)(\Phi^{\hA'_1 \hA_1} \otimes \Phi^{\hA'_2 \hA_2}) = 
(\cF_\tau^{E \rightarrow L_E E} \otimes \I^{\hA'_1 \hA'_2})
(\tau^{\hA'_1 \hA'_2 E}) 
& \equiv &
\tau^{\hA'_1 \hA'_2 E}.
\end{eqnarray*}
Thus with this changed notation we will show
\begin{equation}
\label{eq:decouplingstart}
\E_{U^{\hA_1} U^{\hA_2}}[
\|
(\cT^{\hA_1 \hA_2 \rightarrow E} \otimes \I^R)
((U^{\hA_1} \otimes U^{\hA_2} \otimes \one^R) \circ 
\rho^{\hA_1 \hA_2 R}
) - 
\tau^E \otimes \rho^R\|_1
] \leq 
81 \epsilon^{1/64}.
\end{equation}
Note that from its definition, 
$\cT^{\hA_1 \hA_2 \rightarrow E}$ increases the trace by at most a
multiplicative factor of $4$.

Define the tilting map
\[
T_{\epsilon^{1/4}}^{A_1 \rightarrow \hA_1}:
\ket{a_1}^{A_1} \mapsto
\sqrt{1 - \epsilon^{1/4}} \ket{a_1}^{A_1}\ket{0} + 
\epsilon^{1/8} \ket{a_1}^{A_1}\ket{1}.
\]
Define $T_{\epsilon^{1/4}}^{A_2 \rightarrow \hA_2}$,
$T_{\epsilon^{1/4}}^{A'_1 \rightarrow \hA'_1}$, 
$T_{\epsilon^{1/4}}^{A'_2 \rightarrow \hA'_2}$ similarly.
Let $\Pi(1)^{A_1 R}$, $\Pi(2)^{A_2 R}$, $\Pi(3)^{A_1 A_2 R}$, 
$\Pi(1)^{A'_1 E}$, $\Pi(2)^{A'_2 E}$, $\Pi(3)^{A'_1 A'_2 E}$ be the
projectors promised by Proposition~\ref{prop:projsmoothing} achieving
the respective smooth conditional min entropy quantities for the 
flattened states. More precisely, Propositions~\ref{prop:flattening},
\ref{prop:projsmoothing} guarantee, when the flattening parameter
$\delta$ is small enough, that
\begin{equation}
\label{eq:decoupling1}
\begin{array}{rcl}
\Pi(1)^{A_1 R} \circ \rho^{A_1 R}
& \leq &
\frac{(1+ 3\sqrt{\epsilon}) 2^{-\Hmin^\epsilon(A_1|R)_\rho}}
     {|F'_\rho|}
\,\one^{A_1 F'_\rho}, 
~~
\Pi(2)^{A_2 R} \circ \rho^{A_2 R}
  \leq  
\frac{(1+ 3\sqrt{\epsilon}) 2^{-\Hmin^\epsilon(A_2|R)_\rho}}
     {|F'_\rho|}
\,\one^{A_2 F'_\rho}, \\
\Pi(3)^{A_1 A_2 R} \circ \rho^{A_1 A_2 R}
& \leq &
\frac{(1+ 3\sqrt{\epsilon}) 2^{-\Hmin^\epsilon(A_1 A_2|R)_\rho}}
     {|F'_\rho|}
\,\one^{A_1 A_2 F'_\rho}, \\
\Pi(1)^{A'_1 E} \circ \rho^{A'_1 E}
& \leq &
\frac{(1+ 3\sqrt{\epsilon}) 2^{-\Hmin^\epsilon(A'_1|E)_\rho}}
     {|F'_\rho|}
\,\one^{A'_1 F'_\rho}, 
~~
\Pi(2)^{A'_2 E} \circ \rho^{A'_2 E}
  \leq  
\frac{(1+ 3\sqrt{\epsilon}) 2^{-\Hmin^\epsilon(A'_2|E)_\rho}}
     {|F'_\rho|}
\,\one^{A'_2 F'_\rho}, \\
\Pi(3)^{A'_1 A'_2 E} \circ \rho^{A'_1 A'_2 E}
& \leq &
\frac{(1+ 3\sqrt{\epsilon}) 2^{-\Hmin^\epsilon(A'_1 A'_2|E)_\rho}}
     {|F'_\rho|}
\,\one^{A'_1 A'_2 F'_\rho}.
\end{array}
\end{equation}
Define 
\begin{equation}
\label{eq:decoupling2}
\begin{array}{rcl}
\hrho^{\hA_1 \hA_2 R} 
& := &
(T_{\epsilon^{1/4}}^{A_1 \rightarrow \hA_1} \otimes
 T_{\epsilon^{1/4}}^{A_2 \rightarrow \hA_2} \otimes \I^R
)(\rho^{A_1 A_2 R}), \\
\htau^{\hA'_1 \hA'_2 E} 
& := &
(T_{\epsilon^{1/4}}^{A'_1 \rightarrow \hA'_1} \otimes
 T_{\epsilon^{1/4}}^{A'_2 \rightarrow \hA'_2} \otimes \I^E
)(\tau^{A'_1 A'_2 E}), \\
\hPi^{\hA_1 \hA_2 R} 
& := &
\one^{\hA_1 \hA_2 R} \\
&  &
{} -
\spanning\{
((T_{\epsilon^{1/4}}^{A_1 \rightarrow \hA_1} \otimes \I^R)
 (\one^{A_1 R} - \Pi(1)^{A_1 R})) \otimes \one^{\hA_2}, \\
&  &
~~~~~~~~~~~~~~
((T_{\epsilon^{1/4}}^{A_2 \rightarrow \hA_2} \otimes \I^R)
 (\one^{A_2 R} - \Pi(2)^{A_2 R})) \otimes \one^{\hA_1}, \\
&  &
~~~~~~~~~~~~~~
(T_{\epsilon^{1/4}}^{A_1 \rightarrow \hA_1} \otimes 
  T_{\epsilon^{1/4}}^{A_2 \rightarrow \hA_2} \otimes \I^R
)(\one^{A_1 A_2 R} - \Pi(3)^{A_1 A_2 R})
\}, \\
\hPi^{\hA'_1 \hA'_2 E} 
& := &
\one^{\hA'_1 \hA'_2 E} \\
&  &
{} -
\spanning\{
((T_{\epsilon^{1/4}}^{A'_1 \rightarrow \hA'_1} \otimes \I^E)
 (\one^{A'_1 E} - \Pi(1)^{A'_1 E})) \otimes \one^{\hA'_2}, \\
&  &
~~~~~~~~~~~~~~
((T_{\epsilon^{1/4}}^{A'_2 \rightarrow \hA'_2} \otimes \I^E)
 (\one^{A'_2 E} - \Pi(2)^{A'_2 E})) \otimes \one^{\hA'_1}, \\
&  &
~~~~~~~~~~~~~~
(T_{\epsilon^{1/4}}^{A'_1 \rightarrow \hA'_1} \otimes 
  T_{\epsilon^{1/4}}^{A'_2 \rightarrow \hA'_2} \otimes \I^E
)(\one^{A'_1 A'_2 E} - \Pi(3)^{A'_1 A'_2 E})
\}, \\
\hhrho^{\hA_1 \hA_2 R} 
& := &
\hPi^{\hA_1 \hA_2 R} \circ \rho^{\hA_1 \hA_2 R}, 
~~
\hhtau^{\hA'_1 \hA'_2 E} 
  :=  
\hPi^{\hA'_1 \hA'_2 E} \circ \tau^{\hA'_1 \hA'_2 E}, \\
\sigma(U^{\hA_1}, U^{\hA_2})^{ER}
& := &
(\cT^{\hA_1 \hA_2 \rightarrow E} \otimes \I^R)
((U^{\hA_1} \otimes U^{\hA_2} \otimes \one^R) \circ 
\rho^{\hA_1 \hA_2 R}), \\
\hsigma(U^{\hA_1}, U^{\hA_2})^{ER}
& := &
(\hcT^{\hA_1 \hA_2 \rightarrow E} \otimes \I^R)
((U^{\hA_1} \otimes U^{\hA_2} \otimes \one^R) \circ 
\hrho^{\hA_1 \hA_2 R}), \\
\hhsigma(U^{\hA_1}, U^{\hA_2})^{ER}
& := &
(\hhcT^{\hA_1 \hA_2 \rightarrow E} \otimes \I^R)
((U^{\hA_1} \otimes U^{\hA_2} \otimes \one^R) \circ 
\hhrho^{\hA_1 \hA_2 R}), \\
\end{array}
\end{equation}
where $\hcT^{\hA_1 \hA_2 \rightarrow E}$, 
$\hhcT^{\hA_1 \hA_2 \rightarrow E}$ are the completely positive
superoperators arising as the inverse Choi images of
$\htau^{\hA'_1 \hA'_2 E}$, $\hhtau^{\hA'_1 \hA'_2 E}$.
Note that $\rho^{A_1 A_2 R}$, $\rho^{\hA_1 \hA_2 R}$, 
$\hrho^{\hA_1 \hA_2 R}$ are normalised density matrices, and
$\hhrho^{\hA_1 \hA_2 R}$, $\tau^{A'_1 A'_2 E}$, $\tau^{\hA'_1 \hA'_2 E}$,
$\htau^{\hA'_1 \hA'_2 E}$, $\hhtau^{\hA'_1 \hA'_2 E}$ are subnormalised
density matrices. Even though $\htau^{\hA'_1 \hA'_2 E}$, 
$\hhtau^{\hA'_1 \hA'_2 E}$ are subnormalised, their inverse Choi images
$\hcT^{\hA_1 \hA_2 \rightarrow E}$, $\hhcT^{\hA_1 \hA_2 \rightarrow E}$
can increase trace but by at most a multiplicative factor of $4$.

The following properties can be proved easily as in 
Sections~\ref{sec:convexsplitflatten}, \ref{sec:newconvexsplit}.
\begin{equation}
\label{eq:decoupling3}
\begin{array}{rcl}
\frac{(1-\sqrt{\epsilon})(\Tr \tau)}{|F'_\tau| |F'_\rho|}
\one^{F'_\tau F'_\rho}
& \leq &
\tau^E \otimes \rho^R
\leq
\frac{(1+\sqrt{\epsilon})(\Tr \tau)}{|F'_\tau| |F'_\rho|}
\one^{F'_\tau F'_\rho}, \\
\hrho^{\hA_1 R}
& = &
(T_{\epsilon^{1/4}}^{A_1 \rightarrow \hA_1} \otimes \I^R)(\rho^{A_1 R}),
~
\hrho^{\hA_2 R}
=
(T_{\epsilon^{1/4}}^{A_2 \rightarrow \hA_2} \otimes \I^R)(\rho^{A_2 R}),\\
\htau^{\hA'_1 R}
& = &
(T_{\epsilon^{1/4}}^{A'_1 \rightarrow \hA'_1} \otimes \I^R)(\tau^{A'_1 R}),
~
\htau^{\hA'_2 R}
=
(T_{\epsilon^{1/4}}^{A'_2 \rightarrow \hA'_2}\otimes\I^R)(\tau^{A'_2 R}),\\
\|\hrho^{\hA_1 \hA_2 R} - \rho^{\hA_1 \hA_2 R}\|_1
& \leq &
4\sqrt{2} \epsilon^{1/8}, 
~~
\|\htau^{\hA'_1 \hA'_2 R} - \tau^{\hA'_1 \hA'_2 R}\|_1
  \leq  
4\sqrt{2} \epsilon^{1/8} (\Tr \tau), \\
\Tr[\hhrho^{\hA_1 \hA_2 R}] 
& \geq & 
1 - 121 \epsilon^{1/4}, 
~~
\Tr[\hhtau^{\hA'_1 \hA'_2 E}] 
  \geq   
(1 - 121 \epsilon^{1/4}) (\Tr \tau), \\
\implies
\|\hhrho^{\hA_1 \hA_2 R} - \rho^{\hA_1 \hA_2 R}\|_1  
& \leq &
22 \epsilon^{1/8}, 
~~
\|\hhrho^{\hA_1 \hA_2 R} - \hrho^{\hA_1 \hA_2 R}\|_1  
\leq
28 \epsilon^{1/8}, \\
\implies
\|\hhtau^{\hA'_1 \hA'_2 E} - \tau^{\hA'_1 \hA'_2 E}\|_1  
& \leq &
22 \epsilon^{1/8} (\Tr \tau), 
~~
\|\hhtau^{\hA'_1 \hA'_2 E} - \htau^{\hA'_1 \hA'_2 E}\|_1  
\leq
28 \epsilon^{1/8} (\Tr \tau).
\end{array}
\end{equation}
Following an argument in the proof of Proposition~3 in
\cite{Sen:telescoping}, which in turn originated from a method in
\cite{Smooth_decoupling}, and using Equation~\ref{eq:decoupling3} and the
fact that $\hcT^{\hA_1 \hA_2 \rightarrow E}$ can increase the trace
by at most a factor of $4$, we get
\begin{equation}
\label{eq:decoupling4}
\begin{array}{rcl}
\lefteqn{
\E_{U^{\hA_1}, U^{\hA_2}}[
\|\sigma(U^{\hA_1}, U^{\hA_2})^{ER} -
  \hsigma(U^{\hA_1}, U^{\hA_2})^{ER}\|_1]
} \\
& \leq &
\E_{U^{\hA_1}, U^{\hA_2}}[
\|(\cT^{\hA_1 \hA_2 \rightarrow E} \otimes \I^R)
  ((U^{\hA_1} \otimes U^{\hA_2} \otimes \one^R) \circ 
   \rho^{\hA_1 \hA_2 R}) \\
&  &
~~~~~~~~~~~~~~~~~~
{} -
(\hcT^{\hA_1 \hA_2 \rightarrow E} \otimes \I^R)
((U^{\hA_1} \otimes U^{\hA_2} \otimes \one^R) \circ 
 \hrho^{\hA_1 \hA_2 R})\|_1] \\
& \leq &
\E_{U^{\hA_1}, U^{\hA_2}}[
\|(\cT^{\hA_1 \hA_2 \rightarrow E} \otimes \I^R)
  ((U^{\hA_1} \otimes U^{\hA_2} \otimes \one^R) \circ 
   \rho^{\hA_1 \hA_2 R}) \\
&  &
~~~~~~~~~~~~~~~~~~
{} -
(\cT^{\hA_1 \hA_2 \rightarrow E} \otimes \I^R)
((U^{\hA_1} \otimes U^{\hA_2} \otimes \one^R) \circ 
 \hrho^{\hA_1 \hA_2 R})\|_1] \\
&  &
{} +
\E_{U^{\hA_1}, U^{\hA_2}}[
\|(\cT^{\hA_1 \hA_2 \rightarrow E} \otimes \I^R)
  ((U^{\hA_1} \otimes U^{\hA_2} \otimes \one^R) \circ 
   \hrho^{\hA_1 \hA_2 R}) \\
&  &
~~~~~~~~~~~~~~~~~~~~~~
{} -
(\hcT^{\hA_1 \hA_2 \rightarrow E} \otimes \I^R)
((U^{\hA_1} \otimes U^{\hA_2} \otimes \one^R) \circ 
 \hrho^{\hA_1 \hA_2 R})\|_1] \\
& \leq &
4 \|\rho^{\hA_1 \hA_2 R} - \hrho^{\hA_1 \hA_2 R}\|_1 +
\|\tau^{\hA'_1 \hA'_2 E} - \htau^{\hA'_1 \hA'_2 E}\|_1
\leq
20\sqrt{2} \epsilon^{1/8}, \\
\mbox{Similarly,}~~~~
\lefteqn{
\E_{U^{\hA_1}, U^{\hA_2}}[
\|\hhsigma(U^{\hA_1}, U^{\hA_2})^{ER} -
  \hsigma(U^{\hA_1}, U^{\hA_2})^{ER}\|_1]
} \\
& \leq &
4 \|\hhrho^{\hA_1 \hA_2 R} - \hrho^{\hA_1 \hA_2 R}\|_1 +
\|\hhtau^{\hA'_1 \hA'_2 E} - \htau^{\hA'_1 \hA'_2 E}\|_1
\leq
140 \epsilon^{1/8}.
\end{array}
\end{equation}

Define the set
$
\Good := 
\{
(U^{\hA_1}, U^{\hA_2}):
\|\hhsigma(U^{\hA_1}, U^{\hA_2})^{ER} - 
  \hsigma(U^{\hA_1}, U^{\hA_2})^{ER}]\|_1
< 140 \epsilon^{1/16}
\}.
$
Applying Markov's inequality to the last inequality in 
Equation~\ref{eq:decoupling4}, we get
\[
\Pr_{U^{\hA_1}, U^{\hA_2}}[\Good] > 
1 - \epsilon^{1/16}.
\]
Applying Lemma~\ref{lem:asyml2}, we get
\begin{equation}
\label{eq:decoupling5}
\begin{array}{c}
\forall (U^{\hA_1}, U^{\hA_2}) \in \Good: 
\exists (\sigma'(U^{\hA_1}, U^{\hA_2}))^{ER}: \\
\|(\sigma'(U^{\hA_1}, U^{\hA_2}))^{ER} -
  \hsigma(U^{\hA_1}, U^{\hA_2})^{ER}\|_1 
< 12 \epsilon^{1/32}, \\
(\sigma'(U^{\hA_1}, U^{\hA_2}))^{ER}
\leq \hsigma(U^{\hA_1}, U^{\hA_2})^{ER}, \\
\|(\sigma'(U^{\hA_1}, U^{\hA_2}))^{ER}\|_2^2
\leq
(1 + 24\epsilon^{1/32})
\Tr [\hsigma(U^{\hA_1}, U^{\hA_2})^{ER} 
     \hhsigma(U^{\hA_1}, U^{\hA_2})^{ER}], \\
\forall (U^{\hA_1}, U^{\hA_2}) \not \in \Good: 
(\sigma'(U^{\hA_1}, U^{\hA_2}))^{ER} := \zero^{ER}, \\
\implies \forall (U^{\hA_1}, U^{\hA_2}):
(\sigma'(U^{\hA_1}, U^{\hA_2}))^{ER} \leq
\hsigma(U^{\hA_1}, U^{\hA_2})^{ER}, \\
\|(\sigma'((U^{\hA_1}, U^{\hA_2})))^{ER}\|_2^2 \leq
(1 + 24\epsilon^{1/32})
\Tr [\hsigma((U^{\hA_1}, U^{\hA_2}))^{ER} 
     \hhsigma((U^{\hA_1}, U^{\hA_2}))^{ER}].
\end{array}
\end{equation}
Using Equations~\ref{eq:decoupling4} and \ref{eq:decoupling5} and the
property that $\cT^{\hA_1 \hA_2 \rightarrow E}$ increases the trace
by at most a factor of $4$, we get
\begin{eqnarray*}
\lefteqn{
\E_{U^{\hA_1}, U^{\hA_2}}[
\|\sigma(U^{\hA_1}, U^{\hA_2})^{ER} -
  (\sigma'(U^{\hA_1}, U^{\hA_2}))^{ER}\|_1]
} \\
& = &
\E_{(U^{\hA_1}, U^{\hA_2}) \in \Good}[
\|\sigma(U^{\hA_1}, U^{\hA_2})^{ER} -
  (\sigma'(U^{\hA_1}, U^{\hA_2}))^{ER}\|_1] \\
& &
{} +
\E_{(U^{\hA_1}, U^{\hA_2}) \not \in \Good}[
\|\sigma(U^{\hA_1}, U^{\hA_2})^{ER}\|_1] \\
& \leq &
\E_{(U^{\hA_1}, U^{\hA_2}) \in \Good}[
\|\sigma(U^{\hA_1}, U^{\hA_2})^{ER} -
  \hsigma(U^{\hA_1}, U^{\hA_2})^{ER}\|_1] \\
&  &
{} +
\E_{(U^{\hA_1}, U^{\hA_2}) \in \Good}[
\|\hsigma(U^{\hA_1}, U^{\hA_2})^{ER} -
  (\sigma'(U^{\hA_1}, U^{\hA_2}))^{ER}\|_1] +
4 (1 - \Pr_{U^{\hA_1}, U^{\hA_2}}[\Good]) \\
& \leq &
\E_{U^{\hA_1}, U^{\hA_2}}[
\|\sigma(U^{\hA_1}, U^{\hA_2})^{ER} -
  \hsigma(U^{\hA_1}, U^{\hA_2})^{ER}\|_1] +
12\epsilon^{1/32} + 4\epsilon^{1/16} \\
& \leq &
20\sqrt{2} \epsilon^{1/8} + 12\epsilon^{1/32} + 4\epsilon^{1/16} 
<
45 \epsilon^{1/32}.
\end{eqnarray*}
Thus in order to prove the present proposition, from 
Equation~\ref{eq:decouplingstart}, it suffices to show that
\begin{equation}
\label{eq:decoupling6}
\E_{U^{\hA_1}, U^{\hA_2}}[
\|(\sigma'(U^{\hA_1}, U^{\hA_2}))^{ER} - \tau^E \otimes \rho^R\|_1]
< 36 \epsilon^{1/64}.
\end{equation}

Define 
$\Pi^{ER} := \supp(\tau^E \otimes \rho^R) = F'_\tau \otimes F'_\rho$.
Observe that for all $(U^{\hA_1}, U^{\hA_2})$,
\[
\supp((\sigma'(U^{\hA_1}, U^{\hA_2}))^{ER}) \leq
\supp(\hsigma(U^{\hA_1}, U^{\hA_2})^{ER}) \leq \Pi^{ER},
~
\supp(\sigma(U^{\hA_1}, U^{\hA_2})^{ER}) \leq \Pi^{ER},
\]
where the first operator inequality follows from 
Equation~\ref{eq:decoupling5} and the second operator inequality
follows from the fact that $\hcT^{\hA_1 \hA_2 \rightarrow E}$ is
the completely positive superoperator that projects from 
$\hA_1 \otimes \hA_2$ onto the subspace
$
T_{\epsilon^{1/4}}^{A_1 \rightarrow \hA_1} \otimes
T_{\epsilon^{1/4}}^{A_2 \rightarrow \hA_2},
$
followed by rotating the subspace onto 
$(A_1\otimes\ketbra{0}^{\C^2}) \otimes (A_2\otimes\ketbra{0}^{\C^2})$,
followed by applying $4 \cT^{A_1 A_2 \rightarrow E}$.
By Proposition~\ref{prop:shavedCauchySchwarz},
Equation~\ref{eq:decoupling6} and the convexity of the square function, 
it suffices to show
\begin{equation}
\label{eq:decoupling7}
\begin{array}{c}
\sqrt{\Tr \Pi^{ER}} \cdot \E_{U^{\hA_1}, U^{\hA_2}}[
\|(\sigma'(U^{\hA_1}, U^{\hA_2}))^{ER} - \tau^E \otimes \rho^R\|_1]
< 36 \epsilon^{1/64}, \\
\implies ~ \mbox{it suffices to show} ~~
\E_{U^{\hA_1}, U^{\hA_2}}[
\|(\sigma'(U^{\hA_1}, U^{\hA_2}))^{ER} - \tau^E \otimes \rho^R\|_2^2]
< \frac{1296 \epsilon^{1/32}}{F'_\tau F'_\rho}.
\end{array}
\end{equation}
Using Equations~\ref{eq:decoupling3}, \ref{eq:decoupling5}, the
method of the proof of Proposition~3 of \cite{Sen:telescoping} and the
fact that $\hcT^{\hA_1 \hA_2 \rightarrow E}$ is a completely positive
superoperator that increases the trace by at most a factor of $4$, we have,
\begin{eqnarray*}
\lefteqn{
\E_{U^{\hA_1}, U^{\hA_2}}[
\|(\sigma'(U^{\hA_1}, U^{\hA_2}))^{ER} - \tau^E \otimes \rho^R\|_2^2]
} \\
& = &
\E_{U^{\hA_1}, U^{\hA_2}}[
\|(\sigma'(U^{\hA_1}, U^{\hA_2}))^{ER}\|_2^2] +
\|\tau^E \otimes \rho^R\|_2^2 -
2 \E_{U^{\hA_1}, U^{\hA_2}}[
\Tr [(\sigma'(U^{\hA_1}, U^{\hA_2}))^{ER} (\tau^E \otimes \rho^R)]] \\
& \leq &
\E_{U^{\hA_1}, U^{\hA_2}}[
\|(\sigma'(U^{\hA_1}, U^{\hA_2}))^{ER}\|_2^2] +
\frac{(1+\sqrt{\epsilon})^2 (\Tr \tau)^2}{|F'_\tau|^2 |F'_\rho|^2}
\Tr [\one^{F'_\tau F'_\rho}] \\
&  &
{} -
\frac{2(1-\sqrt{\epsilon}) (\Tr \tau)}{|F'_\tau| |F'_\rho|}
\E_{U^{\hA_1}, U^{\hA_2}}[
\Tr [(\sigma'(U^{\hA_1}, U^{\hA_2}))^{ER} \Pi^{ER}]] \\
& \leq &
\E_{U^{\hA_1}, U^{\hA_2}}[
\|(\sigma'(U^{\hA_1}, U^{\hA_2}))^{ER}\|_2^2] +
\frac{(1+\sqrt{\epsilon})^2 (\Tr \tau)^2}{|F'_\tau| |F'_\rho|} \\
&  &
{} -
\frac{2(1-\sqrt{\epsilon}) (\Tr \tau)}{|F'_\tau| |F'_\rho|}
\E_{(U^{\hA_1}, U^{\hA_2}) \in \Good}[
\Tr [(\sigma'(U^{\hA_1}, U^{\hA_2}))^{ER} \Pi^{ER}]] \\
& \leq &
\E_{U^{\hA_1}, U^{\hA_2}}[
\|(\sigma'(U^{\hA_1}, U^{\hA_2}))^{ER}\|_2^2] +
\frac{(1+\sqrt{\epsilon})^2 (\Tr \tau)^2}{|F'_\tau| |F'_\rho|} \\
&  &
{} -
\frac{2(1-\sqrt{\epsilon}) (\Tr \tau)}{|F'_\tau| |F'_\rho|}
\E_{(U^{\hA_1}, U^{\hA_2}) \in \Good}[
\Tr [\hsigma(U^{\hA_1}, U^{\hA_2})^{ER} \Pi^{ER}]] \\
&  &
{} +
\frac{2(1-\sqrt{\epsilon}) (\Tr \tau)}{|F'_\tau| |F'_\rho|}
\E_{(U^{\hA_1}, U^{\hA_2}) \in \Good}[
\|\hsigma(U^{\hA_1}, U^{\hA_2})^{ER} -
  (\sigma'(U^{\hA_1}, U^{\hA_2}))^{ER}\|_1] \\
& \leq &
\E_{U^{\hA_1}, U^{\hA_2}}[
\|(\sigma'(U^{\hA_1}, U^{\hA_2}))^{ER}\|_2^2] +
\frac{(1+\sqrt{\epsilon})^2 (\Tr \tau)^2}{|F'_\tau| |F'_\rho|} \\
&  &
{} -
\frac{2(1-\sqrt{\epsilon}) (\Tr \tau)}{|F'_\tau| |F'_\rho|}
\E_{U^{\hA_1}, U^{\hA_2}}[
\Tr [\hsigma(U^{\hA_1}, U^{\hA_2})^{ER} \Pi^{ER}]] \\
&  &
{} +
\frac{8(1-\sqrt{\epsilon}) (\Tr \tau) 
      (1 - \Pr_{U^{\hA_1}, U^{\hA_2}}[\Good]) 
     }{|F'_\tau| |F'_\rho|} +
\frac{24\epsilon^{1/32} (1-\sqrt{\epsilon}) (\Tr \tau)}
     {|F'_\tau| |F'_\rho|} \\
& \leq &
\E_{U^{\hA_1}, U^{\hA_2}}[
\|(\sigma'(U^{\hA_1}, U^{\hA_2}))^{ER}\|_2^2] +
\frac{(1+\sqrt{\epsilon})^2 (\Tr \tau)^2}{|F'_\tau| |F'_\rho|} +
\frac{32\epsilon^{1/32}}{|F'_\tau| |F'_\rho|} \\
&  &
{} -
\frac{2(1-\sqrt{\epsilon}) (\Tr \tau)}{|F'_\tau| |F'_\rho|}
\Tr [(\tau^E \otimes \tau^R) \Pi^{ER}] \\
& =  &
\frac{32\epsilon^{1/32}}{|F'_\tau| |F'_\rho|} +
\E_{U^{\hA_1}, U^{\hA_2}}[
\|(\sigma'(U^{\hA_1}, U^{\hA_2}))^{ER}\|_2^2] +
\frac{(1+\sqrt{\epsilon})^2 (\Tr \tau)^2}{|F'_\tau| |F'_\rho|} -
\frac{2(1-\sqrt{\epsilon}) (\Tr \tau)^2}{|F'_\tau| |F'_\rho|} \\
& < &
\frac{32\epsilon^{1/32}}{|F'_\tau| |F'_\rho|} +
\E_{U^{\hA_1}, U^{\hA_2}}[
\|(\sigma'(U^{\hA_1}, U^{\hA_2}))^{ER}\|_2^2] -
\frac{(1-5\sqrt{\epsilon}) (\Tr \tau)^2}{|F'_\tau| |F'_\rho|}.
\end{eqnarray*}
Thus from Equations~\ref{eq:decoupling7}, \ref{eq:decoupling5}, it 
suffices to show that
\begin{equation}
\label{eq:decoupling8}
\begin{array}{c}
\E_{U^{\hA_1}, U^{\hA_2}}[
\|(\sigma'(U^{\hA_1}, U^{\hA_2}))^{ER}\|_2^2] \leq
\frac{(1+1259\epsilon^{1/32}) (\Tr \tau)^2}{|F'_\tau| |F'_\rho|}, \\
\implies ~
\mbox{suffices to show } ~~
\E_{U^{\hA_1}, U^{\hA_2}}[
\Tr [\hsigma(U^{\hA_1}, U^{\hA_2})^{ER}
     \hhsigma(U^{\hA_1}, U^{\hA_2})^{ER}]] \leq
\frac{(1+30\epsilon^{1/32}) (\Tr \tau)^2}{|F'_\tau| |F'_\rho|}.
\end{array}
\end{equation}

Following the arguments in \cite{Chakraborty:simultaneous} involving
the use of the {\em swap trick} and the {\em twirling operator} over
the Haar measure, we get
\begin{equation}
\label{eq:decoupling9}
\begin{array}{rcl}
\lefteqn{
\E_{U^{\hA_1}, U^{\hA_2}}[
\Tr [\hsigma(U^{\hA_1}, U^{\hA_2})^{ER}
     \hhsigma(U^{\hA_1}, U^{\hA_2})^{ER}]]
} \\
& = &
\alpha_0 \Tr [\hrho^{R} \hhrho^{R}] +
\alpha_1 \Tr [\hrho^{\hA_1 R} \hhrho^{\hA_1 R}] +
\alpha_2 \Tr [\hrho^{\hA_2 R} \hhrho^{\hA_2 R}] +
\alpha_{12} \Tr [\hrho^{\hA_1 \hA_2 R} \hhrho^{\hA_1 \hA_2 R}],
\end{array}
\end{equation}
where the 4-tuple $(\alpha_0, \alpha_1, \alpha_2, \alpha_{12})$ is 
defined to be the solution of the following linear equation:
\[
\begin{array}{rcl}
\left[
\begin{array}{c}
\alpha_{0} \\ 
\alpha_{1} \\ 
\alpha_{2} \\ 
\alpha_{12}
\end{array} 
\right]
& = &
\frac{|A_1| |A_2|}{(|A_1|^2-1)(|A_2|^2-1)} \cdot {} \\
&  &
~~~
\left[
\begin{array}{c c c c}
|A_1||A_2| & -|A_2| & -|A_1| & 1 \\
-|A_2| & |A_1||A_2| & 1 & -|A_1| \\
-|A_1| & 1 & |A_1||A_2| & -|A_2| \\
1 & -|A_1| & -|A_2| & |A_1||A_2|
\end{array}
\right]
\left[
\begin{array}{c}
\Tr [\htau^{E} \hhtau^{E}] \\
\Tr [\htau^{\hA'_1 E} \hhtau^{\hA'_1 E}] \\
\Tr [\htau^{\hA'_2 E} \hhtau^{\hA'_2 E}] \\
\Tr [\htau^{\hA'_1 \hA'_2 E} \hhtau^{\hA'_1 \hA'_2 E}] 
\end{array}
\right].
\end{array}
\]
We bound the terms $\alpha_0$, $\alpha_1$, $\alpha_2$, $\alpha_{12}$
as follows, using Fact~\ref{fact:dupuisoperatorineq}.
\begin{equation}
\label{eq:decoupling10}
\begin{array}{rcl}
\alpha_0
&   =  &
\frac{|A_1| |A_2|}{(|A_1|^2-1)(|A_2|^2-1)} \cdot 
(|A_1||A_2| \cdot \Tr [\htau^{E} \hhtau^{E}] -
 |A_2| \cdot \Tr [\htau^{\hA'_1 E} \hhtau^{\hA'_1 E}] \\
&  &
~~~~~~~~~~~~~~~~~~~~~~~~~~~~~
{} -
 |A_1| \cdot \Tr [\htau^{\hA'_2 E} \hhtau^{\hA'_2 E}] +
 \Tr [\htau^{\hA'_1 \hA'_2 E} \hhtau^{\hA'_1 \hA'_2 E}]
) \\
& \leq &
\frac{|A_1| |A_2|}{(|A_1|^2-1)(|A_2|^2-1)} \cdot 
(|A_1||A_2| \cdot \Tr [\htau^{E} \hhtau^{E}] \\
&  &
~~~~~~~~~~~~~~~~~~~~~~~~~~~~~
{} -
 |A_1| \cdot \Tr [\htau^{\hA'_2 E} \hhtau^{\hA'_2 E}] +
 |A_1| \Tr [(\one^{\hA'_1} \otimes \htau^{\hA'_2 E}) 
	    \hhtau^{\hA'_1 \hA'_2 E}]
) \\
&   =  &
\frac{|A_1| |A_2|}{(|A_1|^2-1)(|A_2|^2-1)} \cdot 
(|A_1||A_2| \cdot \Tr [\htau^{E} \hhtau^{E}] \\
&  &
~~~~~~~~~~~~~~~~~~~~~~~~~~~~~
{} -
 |A_1| \cdot \Tr [\htau^{\hA'_2 E} \hhtau^{\hA'_2 E}] +
 |A_1| \Tr [\htau^{\hA'_2 E} \hhtau^{\hA'_2 E}]
) \\
&   =  &
\frac{|A_1|^2 |A_2|^2}{(|A_1|^2-1)(|A_2|^2-1)} \cdot 
\Tr [\htau^{E} \hhtau^{E}], \\
\alpha_1
&   =  &
\frac{|A_1| |A_2|}{(|A_1|^2-1)(|A_2|^2-1)} \cdot 
(-|A_2| \cdot \Tr [\htau^{E} \hhtau^{E}] +
 |A_1||A_2| \cdot \Tr [\htau^{\hA'_1 E} \hhtau^{\hA'_1 E}] \\
&  &
~~~~~~~~~~~~~~~~~~~~~~~~~~~~~
{} +
 \Tr [\htau^{\hA'_2 E} \hhtau^{\hA'_2 E}] -
 |A_1| \cdot \Tr [\htau^{\hA'_1 \hA'_2 E} \hhtau^{\hA'_1 \hA'_2 E}]
) \\
& \leq &
\frac{|A_1| |A_2|}{(|A_1|^2-1)(|A_2|^2-1)} \cdot 
(-|A_2| \cdot \Tr [\htau^{E} \hhtau^{E}] +
 |A_1||A_2| \cdot \Tr [\htau^{\hA'_1 E} \hhtau^{\hA'_1 E}] \\
&  &
~~~~~~~~~~~~~~~~~~~~~~~~~~~~~
{} +
 |A_2| \cdot \Tr [(\one^{\hA'_2} \otimes \htau^{E}) \hhtau^{\hA'_2 E}] 
) \\
&   =  &
\frac{|A_1| |A_2|}{(|A_1|^2-1)(|A_2|^2-1)} \cdot 
(-|A_2| \cdot \Tr [\htau^{E} \hhtau^{E}] +
 |A_1||A_2| \cdot \Tr [\htau^{\hA'_1 E} \hhtau^{\hA'_1 E}] \\
&  &
~~~~~~~~~~~~~~~~~~~~~~~~~~~~~
{} +
 |A_2| \cdot \Tr [\htau^{E} \hhtau^{E}] +
) \\
&   =  &
\frac{|A_1|^2 |A_2|^2}{(|A_1|^2-1)(|A_2|^2-1)} \cdot 
\Tr [\htau^{\hA'_1 E} \hhtau^{\hA'_1 E}], \\
\alpha_2
&   \leq  &
\frac{|A_1|^2 |A_2|^2}{(|A_1|^2-1)(|A_2|^2-1)} \cdot 
\Tr [\htau^{\hA'_2 E} \hhtau^{\hA'_2 E}] 
~~
\mbox{similarly}, \\
\alpha_{12}
&   =  &
\frac{|A_1| |A_2|}{(|A_1|^2-1)(|A_2|^2-1)} \cdot 
(\Tr [\htau^{E} \hhtau^{E}] -
 |A_1| \cdot \Tr [\htau^{\hA'_1 E} \hhtau^{\hA'_1 E}] \\
&  &
~~~~~~~~~~~~~~~~~~~~~~~~~~~~~
{} -
 |A_2| \cdot \Tr [\htau^{\hA'_2 E} \hhtau^{\hA'_2 E}] +
 |A_1||A_2| \cdot \Tr [\htau^{\hA'_1 \hA'_2 E} \hhtau^{\hA'_1 \hA'_2 E}]
) \\
& \leq &
\frac{|A_1| |A_2|}{(|A_1|^2-1)(|A_2|^2-1)} \cdot 
(\Tr [\htau^{E} \hhtau^{E}] +
 |A_1||A_2| \cdot \Tr [\htau^{\hA'_1 \hA'_2 E} \hhtau^{\hA'_1 \hA'_2 E}]
).
\end{array}
\end{equation}
From Equations~\ref{eq:decoupling9}, \ref{eq:decoupling10} we get,
\begin{equation}
\label{eq:decoupling11}
\begin{array}{rcl}
\lefteqn{
\E_{U^{\hA_1}, U^{\hA_2}}[
\Tr [\hsigma(U^{\hA_1}, U^{\hA_2})^{ER}
     \hhsigma(U^{\hA_1}, U^{\hA_2})^{ER}]]
} \\
& \leq &
\frac{|A_1|^2 |A_2|^2}{(|A_1|^2-1)(|A_2|^2-1)} \cdot 
\Tr [\htau^{E} \hhtau^{E}] \cdot \Tr [\hrho^{R} \hhrho^{R}] \\
&  &
{} +
\frac{|A_1|^2 |A_2|^2}{(|A_1|^2-1)(|A_2|^2-1)} \cdot 
\Tr [\htau^{\hA'_2 E} \hhtau^{\hA'_2 E}] 
\cdot \Tr [\hrho^{\hA_1 R} \hhrho^{\hA_1 R}] \\
&  &
{} +
\frac{|A_1|^2 |A_2|^2}{(|A_1|^2-1)(|A_2|^2-1)} \cdot 
\Tr [\htau^{\hA'_2 E} \hhtau^{\hA'_2 E}] 
\cdot \Tr [\hrho^{\hA_2 R} \hhrho^{\hA_2 R}] \\
& &
{} +
\frac{|A_1| |A_2|}{(|A_1|^2-1)(|A_2|^2-1)} \cdot 
(\Tr [\htau^{E} \hhtau^{E}] +
 |A_1||A_2| \cdot \Tr [\htau^{\hA'_1 \hA'_2 E} \hhtau^{\hA'_1 \hA'_2 E}]
) \cdot \Tr [\hrho^{\hA_1 \hA_2 R} \hhrho^{\hA_1 \hA_2 R}].
\end{array}
\end{equation}

As in Section~\ref{sec:newconvexsplit}, we can prove the following
lemma.
\begin{lemma}
\label{lem:decoupling1}
\begin{eqnarray*}
\Tr [\htau^{E} \hhtau^{E}] 
& \leq &
\frac{(1+\sqrt{\epsilon}) (\Tr \tau)^2}{|F'_\tau|},
~~
\Tr [\hrho^{R} \hhrho^{R}] 
\leq
\frac{1+\sqrt{\epsilon}}{|F'_\rho|}, \\
\Tr [\htau^{\hA'_1 E} \hhtau^{\hA'_1 E}] 
& \leq &
\frac{(1+3\sqrt{\epsilon}) (\Tr \tau)^2 \cdot
      2^{-\Hmin^\epsilon(A'_1|E)_\tau}}{|F'_\tau|}, \\
\Tr [\hrho^{\hA_1 R} \hhrho^{\hA_1 R}] 
& \leq &
\frac{(1+3\sqrt{\epsilon}) \cdot
      2^{-\Hmin^\epsilon(A_1|R)_\rho}}{|F'_\rho|}, \\
\Tr [\htau^{\hA'_2 E} \hhtau^{\hA'_2 E}] 
& \leq &
\frac{(1+3\sqrt{\epsilon}) (\Tr \tau)^2 \cdot
      2^{-\Hmin^\epsilon(A'_2|E)_\tau}}{|F'_\tau|}, \\
\Tr [\hrho^{\hA_2 R} \hhrho^{\hA_2 R}] 
& \leq &
\frac{(1+3\sqrt{\epsilon}) \cdot
      2^{-\Hmin^\epsilon(A_2|R)_\rho}}{|F'_\rho|}, \\
\Tr [\htau^{\hA'_1 \hA'_2 E} \hhtau^{\hA'_1 \hA'_2 E}] 
& \leq &
\frac{(1+3\sqrt{\epsilon}) (\Tr \tau)^2 \cdot
      2^{-\Hmin^\epsilon(A'_1 A'_2|E)_\tau}}{|F'_\tau|}, \\
\Tr [\hrho^{\hA_1 \hA_2 R} \hhrho^{\hA_1 \hA_2 R}] 
& \leq &
\frac{(1+3\sqrt{\epsilon}) \cdot
      2^{-\Hmin^\epsilon(A_1 A_2|R)_\rho}}{|F'_\rho|}.
\end{eqnarray*}
\end{lemma}
Plugging Lemma~\ref{lem:decoupling1} into 
Equation~\ref{eq:decoupling11} and using the constraints on the
entropic terms and Hilbert space dimensions given in the assumption
of the theorem, we get
\begin{eqnarray*}
\lefteqn{
\E_{U^{\hA_1}, U^{\hA_2}}[
\Tr [\hsigma(U^{\hA_1}, U^{\hA_2})^{ER}
     \hhsigma(U^{\hA_1}, U^{\hA_2})^{ER}]]
} \\
& \leq &
\frac{|A_1|^2 |A_2|^2}{(|A_1|^2-1)(|A_2|^2-1)} \cdot 
\frac{(1+\sqrt{\epsilon})^2 (\Tr \tau)^2}{|F'_\tau||F'_\rho|} \\
&  &
{} +
\frac{|A_1|^2 |A_2|^2}{(|A_1|^2-1)(|A_2|^2-1)} \cdot 
\frac{(1+3\sqrt{\epsilon})^2 (\Tr \tau)^2 \cdot
      2^{-\Hmin^\epsilon(A'_1|E)_\tau
	 -\Hmin^\epsilon(A_1|R)_\rho}}{|F'_\tau||F'_\rho|} \\
&  &
{} +
\frac{|A_1|^2 |A_2|^2}{(|A_1|^2-1)(|A_2|^2-1)} \cdot 
\frac{(1+3\sqrt{\epsilon})^2 (\Tr \tau)^2 \cdot
      2^{-\Hmin^\epsilon(A'_2|E)_\tau
	 -\Hmin^\epsilon(A_2|R)_\rho}}{|F'_\tau||F'_\rho|} \\
& &
{} +
\frac{|A_1| |A_2|}{(|A_1|^2-1)(|A_2|^2-1)} \cdot 
\frac{(1+\sqrt{\epsilon})(1+3\sqrt{\epsilon}) (\Tr \tau)^2 \cdot
      2^{-\Hmin^\epsilon(A_1 A_2|R)_\rho}}{|F'_\tau||F'_\rho|} \\
& &
{} +
\frac{|A_1|^2 |A_2|^2}{(|A_1|^2-1)(|A_2|^2-1)} \cdot 
\frac{(1+3\sqrt{\epsilon})^2 (\Tr \tau)^2 \cdot
      2^{-\Hmin^\epsilon(A'_1 A'_2|E)_\tau
	 -\Hmin^\epsilon(A_1 A_2|R)_\rho}}{|F'_\tau||F'_\rho|} \\
& \leq &
\frac{|A_1|^2 |A_2|^2}{(|A_1|^2-1)(|A_2|^2-1)} \cdot 
\frac{(1+3\sqrt{\epsilon}) (\Tr \tau)^2}{|F'_\tau||F'_\rho|} \\
&  &
{} +
\frac{|A_1|^2 |A_2|^2}{(|A_1|^2-1)(|A_2|^2-1)} \cdot 
\frac{(1+9\sqrt{\epsilon}) (\Tr \tau)^2 \cdot
      2^{-\Hmin^\epsilon(A'_1|E)_\tau
	 -\Hmin^\epsilon(A_1|R)_\rho}}{|F'_\tau||F'_\rho|} \\
&  &
{} +
\frac{|A_1|^2 |A_2|^2}{(|A_1|^2-1)(|A_2|^2-1)} \cdot 
\frac{(1+9\sqrt{\epsilon}) (\Tr \tau)^2 \cdot
      2^{-\Hmin^\epsilon(A'_2|E)_\tau
	 -\Hmin^\epsilon(A_2|R)_\rho}}{|F'_\tau||F'_\rho|} \\
& &
{} +
\frac{|A_1|^2 |A_2|^2}{(|A_1|^2-1)(|A_2|^2-1)} \cdot 
\frac{(1+7\sqrt{\epsilon}) (\Tr \tau)^2 \cdot
      2^{-\Hmin^\epsilon(A'_1 A'_2|E)_\tau
	 -\Hmin^\epsilon(A_1 A_2|R)_\rho}}{|F'_\tau||F'_\rho|} \\
& &
{} +
\frac{|A_1|^2 |A_2|^2}{(|A_1|^2-1)(|A_2|^2-1)} \cdot 
\frac{(1+9\sqrt{\epsilon}) (\Tr \tau)^2 \cdot
      2^{-\Hmin^\epsilon(A'_1 A'_2|E)_\tau
	 -\Hmin^\epsilon(A_1 A_2|R)_\rho}}{|F'_\tau||F'_\rho|} \\
& \leq &
(1 - \epsilon)^{-2} \cdot
\frac{(1+3\sqrt{\epsilon}) (\Tr \tau)^2}{|F'_\tau||F'_\rho|} +
(1 - \epsilon)^{-2} \cdot
\frac{\epsilon (1+9\sqrt{\epsilon}) (\Tr \tau)^2}{|F'_\tau||F'_\rho|} \\
&  &
{} +
(1 - \epsilon)^{-2} \cdot
\frac{\epsilon (1+9\sqrt{\epsilon}) (\Tr \tau)^2}{|F'_\tau||F'_\rho|} +
(1 - \epsilon)^{-2} \cdot
\frac{\epsilon (2+16\sqrt{\epsilon}) (\Tr \tau)^2}{|F'_\tau||F'_\rho|} \\
& \leq &
\frac{(1+9\sqrt{\epsilon}) (\Tr \tau)^2}{|F'_\tau||F'_\rho|} +
\frac{14\epsilon (\Tr \tau)^2}{|F'_\tau||F'_\rho|} 
<
\frac{(1+23\sqrt{\epsilon}) (\Tr \tau)^2}{|F'_\tau||F'_\rho|},
\end{eqnarray*}
where we used Fact~\ref{fact:CondHminUpperBound}
in the second inequality above.
Combined with Equation~\ref{eq:decoupling8}, Theorem~\ref{thm:decoupling}
is finally proved.

\section{Conclusion}
\label{sec:conclusion}
In this paper, we have noted how the telescoping technique to prove
multipartite soft covering results fails 
when pairwise 
independence is lost amongst certain probability distributions involved
in soft covering. We have then developed a sophisticated machinery
that recovers the desired smooth inner bounds for multipartite soft 
covering under
a slight loss of pairwise independence described precisely in
Definition~\ref{def:nonpairwise}. To develop this machinery,
we had to combine existing results on tilting and augmentation smoothing
of quantum states from \cite{sen:oneshot} together with a novel
observation that the naive classical strategy for flattening probability
distributions in fact preserves the fidelity of quantum states. We have
then seen how this machinery allows us to prove a new inner bound
for sending private classical information over a wiretap quantum 
interference channel by combining existing inner bounds in the
non wiretap case from \cite{sen:simultaneous} together with our
smooth soft covering results. Such an inner bound in the non pairwise
independent case was unknown earlier even in the classical asymptotic
iid setting.

Our work leads to natural avenues for further research. 
One direction to pursue is to further
extend our tilting, augmentation smoothing and flattening based 
machinery, and find more applications to quantum Shannon theory. 
In spirit, our machinery is arguably much closer to simultaneous
smoothing than telescoping.
It gives
a unified treatment of one shot multiparty smooth packing type 
questions as in 
\cite{sen:oneshot, sen:simultaneous, ding:relay}, as well as one
one shot smooth multiparty
convex split and soft covering questions described in this paper. `
However, one important application where our machinery fails is in 
proving a fully smooth multipartite expander matrix Chernoff bound. 
The paper \cite{Sen:MatrixChernoff} proves a fully smooth multipartite
classical quantum soft covering lemma in concentration, aka fully
smooth multipartite matrix Chernoff bound, but for independent choices
of samples by each classical party. That paper also proves a smooth
unipartite soft covering lemma in concentration when the samples
are taken by the single classical party from a random walk on an 
expander graph, which can be called a unipartite expander matrix Chernoff
bound. However, the techniques of that paper 
fail to prove a multipartite expander matrix Chernoff bound.
Expander walks are a key example where pairwise independence amongst the
classical samples of a party is lost. Unfortunately the loss of pairwise
independence does not fit the framework behind 
Definition~\ref{def:nonpairwise}. Hence even fully smooth
multipartite classical quantum soft covering in expectation under
expander walks cannot be handled by the techniques of the present paper.
Proving fully smooth multipartite soft covering results in expectation
and concentration for expander walks is thus a very important open problem.

It is also important to find limitations of our tilting and
augmentation smoothing based machinery. One 
immediate limitation that comes to mind is that 
augmenting quantum states requires tensoring with fresh completely
mixed ancilla states. This is required
so that tracing out a set of registers essentially removes the tilts 
along the traced out directions when  quantified in terms of
the Schatten-$\ell_\infty$ norm. Augmentation can be done without
loss of generality for the packing type questions studied in earlier
works, as well as in convex split and soft covering problems. However,
high dimensional augmentation becomes a bottleneck for decoupling 
problems. This is why we gave the simpler proof for fully smooth
convex split
in Section~\ref{sec:newconvexsplit} which did not need augmentation for 
smoothing purposes.
The proof in Section~\ref{sec:newconvexsplit} only had to enlarge the 
dimensions of the 
Hilbert spaces $X$ and $Y$ by factor of $2$ in order to achieve the 
correct tilting. It led to the proof of 
a fully smooth decoupling theorem in Section~\ref{sec:decoupling}. 
Though the decoupling theorem in
Section~\ref{sec:decoupling} suffices for all the applications of 
decoupling in quantum
information theory that we know of, the factor of $2$ blowup is still
mathematically unsatisfactory. In contrast, the telescoping based proof
of fully smooth decoupling in \cite{Sen:telescoping} has no such blowup.
Removing this factor of $2$ 
for decoupling for our machinery is an interesting open problem.

\section*{Acknowledgements}
I sincerely thank Rahul Jain for many stimulating discussions  
on smooth convex split and flattening quantum states. Many
ideas in this work were slowly developed over several visits hosted
by him during the last few years at the Centre for Quantum Technologies, 
National University of Singapore. I also thank the staff at the Centre
for their efficient hospitality. This research is 
supported in part by the National Research Foundation, Singapore and 
A*STAR under the CQT Bridging Grant and the Quantum Engineering Programme 
Award number NRF2021-QEP2-02-P05.
I acknowledge support of the Department 
of Atomic Energy, Government of India, under project no. RTI4001.

\bibliography{flatten}

\begin{thebibliography}{DGHW20}

\bibitem[ABJT19]{Jain:minimax}
A.~Anshu, M.~Berta, R.~Jain, and M.~Tomamichel.
\newblock A minimax approach to one-shot entropy inequalities.
\newblock {\em Journal of Mathematical Physics}, 60:122201:1--122201:9, 2019.

\bibitem[AJ22]{decoupling_convexsplit}
A.~Anshu and R.~Jain.
\newblock Efficient methods for one-shot quantum communication.
\newblock {\em npj Quantum Information}, 8:97:1--97:46, 2022.

\bibitem[AJW18]{anshu:slepianwolf}
A.~Anshu, R.~Jain, and N.~Warsi.
\newblock A generalized quantum {Slepian–Wolf}.
\newblock {\em IEEE Transactions on Information Theory}, 64:1436--1453, 2018.

\bibitem[CGB23]{Cheng:convexsplit}
H-C. Cheng, L.~Gao, and M.~Berta.
\newblock Quantum broadcast channel simulation via multipartite convex
  splitting.
\newblock arXiv:2304.12056, 2023.

\bibitem[CMGE08]{CMGElGamal}
{Chong, H.}, {Motani, M.}, {Garg, H.}, and {{El Gamal}, H.}
\newblock On the {Han-Kobayashi} region for the interference channel.
\newblock {\em IEEE Transactions on Information Theory}, 54:3188--3195, 2008.

\bibitem[CNS21]{Chakraborty:simultaneous}
S.~Chakraborty, A.~Nema, and P.~Sen.
\newblock A multi-sender decoupling theorem and simultaneous decoding for the
  quantum mac.
\newblock In {\em IEEE International Symposium on Information Theory (ISIT)},
  2021.
\newblock Full version in arXiv::2102.02187.

\bibitem[CW23]{Colomer:decoupling}
P.~Colomer and A.~Winter.
\newblock Decoupling by local random unitaries without simultaneous smoothing,
  and applications to multi-user quantum information tasks.
\newblock arXiv:2304.12114, 2023.

\bibitem[DF13]{drescher:simultaneous}
L.~Drescher and O.~Fawzi.
\newblock On simultaneous min-entropy smoothing.
\newblock In {\em IEEE International Symposium on Information Theory (ISIT)},
  pages 161--165, 2013.

\bibitem[DGHW20]{ding:relay}
{Ding, D.}, {Gharibyan, H.}, {Hayden, P.}, and {Walter, M.}
\newblock A quantum multiparty packing lemma and the relay channel.
\newblock {\em IEEE Transactions on Information Theory}, 66:3500--3519, 2020.

\bibitem[RSW17]{Radhakrishnan:wiretap}
J.~Radhakrishnan, P.~Sen, and N.~Warsi.
\newblock One-shot private classical capacity of quantum wiretap channel:
  {Based} on one-shot quantum covering lemma.
\newblock Proceedings of QCRYPT workshop. Also arXiv:1703.01932, 2017.

\bibitem[SDTR13]{Smooth_decoupling}
O.~Szehr, F.~Dupuis, M.~Tomamichel, and R.~Renner.
\newblock Decoupling with unitary approximate two-designs.
\newblock {\em New Journal of Physics}, 15:053022:1--053022:20, 2013.

\bibitem[Sen12]{sen:interference}
P.~Sen.
\newblock Achieving the {Han-Kobayashi} inner bound for the quantum
  interference channel.
\newblock In {\em IEEE International Symposium on Information Theory (ISIT)},
  pages 736--740, 2012.
\newblock Full version at arXiv:1109.0802.

\bibitem[Sen21a]{sen:simultaneous}
P.~Sen.
\newblock Inner bounds via simultaneous decoding in quantum network information
  theory.
\newblock {\em Sadhana}, 46:0018:1--0018:20, 2021.
\newblock Also available at arXiv:1806.07276.

\bibitem[Sen21b]{sen:oneshot}
P.~Sen.
\newblock Unions, intersections and a one shot quantum joint typicality lemma.
\newblock {\em Sadhana}, 46:0057:1--0057:44, 2021.
\newblock Also available at arXiv:1806.07278.

\bibitem[Sen24]{Sen:telescoping}
P.~Sen.
\newblock Fully smooth one shot multipartite covering and decoupling of quantum
  states via telescoping.
\newblock arXiv:2410.17893, 2024.

\bibitem[Sen25]{Sen:MatrixChernoff}
P.~Sen.
\newblock Matrix {Chernoff} concentration bounds for multipartite soft covering
  and expander walks.
\newblock arXiv:2504.04067, 2025.

\bibitem[Wil17]{Wilde:wiretap}
M.~Wilde.
\newblock Position-based coding and convex splitting for private communication
  over quantum channels.
\newblock {\em Quantum Information Processing}, 16:264:1--264:31, 2017.

\end{thebibliography}

\end{document}